\documentclass[manuscript]{acmart}

\acmYear{2025}\copyrightyear{2025}
\setcopyright{rightsretained}
\acmConference[ACM FAccT '25]{ACM Conference on Fairness, Accountability, and Transparency}{}{}
\acmBooktitle{Submitted for publication}
\acmDOI{xx.xxxx/xxxx}
\acmISBN{979-8-4007-0450-5/24/06}
\AtBeginDocument{%
  \providecommand\BibTeX{{%
    \normalfont B\kern-0.5em{\scshape i\kern-0.25em b}\kern-0.8em\TeX}}}
\usepackage[utf8]{inputenc}
\PassOptionsToPackage{numbers, sort}{natbib}
\usepackage{amsmath, bm, graphicx}

\usepackage{hyperref}
\usepackage{multirow}
\usepackage{graphicx} % Required for inserting images
\usepackage{amsthm}
\usepackage{algorithm}
\usepackage{algorithmic}
\usepackage{enumitem}
\usepackage{subcaption}  % For subfigures and subcaptions
\usepackage{xcolor}  % To use color in the document
\usepackage{todonotes}

\newtheorem{theorem}{Theorem}[section]  
\theoremstyle{remark}
\newtheorem*{remark}{Remark}

\newcommand{\eb}[1]{\textcolor{blue}{#1}}
\newcommand{\dn}[1]{\textcolor{green}{#1}}

\newcommand{\ignore}[1]{}

\author{Gordon Dai\textsuperscript{*1}}
\orcid{0009-0004-0572-3571}
\email{td2568@nyu.edu}
\author{Pavan Ravishankar\textsuperscript{*1}}
\orcid{0000-0002-5972-9539}
\email{pr2248@nyu.edu}

\author{Rachel Yuan\textsuperscript{1}}
\email{ry1023@nyu.edu}

\author{Daniel B. Neill\textsuperscript{\textdagger1~2~3}}
\orcid{0000-0001-6282-1240}
\email{daniel.neill@nyu.edu}

\author{Emily Black\textsuperscript{\textdagger4~5}}
\orcid{0000-0002-4301-3739}
\email{emilyblack@nyu.edu}

\title{Be Intentional About Fairness!: Fairness, Size, and Multiplicity in the Rashomon Set}

\begin{document}

\thanks{\textsuperscript{*}Co-first authors (equal contribution)}
\thanks{\textsuperscript{†}Co-senior authors (equal contribution)}

\thanks{\textsuperscript{1}Machine Learning for Good Laboratory, Department of Computer Science, NYU Courant Institute of Mathematical Sciences, New York University}
\thanks{\textsuperscript{2}Robert F. Wagner Graduate School of Public Service, New York University}
\thanks{\textsuperscript{3}Center for Urban Science and Progress, Tandon School of Engineering, New York University}
\thanks{\textsuperscript{4}Center for Data Science, New York University}
\thanks{\textsuperscript{5}Department of Computer Science and Engineering, Tandon School of Engineering, New York University}

\begin{abstract}
When selecting a model from a set of equally performant models, how much unfairness can you really reduce? Is it important to be intentional about fairness when choosing among this set, or is arbitrarily choosing among the set of ``good'' models good enough?
Recent work has highlighted that the phenomenon of model multiplicity---where multiple models with nearly identical predictive accuracy exist for the same task---has both positive and negative implications for fairness, from strengthening the enforcement of civil rights law in AI systems to showcasing arbitrariness in AI decision-making.
Despite the enormous implications of model multiplicity, there is little work that explores the properties of sets of equally accurate models, or Rashomon sets, in general. In this paper, we present five main theoretical and methodological contributions which help us to understand the relatively unexplored properties of the Rashomon set, in particular with regards to fairness. Our contributions include methods for efficiently sampling models from this set and techniques for identifying the fairest models according to key fairness metrics such as statistical parity. We also derive the probability that an individual’s prediction will be flipped within the Rashomon set, as well as expressions for the set's size and the distribution of error tolerance used across models. These results lead to  
policy-relevant takeaways, such as the importance of intentionally looking for fair models within the Rashomon set, and understanding which individuals or groups may be more susceptible to arbitrary decisions.
\end{abstract}

\maketitle

\section{Introduction}\label{sec:intro}

Recent work has drawn renewed attention to the fact that there are often many (approximately) equally accurate models available for the same prediction task~\cite{marx2019, black2022model, breiman2001statistical}. This phenomenon---often called the Rashomon effect~\cite{breiman2001statistical}, predictive multiplicity~\cite{marx2019}, or model multiplicity~\cite{black2022model}---has wide-ranging implications for both understanding and improving fairness, as these equally accurate models often differ substantially in other properties such as fairness~\cite{rodolfa2021empirical, laufer2024fundamental} or model simplicity~\cite{semenova2019study,semenova2022existence, rudin2024amazingthings}. 

As prior work has pointed out, this multiplicity of models can be viewed as both a fairness opportunity and a concern~\cite{black2022model, creel2022algorithmic}.
On the positive side, legal scholarship 
has pointed to the fact that model multiplicity is relevant to how to interpret and enforce U.S. anti-discrimination law, and specifically, can strengthen the disparate impact doctrine to more effectively combat algorithmic discrimination~\cite{black2023less}. 
In a recent paper, Black et al.~\cite{black2023less} suggest that the phenomenon of model multiplicity could support a reading of the disparate impact doctrine that requires companies to proactively search the set of equally accurate models for less discriminatory alternatives that have equivalent accuracy to a base model 
deemed acceptable for deployment from a model performance perspective.

On the negative side, several scholars have pointed out that facially similar models, with equivalent accuracy but differences in their individual predictions, can suggest that some model decisions are arbitrary since they seem to be made on the basis of model choice that does not impact performance (e.g., a <1\% change in a model's training set accuracy)~\cite{marx2019,black2021leave,gomez2024algorithmic}. 

This arbitrariness can impact model explanations and recourse as well: individuals with decisions that are unstable across small model changes may not receive reliable explanations for their model outcome, or ways to change it~\cite{pmlr-v124-pawelczyk20a,black2022consistent,black2022selective}. 
Further, if there is a group-based asymmetry of arbitrariness--e.g., if female loan applicants have more arbitrariness in their decisions than male loan applicants--- this could lead to a group-based equity concern in and of itself.

Understanding the extent of the benefits and risks of model multiplicity relies upon an understanding of the properties of the \emph{Rashomon set}, or the set of approximately equally accurate models for a given prediction task, i.e., equally accurate up to some error tolerance $\epsilon$. While models in the Rashomon set are considered equivalent from a performance perspective, they may differ substantially in other properties---for the purposes of our paper, we focus on fairness. In order to understand the utility of searching for fairer models within the Rashomon set as suggested by recent legal literature, or the extent of the dangers of arbitrariness surfaced by the algorithmic fairness community, we need to understand 

more about Rashomon sets themselves. For example, whether companies should be required to search for less discriminatory models~\cite{black2023less} rests on the question of how much of the disparity can be reduced by optimizing
over the Rashomon set, as compared to choosing an arbitrary model without regard to fairness. In other words, \textbf{how much do we
gain by being intentional about fairness} when selecting models within the Rashomon set? Similarly, concerns about arbitrariness relate to rates and distributions of the chance that an individual will have their prediction changed---\textbf{is this arbitrariness harmful} if only predictions that are very uncertain get flipped, or if all demographic groups have an equal chance of flipping? 
We can shed light on these important questions by understanding even basic facts about the Rashomon set, such as: what does a randomly sampled model from the Rashomon set look like? What is the \textbf{average fairness} for various metrics on the Rashomon set? How might one search through the Rashomon set? Can we find the \textbf{fairest model} within the Rashomon set? What is the chance that any one individual might experience a change in prediction in the Rashomon set? Or even, \textbf{how large is the Rashomon set?}  Despite the enormous implications of model multiplicity,
there is little work that explores the properties of Rashomon sets in general.

In this paper, we present five main theoretical and methodological contributions that answer the above questions and more---furthering our understanding of the relatively unexplored properties of the Rashomon set, in particular with regards to fairness:

\begin{itemize}[left=10pt]
    \item First, we define the \emph{largest possible Rashomon set}
    $R_N(\epsilon)$, for $N$ records drawn from a given data distribution, assuming an allowable error tolerance of $\epsilon$.  We then
    develop an efficient method for sampling models uniformly at random from within the Rashomon set.
    \item Second, we present two computationally efficient methods to find the fairest model within the Rashomon set, for statistical parity and error rate balance respectively. 
    \item Third, we derive the asymptotic probability that any individual will have their prediction flipped within the models of the Rashomon set for a given $\epsilon$. 
    \item Fourth, we derive a closed-form expression for the size of the Rashomon set for a given $N$ and $\epsilon$.
    \item Fifth, we show that for sufficiently large datasets and small enough $\epsilon$, models in the Rashomon set will use the full error tolerance (i.e., the average accuracy of models in $R_N(\epsilon)$ converges to the accuracy of the Bayes-optimal model minus $\epsilon$). 
\end{itemize}
\ignore{Importantly, as we explain further in Section~\ref{sec:prelim}, the theoretical setup of our work is such that we are exploring properties of an idealized Rashomon set: for one, we model the Rashomon set as a set of mappings from feature inputs to classification decisions based on  flipping predictions from the Bayes-optimal model, which is a theoretical model that can predict the true underlying probability of a positive or negative outcome of an individual given the information available. Secondly, we explore the largest possible Rashomon set based around this Bayes-optimal model--- the set of \emph{all} possible mappings from input to classification decisions over a dataset that satisfy a given error. In particular, viewing the Rashomon set in this way does not take into account the possibility of not being able to arrive at these mappings via training a parameterized model---e.g., there might be some mappings that are so discontinuous in the input space that it would be difficult to train a model to reach that mapping.}

These theoretical results create important newfound understanding of the Rashomon set with a focus on fairness and fairness-relevant properties--- to our knowledge, there are no results about how to sample randomly from the Rashomon set, Rashomon set size, individual flip probabilities within the Rashomon set, and the distribution of error used in the Rashomon set for any generalized theoretical setup. While concurrent work has shown that finding the fairest model within the Rashomon set is hard (NP-hard) in general~\cite{laufer2024fundamental}, we are able to show that under certain conditions we can find the fairest model very efficiently. 

Further, our theoretical results lead us to some interesting, policy-relevant takeaways, which we expand on further in Sections~\ref{sec:intentional}-\ref{sec:rashomon_set_size_and_error} and support with experiments on three datasets:
\begin{enumerate}[left=10pt, label=\Alph*.]
    \item We can gain a lot of fairness by intentionally searching for fairer models within the set of equally accurate models. Sampling randomly within the Rashomon set--- only optimizing for accuracy when selecting a model and hoping that it is fair--- will yield a much less fair model than searching for the fairest possible model \emph{even among those that are approximately equally accurate}, so explicitly optimizing for fairness within the Rashomon set is important.

    \item We can calculate the probability that any given individual will experience a flip in prediction among models in the (largest possible) Rashomon set. This allows us to shed light on the fates of individuals in the Rashomon set and potential inequities in flip probabilities when viewing inconsistency in the Rashomon set as a source of arbitrariness. We can see what factors--such as the distribution of prediction certainty and other dataset-specific factors---influence the individual and overall probability of flipping in a given Rashomon set.
    \item Finally, our theoretical results allow us to understand the size of the Rashomon set and amount of error tolerance used on average within the set.  In particular, we derive large-sample convergence results for the size of the Rashomon set over $N$ data records, as a function of the error tolerance $\epsilon$. This leads to the takeaway that a company searching for fairer models within the Rashomon set should use the largest error tolerance acceptable to their business, since the size of the set (and thus the opportunity to find a fairer model) increases very quickly in $\epsilon$. However, the company should be ready to have all of that error tolerance be used--- since our results around error tolerance show that as the dataset increases in size, the average model in the Rashomon set uses all of the error tolerance (i.e., has accuracy $\epsilon$ less than the base model).
\end{enumerate}

The remainder of the paper will proceed as follows: after discussing related work in Section~\ref{sec:related}, we will outline our theoretical setup and notation in Section~\ref{sec:prelim}. We encourage even non-technical audiences to read this section as it also outlines some of the important benefits and limitations of our theoretical setup. We then turn to presenting our theoretical work and policy takeaways together in the next three sections: in Section~\ref{sec:intentional}, we present new, efficient optimization and sampling approaches to find the fairest model and to sample a model uniformly at random from the Rashomon set, respectively, and demonstrate how that leads to our results showing the importance of intentionally searching for fair models. Next, in Section~\ref{sec:flip_probs}, we present our results on individual prediction flip probabilities, and how this sheds light on arbitrariness and other fairness properties within the Rashomon set. Finally, in Section~\ref{sec:rashomon_set_size_and_error}, we introduce our results on Rashomon set size and use of the error tolerance $\epsilon$, and discuss how they can inform how one might search within the Rashomon set for fairer models. Following this, in Section~\ref{sec:conclusion} we conclude the paper.

\section{Related work and Legal Background}
\label{sec:related}

\noindent\textbf{Related Work.} There has been a growing stream of work exploring the phenomenon of multiple approximately equally accurate models existing for the same prediction task~\cite{breiman2001statistical,marx2019, black2022model,watson2023predictive,rudin2024amazingthings,d2020underspecification}. Outside of fairness concerns, a series of papers have demonstrated how model multiplicity can be harnessed to find simpler models within the Rashomon set~\cite{semenova2022existence,rudin2024amazingthings,dong2019variable}, how the existence of multiple equally accurate models can disrupt model explainability~\cite{pmlr-v124-pawelczyk20a,black2022consistent}, and how sets of equally accurate models can differ greatly in their adversarial robustness~\cite{d2020underspecification}. Most related to this work, a series of papers focusing on interpretability of models within the Rashomon set have demonstrated how to search for more interpretable models in practice for particular model classes, e.g., decision trees~\cite{xin2022exploring, mata2022computing}, and have provided empirical observations of Rashomon set size for given model classes~\cite{xin2022exploring}. 

Within literature related to fairness concerns, two main themes have emerged: the optimistic vision of using the variability within the Rashomon set to achieve fairness goals with little impact on accuracy~\cite{gillis2024operationalizing, black2023less,rodolfa2021empirical}, and works bringing to light concerns about the arbitrariness of individual decisions from models with many nearly equally accurate counterparts that differ in their predictions, explanations, or other properties~\cite{marx2019,gomez2024algorithmic,cooper2024arbitrariness,black2021leave}. 
%For example, on the fairness searching side, Rodolfa et al~\cite{rodolfa2021empirical} demonstrate that for a series of real-world policy problems, such as housing safety prediction and student drop-out risk prediction, it is possible to find less discriminatory models within the Rashomon set with negligible loss in performance~\cite{rodolfa2021empirical}. 
On the arbitrariness side, many works show how models with minimal differences between them---e.g, a change in random seed or sampling of training data---can result in models with different predictions for certain individuals~\cite{marx2019,gomez2024algorithmic,cooper2024arbitrariness,black2021leave}. 
In this line of work, perhaps the most related is~\cite{cooper2024arbitrariness}, who show empirically that different individuals have radically different chances of experiencing a change in prediction among approximately equivalent models. In our work, we derive the exact probability that an individual will experience a change in prediction in the Rashomon set, and show that this probability varies as a result of a person's underlying certainty of prediction as well as dataset-dependent factors.
%Gomez et al.~\cite{gomez2024algorithmic} demonstrate that changes as small as the random seed of an LLM can result in a difference in classification of toxic content.
%The most relevant work from the fairness searching literature is discussed in the legal background and in the next paragraph. F

On the fairness side, some of the most related works touch on the details of searching through the Rashomon set for less discriminatory models, or less discriminatory alternatives (LDAs). For example, Gillis et al.~\cite{gillis2024operationalizing} outline what an LDA search may look like in practice, and develop an algorithm for searching through the set of \emph{linear} models for the least discriminatory alternative. 
Perhaps the most closely related work, by Laufer et al.~\cite{laufer2024fundamental}, outlines a series of theoretical results related to the search for less discriminatory models within the Rashomon set, such as the computational hardness of finding fairer models within the Rashomon set in general, the theoretical limits of fairness within the Rashomon set, and problems around generalizability of less discriminatory models discovered through search. The paper largely points to difficulties around finding a fairer model within the Rashomon set. In contrast, on a high level, one of the major points of our work is to showcase the \emph{importance} of intentionally searching for fairer models within the Rashomon set, by showing the immense fairness difference between models randomly chosen from the Rashomon set (i.e., on the basis of accuracy alone) and the fairest models within the Rashomon set. 

More generally, our work presents, for the first time, general properties about the Rashomon set itself--- such as the average fairness of models within the Rashomon set, the probability that any individual within the Rashomon set will experience a change in prediction across the models in the set, Rashomon set size, the distribution of model error within the Rashomon set, and others---and discusses how these results influence our understanding of not only how to search for fairer models within the Rashomon set, but also how we think about the arbitrariness of individual decisions within the Rashomon set.

\noindent\textbf{Legal Background and LDA Search} We now discuss some of the legal background necessary to understanding how model multiplicity can strengthen the enforcement of civil rights law in AI systems--- but also raises important questions about the utility of searching through the Rashomon set for fairer models.
Multiplicity relates to the interpretation and enforcement of civil rights law most directly through the \emph{disparate impact doctrine}. The disparate impact doctrine applies in decision making systems determining access to credit, housing, and employment opportunities, stating that it is illegal to have a decision-making system that distributes these opportunities across different protected demographic groups at different rates unless it is a ``business necessity''. 
In practice, the disparate impact doctrine is enforced through a three-step process. First, a plaintiff finds evidence of a decision-making system within a company that distributes opportunities at different rates among demographic groups, such as a bank that approves loans to more men than women. Next, the company argues that this disparate impact is a business necessity---while there is no exact description of what a business necessity is, a general understanding is that the disparity would be necessary for the business to function. In the case of AI decision-making systems, this is often argued by stating that the algorithm used has the highest accuracy possible, that this accuracy is necessary for business function, and that the observed disparity is necessary to achieve this accuracy. However, even if this business necessity defense is accepted, if the plaintiff can demonstrate that there is a \emph{less discriminatory alternative} decision-making system that satisfies business necessity but reduces disparate impact, the firm can be legally liable for the discrimination they have caused, and forced to use the less discriminatory alternative. In the case of algorithmic systems, i.e., when the alternative decision-making system is another algorithm, we follow~\cite{black2023less} in calling the less discriminatory alternative algorithm an LDA. 
%Thus, already within the law, there is indirect pressure for a company to search through the range of equivalent accurate (or business-necessity meeting) models and search for the one with the least disparate impact, to mitigate legal risk--- but the extent to which companies do so is varied. 
Thus, companies subject to the disparate impact doctrine are theoretically incentivized to search for less discriminatory yet still effective models, for fear of being held liable should another entity find a less discriminatory alternative. 
Some businesses, mostly financial institutions, do this in practice, though domain experts note that ``there is an uneven landscape with respect to how or whether institutions assess their models for discrimination, and the effectiveness of existing programs''~\cite{colfax2024report4}. 

In a recent paper, Black et al.~\cite{black2023less} outline a novel interpretation of the disparate impact doctrine that puts even greater pressure on companies to search for LDAs. They suggest that since multiple equally accurate models exist for the same prediction task--some of which will likely have different fairness properties--- the business necessity argument fails to make sense, and instead, a company should do a proactive search through the Rashomon set of equally accurate models in order to ensure there is no less discriminatory model easily available.  
A critical question that this raises, however, is 
how much of the disparity can be reduced by optimizing over the Rashomon set, as compared to choosing an arbitrary model within that set without regard to fairness.  In other words, how much do we gain by being \emph{intentional} about fairness within the Rashomon set---by looking for fair models among those that are approximately equally accurate? In this paper, we show that it is well worth it to search for fairer models within the Rashomon set, and that being intentional about doing so is important, as well as other critical insights about the Rashomon set.

\section{Preliminaries and Notation}
\label{sec:prelim}

In this section, we introduce the mathematical setup and assumptions behind our theoretical results and discuss the implications of these decisions. To define the Rashomon set of approximately equally accurate models, we consider four questions: (i) how do we define a model? (ii) when are models considered distinct? (iii) how do we measure the accuracy of a model? and (iv) if the Rashomon set consists of all models with accuracy within $\epsilon$ of some ``optimal'' model, how is that model defined?  

\vspace{0.5em}
\noindent\textbf{Basics and Model Definition.} 
To answer the first and second questions above, we consider Rashomon sets in the finite-sample case, i.e., assuming we have a fixed number of data records $N$. 
Later in the paper, we present theoretical results in the large-sample case, as $N$ goes to infinity. Additional preliminaries and assumptions necessary for those results are presented in Section~\ref{sec:flip_probs-assumptions}.  
Let $D_N = \langle d_1, d_2, \ldots, d_N \rangle$ be a set of $N$ data records drawn i.i.d. from distribution $D$. We focus on the binary classification setting, where each data record $d_i = (x_i, y_i)$, $x_i = \{x_{ij}\}$ represents a set of input features (including a binary \emph{sensitive attribute} which we denote as $A_i$), and $y_i$ is a binary outcome variable. Thus our models are binary classification models, which predict an outcome in $\{0,1\}$. 
We define a predictive model by its classification $\hat y_i = f(x_i)$ for each data record $d_i$, that is, by its mapping from input features $x_i$ to decisions $\{0,1\}$ on the data $D_N$. Thus, there are $2^N$ distinct models possible for a set of data records of size $N$. 

\vspace{0.3em}
\noindent\textit{Implications.} Note that this definition focuses on \emph{predictive multiplicity}~\cite{marx2019} rather than \emph{procedural multiplicity}~\cite{black2022model}: while previous work has shown that models can reach the same decisions in different ways, e.g., relying on different input features~\cite{black2022model}, we view models as identical if they reach the same predictions.
Further, we note that some of the %$2^N$ possible 
%\eb{taking out 2N because its confusing?}
models in the Rashomon set 
may be difficult or impossible to find via supervised learning, 
since one would typically restrict the model class (e.g., to logistic regression models, or neural networks with a given architecture, etc.) and optimize the model parameters within this class (to minimize some loss function) on a separate, labeled training dataset. 
\textbf{Thus we term the Rashomon set $R_N(\epsilon)$ of approximately equally accurate models within this set of $2^N$ models as the \emph{largest possible Rashomon set} for error tolerance $\epsilon$, because it places no restrictions on the model class, smoothness or consistency of predictions, etc.}.

\vspace{0.5em}
\noindent\textbf{Model Accuracy and Optimal Model.} Our answers to the third and fourth questions above rely on the concept of a \emph{Bayes-optimal classifier} $f_\text{opt}(x_i)$. This model is assumed to have access to the true probabilities $p_i = \mbox{Pr}(y = 1 \:|\: x = x_i)$ but not the observed labels $y_i$. In other words, the Bayes-optimal classifier has access to the underlying probability that given the available input information, an individual data record will have true outcome $y=1$ in the classification problem (e.g., the probability that an individual will repay a loan based on their application), but not the actual outcome (whether or not that individual defaulted on the loan). The Bayes-optimal classifier predicts $f_\text{opt}(x_i) = 1$ if $p_i > 0.5$, and $f_\text{opt}(x_i) = 0$ otherwise, and has the highest expected classification accuracy, $\mathbb{E}[\max(p_i,1-p_i)]$, among all classifiers using the same set of features $x$.
Thus, given the data records $D_N=\langle d_1,d_2,\ldots,d_N\rangle$ and the corresponding true probabilities $P_N = \langle p_1,p_2,\ldots,p_N\rangle$, we define the Rashomon set $R_N(\epsilon)$ for error tolerance $\epsilon$ as the set of all models with expected classification accuracy greater than or equal to $\mathbb{E}[\max(p_i,1-p_i)] - \epsilon$.  

\vspace{0.3em}
\noindent\textit{Implications.} \textbf{This definition has the advantage of not allowing models to \emph{overfit} the observed data,} since expected error is calculated as a function of the underlying probability $p_i$ of an input $x$ having an outcome of 1. If we instead used the observed labels $y_i$ and computed the empirical accuracy $\mathbb{E}[\mathbf{1}\{f(x_i) = y_i\}]$, a non-Bayes-optimal model (e.g., a classifier trained on the test data $D_N$) could obtain higher empirical accuracy than the Bayes-optimal model, e.g., by predicting $f(x_i) = 1$ for a data record that was \emph{a priori} unlikely to have $y_i=1$ (i.e., $p_i < 0.5$) but just happens to have $y_i=1$ in this instance.
However, it is also worth noting that the Bayes-optimal probabilities $p_i$ may be unknown for a real-world dataset $D_N$, since they are based on the distribution $D$ from which the data records in $D_N$ are drawn.  Nevertheless, these probabilities can be estimated from a separate, large training dataset.
Additionally, in this setup $R_N(\epsilon)$ will differ depending on what model you start with. Thus, when exploring Rashomon sets with this method in practice, different estimations of the Bayes-optimal model will result in slightly different Rashomon sets. \ignore{Even with only Bayes-optimal models, if you use a different feature set and end up with a different Bayes-optimal model, you will end up with a different Rashomon set.} 

\vspace{0.5em}
\noindent\textbf{Defining Other Models in The Rashomon Set.} To more easily determine which of the $2^N$ possible models (mappings of each $d_i$, $i\in\{1,\ldots,N\}$, to \{0,1\}) belong to the Rashomon set $R_N(\epsilon)$, we represent each possible model by a binary \emph{flip vector} representing its changes in prediction from the Bayes-optimal model. This allows us to easily tell which models are in the Rashomon set, since we can easily calculate a model's error difference from the Bayes-optimal model using its flip vector. In particular, we define a flip vector $\theta \in \{0,1\}^N$, where $\theta_i = 1$ if $f(x_i) \ne f_\text{opt}(x_i)$, and $\theta_i = 0$ if $f(x_i) = f_\text{opt}(x_i)$. The Bayes-optimal model $f_\text{opt}(\cdot)$ has a corresponding flip vector $\theta_0$ consisting of $N$ zeros.  We can then compute the accuracy of any model $f(\cdot)$ with corresponding flip vector $\theta$, which we denote as $acc(\theta)$, as $acc(\theta) = acc(\theta_0) - \frac{1}{N} \sum_{i=1\ldots N} \theta_i |2p_i-1|$.
This follows from the fact that the Bayes-optimal classifier's probability of predicting $y_i$ correctly is $\max(p_i,1-p_i)$, while the flipped prediction ($\theta_i = 1$) would be correct with probability $\min(p_i,1-p_i)$, leading to a difference of $|2p_i-1|$. We thus define the weight $w_i$ corresponding to probability $p_i$ as $w_i = |2p_i - 1|$.  These weights can be thought of as the Bayes-optimal classifier's confidence in each positive or negative prediction, and range from 0 (for $p_i=0.5$) to 1 (for $p_i=0$ or $p_i=1$).  Let $W_N = \langle w_1,w_2,\ldots,w_N\rangle$ be the weight vector for data records $\langle d_1,d_2,\ldots,d_N\rangle$, and then we can write:
\[
acc(\theta) = acc(\theta_0) - \frac{\theta \cdot W_N}{N}.
\]
Finally, for a given error tolerance $\epsilon$, we define the largest possible Rashomon set $R_N(\epsilon)$ as all flip vectors $\theta \in \{0,1\}^N$ with $acc(\theta) \ge acc(\theta_0) - \epsilon$, and thus:
\[
R_N(\epsilon) = \left\{ \theta \in \{0,1\}^N : \frac{\theta \cdot W_N}{N} \le \epsilon \right\}.
\]
\noindent\textit{Implications.} The critical takeaway here is that we can enumerate all of the models in the Rashomon set by checking to see which of the $2^N$ possible flip vectors fall 
within the accuracy constraint $\epsilon$. But since it would be too costly in practice to do this for all $2^N$ flip vectors, we show below how to randomly sample (efficiently) from the Rashomon set---and how to find the fairest model. 

\vspace{0.5em}
\noindent\textbf{Benefits and Limitations of our Setup.}
We summarize the main takeaways from our theoretical setup that frame our results. 
First, we think of models as mappings from input features to binary decisions in $\{0,1\}$. This means that models will exist in the Rashomon set that may not be reachable by regular training methods. Thus, this paper explores behavior of the largest possible Rashomon set. 
Second, the way we calculate error was designed to avoid rewarding models for overfitting. Thus, in the context of this paper, we do not have to worry about generalizability concerns.
Finally, we assume that the basis of the Rashomon set is the Bayes-optimal model. This model is sometimes not discoverable from available data in practice. Also, since we define our Rashomon set in terms of deviation from a base model, changes in that model will lead to changes in the Rashomon set-- and so, it is important to find as accurate as possible a baseline model when using this method of Rashomon set exploration in practice.

These choices lead to both a benefit and a downside: it allows us to see what is \emph{possible} to achieve within the Rashomon set with maximum flexibility, allowing us to see how much we can strive to accomplish. At the same time, it also may deviate from what people observe in practice (e.g., by searching only through limited model classes, or not starting with the Bayes-optimal model). Testing to see how much exploration of subsets of the Rashomon set reachable through traditional training methods differs from the results we present here is an important and ongoing area of future work, and is also touched upon later in this paper.

\section{The Importance of Intentional Fairness}\label{sec:intentional}

\begin{figure}[t]
    \centering
    % First subfigure
    \begin{subfigure}[b]{0.32\textwidth}  % Adjust width to your needs
        \centering
        \includegraphics[width=\textwidth]{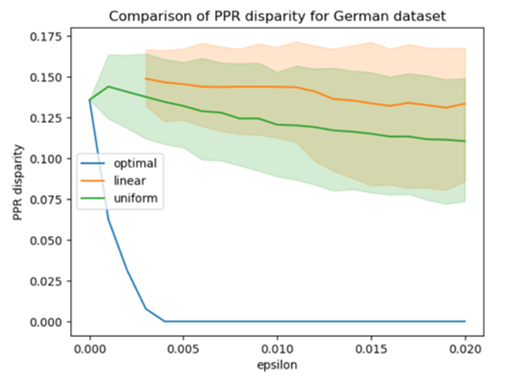}  % Replace with your image
    \end{subfigure}
    \hfill
    % Second subfigure
    \begin{subfigure}[b]{0.32\textwidth}
        \centering
        \includegraphics[width=\textwidth]{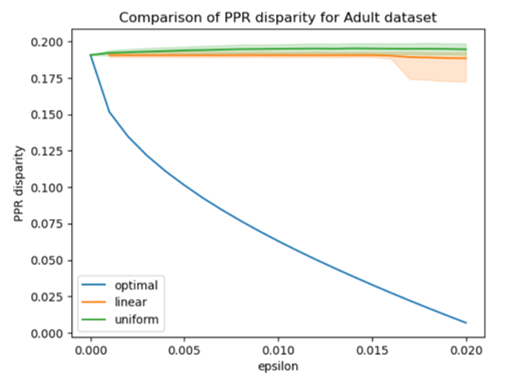}  % Replace with your image
    \end{subfigure}
     % Second subfigure
    \begin{subfigure}[b]{0.32\textwidth}
        \centering
        \includegraphics[width=\textwidth]{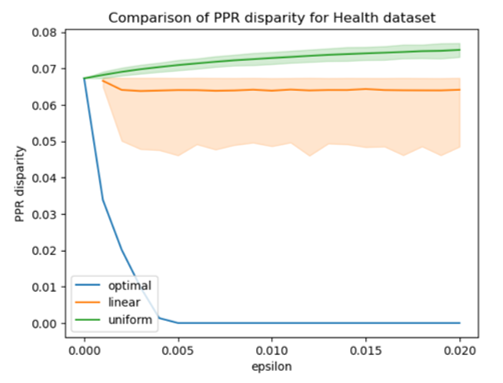}  % Replace with your image
    \end{subfigure}
    \caption{Disparity in positive prediction rate for the German, Adult, and Health datasets, as a function of the error tolerance $\epsilon$. Comparison of methods for optimizing PPR (Section~\ref{sec:optimizing-PPR}), uniform random sampling (Section~\ref{sec:sampling}), and sampling linear models (Section~\ref{sec:linear}) over the Rashomon set $R_N(\epsilon)$.}
    \label{fig:PPR_disparity}
\end{figure}

%In this Section, we demonstrate the importance of explicitly searching for fair models within the Rashomon set. 

A natural question that may come up when considering searching for fairer models within the Rashomon set is---is it worth it? While it is clear that intentionally searching for fair models \emph{without} a strict bound on accuracy leads to large fairness gains, it is not obvious \emph{a priori} that this holds true within sets of models that are approximately equally accurate. What if the fairness of all the models in the Rashomon set is more or less the same, and a randomly sampled model---akin to selecting a model solely on the basis of accuracy and not paying attention to fairness---is just as fair as the fairest ones within the set? 
We show that this is definitely not the case--- the fairness difference between the average, or randomly sampled, model within the Rashomon set and the fairest models can be very large. We also show experimentally that how you look for fairer models can influence your success---while searching directly for the fairest model is always the most effective method, whether or not you can reach significantly fairer models by randomly sampling models in the Rashomon set is dataset-dependent, and also depends on how you search.

To compare what we gain by being intentional or arbitrary about fairness within the Rashomon set, 
we must show both how to draw randomly from the Rashomon set and how to find the fairest model.
We present \textbf{novel, computationally efficient approaches} for (i) optimizing different fairness metrics over the Rashomon set, as described in Section~\ref{sec:optimizing-fairness}; and (ii) sampling models uniformly at random from the Rashomon set, as described in Section~\ref{sec:sampling}.  We also describe a simple baseline for comparison in Section~\ref{sec:linear}: restricting the model class (here we assume penalized logistic regression models) and learning models from that class with different sources of random variation.  We compare the fairness of the models found by the optimization, uniform sampling, and restricted model class approaches in Section~\ref{sec:intentional-experiments}, and explore policy takeaways in Section~\ref{sec:intentional-takeaways}. 

\subsection{Optimizing fairness over the Rashomon set}
\label{sec:optimizing-fairness}

Despite computational hardness results for finding the fairest model in the Rashomon set~\cite{laufer2024fundamental}, we show that under certain conditions, it is possible to find the fairest model within the Rashomon set $R_N(\epsilon)$ defined on $N$ data records in log-linear time, $\mathcal{O}(N \log N)$. In particular, when we are concerned with mitigating demographic disparity (i.e., equalizing the positive prediction rate or PPR) between two groups, we show that we can find the exact fairest model within the Rashomon set. For equalizing false positive rate (FPR) or true positive rate (TPR) between two groups, we can find a model which is guaranteed to have error rate disparity no more than $\mathcal{O}(\frac{1}{N})$ higher than the fairest model. As we show in Section~\ref{sec:intentional-experiments}, using these algorithms on three real-world datasets, we see that in practice, it is often possible to completely eradicate disparities by searching within the Rashomon set for very small $\epsilon$--- less than half of a percent in many cases.

For PPR, FPR, and TPR, we can express the optimization of the fairness criterion over flip vectors $\theta$, subject to the constraint that 
$\theta$ is in the Rashomon set $R_N(\epsilon)$, as a \emph{knapsack problem}, where each data record $d_i$ has a weight $w_i = |2p_i-1|$ corresponding to the error incurred by its flip, and a value $v_i$ corresponding to how much it reduces disparity. A flip occurring, i.e., $\theta_i=1$, corresponds to the inclusion of element $i$ in the knapsack, adding $w_i$ to the total weight and $v_i$ to the total value.  The 0-1 knapsack problem is then the constrained optimization with capacity $N\epsilon$: $\max \sum_i \theta_i v_i$ subject to $\theta_i \in \{0,1\}$ and $\sum_i \theta_i w_i \le N\epsilon$. 

We note that concurrent work by Laufer et al.~\cite{laufer2024fundamental} formulates the optimization of fairness over the Rashomon set as a subset sum problem (closely related to the knapsack problem) and uses this equivalence to show that their problem (i) is NP-hard to solve in general, and (ii) can be approximated in $\mathcal{O}(N^3)$ time. While the knapsack problem is also NP-hard in general, we present efficient $\mathcal{O}(N \log N)$ solutions for the special cases below.

\subsubsection{Optimizing for statistical parity} 
\label{sec:optimizing-PPR}
We present an efficient, $\mathcal{O}(N\log N)$ knapsack approach to find the exact fairest model that minimizes PPR disparity over the Rashomon set $R_N(\epsilon)$, as described in detail in Appendix~\ref{appendix:statdisp}, Algorithm~\ref{algorithm:knapsack}. The goal of this algorithm is to find the individual predictions to flip (setting $\theta_i=1$) to reduce disparity, until we either use up the entire error tolerance $\epsilon$ or completely remove the disparity. Intuitively, we want to flip individuals who will increase the error as little as possible (low weights $w_i$) and reduce the disparity as much as possible (high values $v_i$).  

The key idea for making this efficient is that there are only two distinct values of $v_i$: for instance, if group $A$ has a higher positive prediction rate, flipping the prediction of an individual in group $A$ from 1 to 0 reduces the disparity by $\frac{1}{|A|}$, flipping the prediction of an individual in group $B$ from 0 to 1 reduces the disparity by $\frac{1}{|B|}$, and other flips would increase disparity.
In this case, the optimal knapsack solution is to flip the predictions of the $k_A$ lowest-weight individuals with $f_\text{opt}(x_i)=1$ from group $A$ and the $k_B$ lowest-weight individuals with $f_\text{opt}(x_i)=0$
from group $B$.  We can then find the optimal values of $k_A$ and $k_B$ (that minimize disparity while satisfying the constraint on accuracy) through a linear-time, incremental search, as shown in Appendix~\ref{appendix:statdisp}, Algorithm~\ref{algorithm:knapsack}, and thus the run time is dominated by the $\mathcal{O}(N\log N)$ sorting of items by weight.

\subsubsection{Optimizing for error rate balance}
\label{sec:optimizing-TPR and FPR}
We present an efficient, $\mathcal{O}(N\log N)$ fractional knapsack approach to find the model that minimizes FPR or TPR disparity over the Rashomon set $R_N(\epsilon)$, to within $\mathcal{O}(\frac{1}{N})$ of the optimal disparity, as described in detail in Appendix~\ref{appendix:errbal}, Algorithm~\ref{algorithm:fractional_knapsack}. Again, the goal of this algorithm is to find the lowest-cost individual predictions to flip to reduce disparity, until we either use up the entire error tolerance $\epsilon$ or completely remove the disparity.

In this case, however, there are more than two distinct values of $v_i$ (how much flipping an individual reduces disparity) so the PPR solution described above does not work. Instead we use an approximation, the fractional knapsack solution, which flips individuals' predictions (setting $\theta_i=1$) in descending order of the ratio of their value $v_i$ (the amount they reduce the model's disparity) to their weight $w_i$ (the amount they increase the model's error). This continues until an individual will not ``fit'' in the knapsack since maximum weight (i.e., error threshold) is reached. Then, a ``fraction'' of this individual is added to the knapsack. In our case, we cannot flip a fraction of an individual--- thus, rather than adding the fractional element, we show that it would reduce disparity by an amount $\theta_i v_i$ that is $\mathcal{O}(\frac{1}{N})$. Since the fractional knapsack solution $\sum \theta_i v_i$ is an upper bound on the 0-1 knapsack solution, we know that our solution (excluding the fractional element) reduces disparity to within $\mathcal{O}(\frac{1}{N})$ of the optimal disparity.

\subsection{Sampling models uniformly at random from the Rashomon set}
\label{sec:sampling}
We now turn to showing how we can sample models uniformly from the Rashomon set, which shows us what typical models from the Rashomon set look like.
While we could just sample random flip vectors and keep the ones that are in the Rashomon set, this approach will be ineffective: as we discuss in Appendix~\ref{appendix:gibbs}, the vast majority of flip vectors will not be in the Rashomon set.  Instead, we propose a new approach based on Gibbs sampling~\cite{geman1984} to sample models uniformly from the Rashomon set.
This approach, described in Appendix~\ref{appendix:gibbs}, Algorithm~\ref{algorithm:gibbs}, is computationally efficient, requiring $\mathcal{O}(N)$ time per sample.

The key idea of Gibbs sampling is to exploit knowledge of \emph{conditional distributions} even when the full distribution is unknown. In our setting, while we do not know the joint distribution of flip probabilities $\theta$ for all records in the dataset, we can easily compute the chance that a data record $d_i$ will flip ($\theta_i = 1$) conditional on which other data records are flipped ($\theta_{-i}$).  We show in Appendix~\ref{appendix:gibbs} that there are only two possibilities: if the flip vector $\theta_{i=1}$ (with $\theta_i = 1$ and all other flips the same as $\theta_{-i}$) is in the Rashomon set, then there is a 50/50 chance that $\theta_i=1$, and otherwise we know $\theta_i = 0$.  We can then redraw $\theta_i$ with the corresponding probability (either 0.5 or 0) of being 1. Given this simple and computationally efficient conditional sampling step, our Gibbs sampling approach starts with the zero vector $\theta_0$, which is guaranteed to be in the Rashomon set, and iteratively samples $\theta_i$
(given the current values of $\theta_{-i}$) for all $N$ data elements. To ensure uncorrelated samples from the joint distribution, we take one sample every 10 iterations (where one iteration includes resampling all $N$ elements of $\theta$ in randomly permuted order), after an initial burn-in period of 500 iterations.  For each dataset and each value of $\epsilon$ considered, we run 10,000 iterations of Gibbs sampling, resulting in 950 samples.

\subsection{Experiments on real data}
\label{sec:intentional-experiments}
%\paragraph{Experimental Methodology} 
We now describe our experimental design for comparing randomly sampled and optimally fair models within Rashomon sets on real data, showing the importance of intentional fairness. Throughout this paper, we present experimental results on three real-world datasets that are commonly used as benchmarks in the fair machine learning literature: German Credit (``German''), Adult, and Heritage Health (``Health'').
Details of all three datasets are described in Appendix~\ref{appendix:datasets}.

As noted above, the Bayes-optimal probabilities $p_i$ are unknown for these real-world datasets, but can be well-estimated using sufficient training data.  Since we wish to compare the methods over all $N$ data records ($N=1,000$ for German, $N=46,443$ for Adult, and $N=184,308$ for Health), we performed 5-fold cross-validation to estimate these probabilities.  For each held-out 20\% of the data, we trained a model $\hat f_\text{opt}(x)$ using the remaining 80\% of the data to approximate the Bayes-optimal model $f_\text{opt}(x)$, and used its predicted probabilities $\hat p_i$ to estimate the Bayes-optimal probabilities $p_i$ for that fold.  More precisely, we trained logistic regression models on each dataset that matched typical (maximal) accuracies reported in the wider literature.  To check the robustness of our results to the choice of model used for estimation of $p_i$, we re-ran all experiments using the estimated probabilities $\hat p_i$ from XGBoost models learned using 5-fold cross-validation (Appendix~\ref{appendix:robustness}), and found no notable differences. 

To test the difference between randomly sampling from the Rashomon set and directly optimizing for the fairest model within the set on real data, for each dataset and each $\epsilon$ value, we compared the model found by optimizing the desired fairness metric (PPR, TPR, or FPR) over the Rashomon set $R_N(\epsilon)$, as described in Section~\ref{sec:optimizing-fairness}, to the distributions of models found by (i) uniform random sampling over all models (flip vectors) $\theta \in R_N(\epsilon)$, as described in Section~\ref{sec:sampling}, and (ii) a simple baseline approach, sampling penalized logistic regression models (and corresponding flip vectors $\theta$) from $R_N(\epsilon)$, as described in Section~\ref{sec:linear} below. 
For each distribution of samples, we report the mean and 95\% interquantile range, i.e., the 2.5 and 97.5 percentiles of the distribution.

We compare these approaches using three fairness criteria: statistical parity, or balanced positive prediction rate (PPR), balanced false positive rate (FPR), and balanced true positive rate (TPR).  
Disparities with respect to all three criteria were measured between the protected class ($A_i=a$) and non-protected class ($A_i\ne a$), using the sensitive attribute value for each dataset described in Appendix~\ref{appendix:datasets}. All three measures of disparity for a given flip vector $\theta$ were computed using the (estimated) Bayes-optimal probabilities $p_i$ and corresponding weights $w_i$, rather than the observed outcomes $y_i$, as described in Appendix~\ref{appendix:optifairness}.
Results for PPR disparity are shown in Figure~\ref{fig:PPR_disparity}, and results for FPR and TPR disparity are shown in Appendix~\ref{appendix:intentional-experiments}, Figures~\ref{fig:FPR_disparity} and~\ref{fig:TPR_disparity}.

\ignore{Algorithms 1 (random sampling) 2, and 3 (optimal fairness) (available in Appendix \todo{APP}, using the estimated Bayes-optimal probabilities $\hat{p_i}$,Figure~\ref{fig:PPR_disparity} shows the average, and upper and lower 95 percentile range for models found with the optimal fairness search strategy, randomly sampling, and randomly sampling within linear models only.}

\subsubsection{Baseline approach: sampling linear models from the Rashomon set}
\label{sec:linear}

As a simple baseline for comparison, which might be representative of how a company would typically choose a predictive model for deployment, we assume that a binary classifier is learned from a separate, large training dataset, where the model class is chosen $\emph{a priori}$ and therefore the set of possible flip vectors $\theta$ is restricted to members of that class. In particular, we assume that an $L_2$-penalized logistic regression model is learned. For consistency (since our experiments use all $N$ data records), we use $k$-fold cross-validation and compute all metrics using predictions (for a given data record $d_i$) made using a model learned from the other $k-1$ folds (excluding the fold that contains $d_i$).  Moreover, since a company would typically explore the space of parameter values and choose a model with high accuracy, we learn penalized logistic regression models with different sources of random variation, evaluate their accuracy, and keep those models which are in the Rashomon set.
More precisely, to sample over the Rashomon set of $L_2$-penalized logistic regression models, for a given dataset and value of $\epsilon \in \{0.001, 0.002, \ldots, 0.02\}$, we sample 1,000 models, where for each model we randomly sample the number of cross-validation folds $k \in \{2, 3, \ldots, 10\}$, the logistic regression solver, and the amount of $L_2$ penalization $C\in \{ 0.001, 0.01, 0.1, 1.0, 10, 100 \}$, and then fit the penalized logistic regression model using scikit-learn.
Given the model's predictions, we compute the flip vector $\theta$ and include the sampled model in the Rashomon set if $\frac{\theta \cdot W_N}{N} \le \epsilon$.

\subsection{Takeaways for policy and practice}
\label{sec:intentional-takeaways}
\begin{itemize}[left=10pt]
    \item \textbf{A randomly sampled model within the Rashomon set is nowhere near as fair as the fairest model at any given $\epsilon$.} As we can see from the gap between the blue and green lines in Figure~\ref{fig:PPR_disparity}, searching intentionally for the fairest model within the Rashomon set leads to much fairer models at the same $\epsilon$ than randomly sampling within the set. This shows us that a random model from the Rashomon set--- one selected on the basis of accuracy alone--- will have an extremely low chance of being the fairest, or even one of the fairer, models within the set. This in turn underscores the necessity of explicitly searching for fairer models before deployment-- i.e., an LDA search. 

    \item \textbf{In practice, it is often possible to completely eradicate disparities by searching within the Rashomon set for quite small $\epsilon$}. As we can see in Figure~\ref{fig:PPR_disparity}, for German Credit and Health datasets, a model exists that completely eradicates demographic disparity in the dataset for $\epsilon<0.005$, i.e., half a percentage point of accuracy loss. While the Adult dataset requires very slightly more than 2\% accuracy loss to fully eradicate the disparity, this is still a small enough gap considered to be acceptable based on case studies of LDA searches~\cite{colfax2022report3}.

    \item \textbf{Using repeated random sampling as a search strategy-- i.e., looking across many models selected on the basis of accuracy and searching for the fairest among them---can give mixed results.} While our theoretical setup does not map onto how LDA searches would be done in practice--- since we search through all the possible mappings of input to output for a dataset instead of generating actual parametric models---loosely, repeated random sampling corresponds to an LDA search that does not directly use protected attribute information until after all the models are trained, i.e., only as a step to evaluate models and choose among them post-hoc. Our optimal search method, on the other hand, corresponds to an LDA search method that uses some direct minimization of disparities across demographic groups during the model creation process, whether that be in hyperparameter tuning, optimization, or other parts of the pipeline. There is disagreement in the legal literature as to whether and to what extent interventions for disparity reduction across demographic groups that use protected class information are legally permissible~\cite{ho2020affirmative, kim2022race,gillis2021input}. Our experimental results show that we gain a lot by being able to directly intervene using protected attributes---however, repeated random sampling without direct use of protected attributes can in some cases be an effective technique as well, even if not as effective as direct intervention. In particular, we see that the German Credit results allow for a large reduction in disparity.  
    At $\epsilon=.02$ (i.e., 2\% error tolerance from the optimal model), the total PPR disparity could be reduced by 46\% compared to the Bayes-optimal model by reaching the 2.5 percentile of the PPR disparity distribution, which could be achieved in practice by taking the fairest (lowest PPR disparity) of 40 samples from the Rashomon set. In addition, note that while random sampling over the \emph{entire} Rashomon set of all possible mappings of $x$ to $y$ is not particularly effective at reducing disparity in the Health dataset, only searching within linear models is more effective--this is promising given that in practice, LDA searches can only be done within various model classes and not across all possible mappings. Divergence in the effectiveness of different random sampling methods based on model class is an interesting phenomenon that we look forward to studying in future work. 
\end{itemize}

%First, we experimentally demonstrate that while fair models exist within the Rashomon set, they can sometimes be difficult to find unless they are explicitly searched for. 

%we show that by combining algorithms 1 and 2, which search for random models within the Rashomon set and the fairest model within the Rashomon set respectively, 

\section{Understanding Individual Flip Probabilities}\label{sec:flip_probs}

In this section, we present our results showing how to compute expected \emph{flip probabilities} for every record $d_i$ across all models in the Rashomon set, i.e., the chance that a given individual will experience a change in prediction from the Bayes-optimal model in a randomly sampled model in the Rashomon set. 
%Most critically, knowing the flip probabilities allows us to address the question raised in Section~\ref{sec:intro} of \emph{arbitrariness} of predictions. 
Knowing flip probabilities allows us to explore the \emph{arbitrariness} that arises from the Rashomon set:
many authors have pointed to the phenomenon of predictive multiplicity~\cite{marx2019}, where an individual can have different outcomes among different models in a Rashomon set, as a form of inequity through arbitrariness~\cite{marx2019,black2021leave}. By seeing the flip probabilities of any individual in the Rashomon set, we can see who is %likely to be 
more and less susceptible to potentially arbitrary changes in outcome---and as we discuss in Section~\ref{sec:flip_probs-experiments}, group-level disparities across who is likely to experience a change in prediction. 

\subsection{Preliminaries and assumptions for our large-sample theoretical results}
\label{sec:flip_probs-assumptions}

Throughout Sections~\ref{sec:flip_probs} and~\ref{sec:rashomon_set_size_and_error}, we present various theoretical results, and the corresponding takeaways for policy and practice, about individual flip probabilities, Rashomon set size, and use of error tolerance, in the \emph{large-sample limit} where the number of data records $N\rightarrow\infty$.  \textbf{For full statements and proofs of all theorems, see Appendix~\ref{appendix:proofs}.}  In this subsection, we present the notation needed to understand the theoretical results, along with the key assumptions that these results depend on.  

As in Section~\ref{sec:prelim}, we assume data records $d_i = (x_i, y_i)$ drawn i.i.d. from distribution $D$, with corresponding Bayes-optimal probabilities 
$p_i = \mbox{Pr}(y = 1 \:|\: x = x_i)$, and weights $w_i = |2p_i - 1|$.  Let $\langle d_1, d_2, \ldots \rangle$ denote an infinite sequence of data records drawn i.i.d. from $D$, and let $D_N$ denote records $\langle d_1, d_2, \ldots, d_N \rangle$, with corresponding Bayes-optimal probabilities $P_N=\langle p_1, p_2, \ldots, p_N \rangle$ and weights $W_N=\langle w_1, w_2, \ldots, w_N \rangle$. Moreover, let $W$ be the distribution of weights for data records drawn i.i.d. from $D$, $w_i \sim W$ for all $i$, with probability density function (pdf) $f(w)$.  

Let $R_N(\epsilon)$ denote the Rashomon set of models for error tolerance $\epsilon$ defined over data records $\langle d_1, \ldots, d_N\rangle$. We represent each model in $R_N(\epsilon)$ by a length-$N$ binary \emph{flip vector} $\theta \in \{0,1\}^N$, where $\theta_i = 1$ if $f(x_i) \ne f_\text{opt}(x_i)$, $\theta_i = 0$ if $f(x_i) = f_\text{opt}(x_i)$, and the Bayes-optimal classification $f_\text{opt}(x_i) = \mathbf{1}\{p_i > 0.5\}$.  As shown in Section~\ref{sec:prelim}, a flip vector $\theta \in R_N(\epsilon)$ if and only if $\frac{\theta \cdot W_N}{N } \le \epsilon$.

\textbf{Key assumptions} underlying the theoretical results below are threefold: (1) the number of data records $N$ is large; (2) the distribution of weights $f(w)$ is continuous and positive on the interval [0,1]; and (3) $\epsilon$ is sufficiently small, less than half of the average weight.  We observe that these assumptions are reasonable for all three datasets considered: (1) $N$ is large enough (ranging from $N=1,000$ for German Credit to $N=184,308$ for Health) for the finite-sample results to be very close to their large-sample limits; (2) there is enough variability in the weights $w_i$ to assume that they are drawn from a continuous, positive distribution; and (3) average weights for all three datasets range from 0.50 (German Credit) to 0.74 (Health), while the $\epsilon$ values we consider for our Rashomon sets are at most 0.02. Nevertheless, the assumptions might be violated for very small datasets (insufficient $N$); low-dimensional datasets with discrete-valued predictor variables (insufficient variability in $w_i$); or datasets where the prediction is extremely uncertain, $p_i \approx 0.5$ and thus $w_i \approx 0$, for many data records (average weight too small for the range of $\epsilon$ considered).

\subsection{Individual flip probabilities}
\label{sec:flip_probs-individual}

In order to reason about the arbitrariness of individual predictions, we define the \emph{flip probability} $q_{N,i}$ for a given data record $d_i$, $i \in \{1, \ldots, N\}$, as the proportion of models in the Rashomon set $R_N(\epsilon)$ for which the model prediction $f(x_i)$ differs from the Bayes-optimal prediction $f_\text{opt}(x_i) = \mathbf{1}\{p_i > 0.5\}$, or equivalently, the proportion of flip vectors for which $\theta_i = 1$:
\[q_{N,i} = \frac{|\theta \in R_N(\epsilon): \theta_i = 1|}{|R_N(\epsilon)|}.\]
As $N\rightarrow\infty$ for a given weight distribution $W$ and error tolerance $\epsilon$, flip probabilities become pairwise independent (Appendix~\ref{appendix:proofs}, Lemma~\ref{lemma:independence}), and the flip probability $q_i = \lim_{N\rightarrow\infty} q_{N,i}$ depends only on the weight $w_i$.  Thus we define the \emph{asymptotic flip probability} function $q(w)$ as the flip probability $q_i$ corresponding to a data record $d_i$ with weight $w_i = w$.  We then prove the following theorem (Appendix~\ref{appendix:proofs}, Theorem~\ref{theorem:flip_probs}):

\begin{theorem}[Asymptotic flip probabilities]\label{theorem:flip_probs_mainpaper}
Given the preliminaries and assumptions above, as $N\rightarrow\infty$, the flip probability corresponding to a data record with weight $w_i = w$ converges to
\[
q(w) = \frac{1}{1 + \exp(C(\epsilon)\, w)},
\]
where $C(\epsilon) = g^{-1}(\epsilon)$ and $g(C) = \int_0^1 \frac{w f(w)}{1+\exp(Cw)} \, dw$.
\end{theorem}

As a consequence of this theorem, for a Rashomon set $R_N(\epsilon)$ with $N$ large, we can obtain the flip probabilities for each individual, which we outline in Appendix~\ref{appendix:proofs}, Theorem~\ref{theorem:flip_probs}.
Computing these flip probabilities provides us with multiple pieces of valuable information about the Rashomon set.  First, we can use the flip probabilities to \emph{exactly} (in the large-sample limit) and \emph{efficiently} compute the \emph{average} over the entire Rashomon set of any metric (e.g., PPR, FPR, or TPR disparity) which can be decomposed as a linear function of the individual predictions, as shown in Appendix~\ref{appendix:flip_probs}.  This can help us better understand, without the need for computationally expensive random sampling, how much fairness we expect for a model drawn randomly from the Rashomon set, i.e., whether or not we will arrive at a reasonably fair model by optimizing solely for accuracy and not considering fairness.  Second, the flip probabilities $q_{N,i}$ are related to the size of the Rashomon set (Appendix~\ref{appendix:proofs}, Lemma~\ref{lemma:size}), and thus, as we show in Section~\ref{sec:rashomon_set_size}, the asymptotic size of the Rashomon set as $N\rightarrow\infty$ can be computed from the quantity $C(\epsilon)$ defined in Theorem~\ref{theorem:flip_probs_mainpaper}.

%Most critically, knowing the flip probabilities allows us to address the question raised in Section~\ref{sec:intro} of \emph{arbitrariness} of predictions. Many authors have pointed to the phenomenon of predictive multiplicity~\cite{marx2019}, where an individual can have different outcomes among different models in a Rashomon set, as a form of inequity through arbitrariness~\cite{marx2019,black2021leave}. By seeing the flip probabilities of any individual in the Rashomon set, we can see who is likely to be more and less susceptible to potentially arbitrary changes in outcome,
Most importantly, however, understanding flip probabilities helps us reason about arbitrariness of prediction in the Rashomon set as we can see who is likely to be more and less susceptible to potentially arbitrary changes in outcome. 
More precisely, the flip probability $q_{N,i}$ for a given individual is the probability that their prediction will differ from that of the Bayes-optimal model, across all models in the Rashomon set, and we note that $0 < q_{N,i} < \frac{1}{2}$ for all $i \in \{1,\ldots,N\}$. Therefore, individuals with $q_{N,i} \approx 0$ have consistent predictions across the Rashomon set, while individuals with $q_{N,i} \not\approx 0$ may receive either classification depending on which model happens to be drawn, i.e., their prediction is \emph{arbitrary}.
While we expect individuals with low-confidence Bayes-optimal probabilities $p_i \approx \frac{1}{2}$ to receive arbitrary predictions, and individuals with high-confidence probabilities $p_i \approx 0$ or $p_i \approx 1$ to receive consistent predictions, the question remains: how far from the decision boundary $p_i = \frac{1}{2}$ must an individual be for their predictions to be consistent?  We see in Figure~\ref{fig:main}(left) that the answer to this question differs across datasets and varies with $\epsilon$.

\ignore{Recall that in our idealized setting, we base the Rashomon set off of the Bayes-optimal model, which in addition to having underlying probabilities of the positive outcome $p_i$, also has a mapping of $X$ to $y$, i.e. a mapping of inputs to actual classification decisions, e.g. who is granted a loan and who is not.

This theorem allows us to calculate the probability that any person's classification outcome will change within the set of models in the Rashomon set. In other words, when we consider some individual $x_i$ within the data set, the expression above gives the probability that their prediction outcome would change in a random model within the Rashomon set. This is a powerful piece of information: understanding the individual flip probabilities of the Rashomon set can help us understand quantities about the overall distribution of the Rashomon set. For example, the average fairness of all models over the Rashomon set for various definitions of fairness (demographic parity, false positive rate parity, true positive rate parity). As we discuss further below in our takeaways, this can help us better understand what kind of fairness we can expect for a model drawn randomly from the Rashomon set-- i.e. what sort of fairness we can expect by optimizing solely for accuracy and not considering fairness.}

In addition, understanding individual flip probabilities within the Rashomon set can shed light on another source of inequity: certain demographic subgroups may have systematically higher flip probabilities than others, meaning that they are more likely to be exposed to arbitrary, inconsistent predictions, with potentially less reliable explanations %less justification %(and less trustworthy and reliable explanations) 
for the outcomes they receive~\cite{black2022model}.\footnote{As a caveat, these flip probabilities will not necessarily translate to who is most likely to get flipped in any given search for a less discriminatory algorithm (LDA), as LDA searches will typically restrict the class of models prior to searching, and thus will not exactly match random sampling from within the Rashomon set. However, it does let us understand who is most likely to get flipped in the largest possible Rashomon set $R_N(\epsilon)$.}
%\dn{(Is it worth noting that inconsistency can sometimes \emph{benefit} individuals or subgroups by avoiding algorithmic monoculture, for example, if an individual applies for loans from multiple lenders, and wants to maximize the probability that at least one lender gives them the loan? Or does this just confuse things?) \eb{I think just confused things!}}
We note that this is a separate form of group-level unfairness from the typical measures of statistical parity and error rate balance, since two groups may have equal positive prediction rates but very different flip probabilities (see Appendix~\ref{appendix:flip_probs} for an example).

\begin{figure}[t]
    \centering
    % First subfigure
    \begin{subfigure}[b]{0.32\textwidth}  % Adjust width to your needs
        \centering
        \includegraphics[width=\textwidth]{ 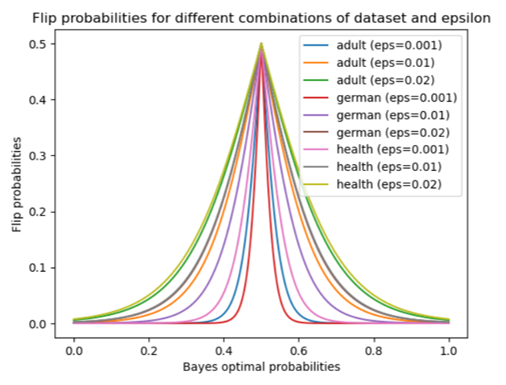}  % Replace with your image
        %\caption{}
    \end{subfigure}
    % Second subfigure
    \ignore{
    \begin{subfigure}[b]{0.32\textwidth}
        \centering
        \includegraphics[width=\textwidth]{ 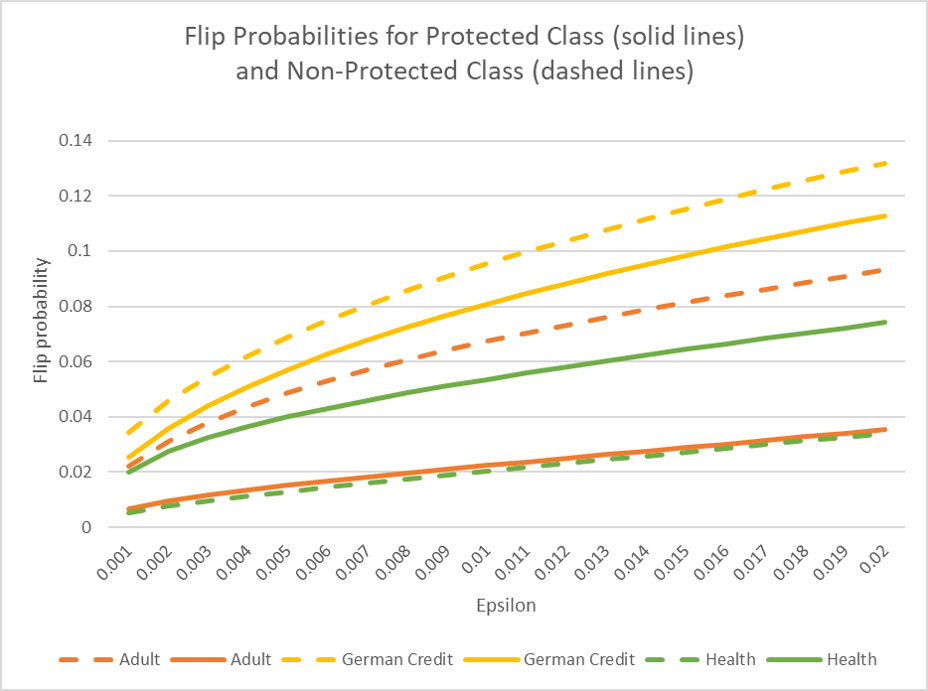}  % Replace with your image
        %\caption{
        %Here we show the disparities in flip probabilities between different demographic groups across the three datasets considered in this paper. Dashed lines represent historically disadvantaged group, and  color represents dataset. As we can see, disparities exist, though it is not always the historically disadvantaged group that experiences higher flip probabilities.}
    \end{subfigure}
    }
        \begin{subfigure}[b]{0.32\textwidth}  % Adjust width to your needs
        \centering
        \includegraphics[width=\textwidth]{ 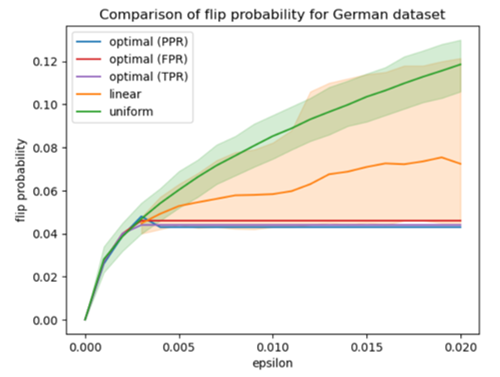}  % Replace with your image
    \end{subfigure}    
        \begin{subfigure}[b]{0.32\textwidth}  % Adjust width to your needs
        \centering
        \includegraphics[width=\textwidth]{ 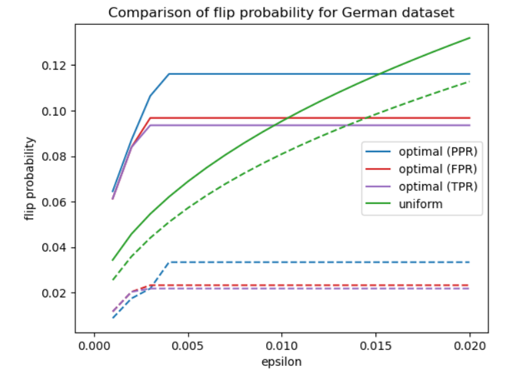}  % Replace with your image
    \end{subfigure}
     % Third subfigure
     \ignore{
    \begin{subfigure}[b]{0.32\textwidth}
        \centering
        \includegraphics[width=\textwidth]{ 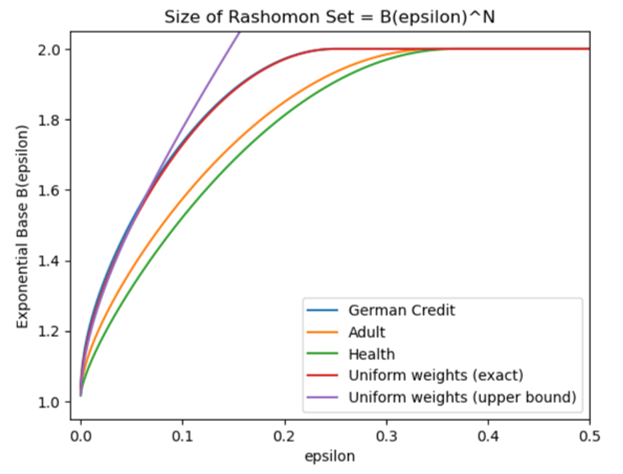}  % Replace with your image
       % \caption{
%        Growth of exponential base of the size of the Rashomon set as a function of Epsilon.}
    \end{subfigure}
    }
    \caption{Left: Flip probability $q_{N,i}$ as a function of the Bayes-optimal probability $p_i$--- in other words, how likely is an individual $i$ to experience a change of prediction among models in the Rashomon set as a function of their true probability that $y_i = 1$? We show results for 
    the German Credit, Adult, and Health datasets for $\epsilon\in \{0.001, 0.01, 0.02\}$, and see that there is large variation in flip probability distribution both as a function of dataset and $\epsilon$.
   \ignore{Center: Disparities in flip probabilities between different demographic groups ($A_i=1$ vs. $A_i=0$) across the three datasets considered in this paper. Dashed and solid lines represent the historically disadvantaged and advantaged groups respectively, and color represents dataset. As we can see, disparities exist, though it is not always the historically disadvantaged group that experiences higher flip probabilities.}
Center: Overall (population average) flip probability as a function of error tolerance $\epsilon$ for the German Credit dataset, for uniformly sampled models, linear models, and optimally fair models from the Rashomon set. For results for Adult and Health datasets, see Appendix~\ref{appendix:flip_probs}, Figure~\ref{fig:appendix_overall_flip_probs}.
%Center:
Right: Group average flip probability, comparison between protected group (solid lines) and non-protected group (dashed lines), for the German Credit dataset, as a function of the error tolerance $\epsilon$. Comparison of methods for optimizing PPR, FPR, and TPR (Section~\ref{sec:optimizing-fairness}) 
and uniform random sampling (Section~\ref{sec:sampling}), over the Rashomon set $R_N(\epsilon)$.  For results for Adult and Health datasets, see Appendix~\ref{appendix:flip_probs}, Figure~\ref{fig:appendix_stratified_flip_probs}.   
    \ignore{
    Right: Rashomon set size as a function of $\epsilon$  for Adult, German Credit, and Health datasets, and for uniformly distributed weights. Note that the German Credit and uniform weights curves coincide. The size of the Rashomon set is $|R_N(\epsilon)| = B(\epsilon)^N$, where the exponential base $B$ (plotted here) ranges between 1 (for $\epsilon = 0$) and 2 (for large $\epsilon$). We also separately plot $|R_N(\epsilon)|$ for each dataset in Appendix~\ref{appendix:rashomon_set_size-experiments}, Figure~\ref{fig:appendix_Rashomon_set_size}.}
    }
    \label{fig:main}
    \vspace{-1em}
\end{figure}

\subsection{Experiments on real data}\label{sec:flip_probs-experiments}

As in previous sections, we perform our experiments on the German Credit, Adult, and Health datasets. We first calculate the flip probabilities for all individuals in each dataset for varying values of $\epsilon$, and use these flip probabilities to perform four experiments. 

First, we graph the overall (population average) flip probability for all three datasets for models sampled uniformly at random from the Rashomon set $R_N(\epsilon)$ as a function of $\epsilon$, compared to sampling linear models from the Rashomon set (Section~\ref{sec:linear}) and the models that optimize PPR, FPR, and TPR over the Rashomon set (Section~\ref{sec:optimizing-fairness}).  These graphs are shown in Appendix~\ref{appendix:flip_probs}, Figure~\ref{fig:appendix_overall_flip_probs}.

Second, we use the flip probabilities to calculate the average fairness of the Rashomon set as a function of $\epsilon$ for all three datasets. We display the output in Appendix~\ref{appendix:flip_probs}, Figure~\ref{fig:flip_vs_sampling}. These differ from the results in Section~\ref{sec:intentional-experiments} since these are the average fairness of models across the \emph{entire} Rashomon set, not only from a sample of models, but we note the close correspondence between the sampled and entire-Rashomon-set results.

Third, in Figure~\ref{fig:main}(left) we turn to displaying empirical results about arbitrariness within the Rashomon set: we show how the chance of an individual experiencing a flip in their predictions in the Rashomon set (as a function of how close their Bayes-optimal probability $p_i$ is to the threshold of 0.5) differs across different datasets and values of error tolerance $\epsilon$. To do this, we compute the value of $C(\epsilon)$ for each dataset and $\epsilon$, and then compute the flip probability $q(w) = \frac{1}{1+\exp(C|2p-1|)}$ for a fine grid of $p$ values. 

Fourth, we show the disparities in average flip probability in the three datasets between protected and non-protected groups as a function of $\epsilon$, suggesting that some groups have systematically higher exposure than others to arbitrary, inconsistent decisions.  We compare uniform sampling to the models that optimize PPR, FPR, and TPR over the Rashomon set (as described in Section~\ref{sec:optimizing-fairness}).  Graphs for the German, Adult, and Health datasets are shown in Figure~\ref{fig:main}(right) and Appendix~\ref{appendix:flip_probs}, Figure~\ref{fig:appendix_stratified_flip_probs}. 

%We observe across datasets that optimizing for fairness and uniform sampling lead to large differences in who is flipped: while both optimization and uniform sampling approaches tend to flip individuals who are near the decision boundary and thus have lower weights $w_i$, the optimization approaches also tend to flip individuals who are from the group that is less represented in the dataset and thus have higher values $v_i$, because flipping one person's prediction has a larger impact on the group average for the group that is smaller in size.  

% In our datasets, the disadvantaged group (women for German and Adult, individuals over the age of 65 for Health) is also less represented. Thus, if the disadvantaged group is already flipped more than the advantaged group on average, because they tend to be closer to the decision boundary (as is the case for German and Health), optimizing for fairness will further exacerbate this disparity in who is receiving arbitrary and inconsistent predictions.  For example, 
% for German Credit, at $\epsilon=0.004$ (the point at which PPR disparity is eliminated), the flip proportion for uniform random sampling is balanced (6.2\% for women vs. 5.1\% for men), but the model that optimizes PPR demonstrates a substantial disparity in who is flipped (11.6\% for women vs. 3.3\% for men).  On the other hand, if the disadvantaged group is flipped substantially less than the advantaged group on average, because they tend to be farther from the decision boundary (as is the case for Adult) optimizing for fairness will instead mitigate this disparity. 
 
\subsection{Takeaways for policy and practice}
%Our experiments demonstrate that flip probabilities 

%\eb{We need experimental results for this--- what more experiments should we do? I think at the least, we should explicitly calculate the average fairness (not from sampling) for the datasets for various epsilons? But I'm struggling to see how this will be meaningfully different from the sampling result except that it's exact. Anything else we want to do with it?
%Second, I looked very quickly at the flip probabilities for the real datasets and I saw that indeed, the class that has more predicted probabilities in the 0.4-0.6 region had more flips, whether this was the advantaged or disadvantaged group. I wonder what more we can say here besides showing that we can see it.}
\begin{itemize}[left=7pt]
\ignore{
    \item \textbf{The average fairness of \emph{all} models on the Rashomon set is nowhere near as fair as the fairest model.} This further supports the finding that intentionally searching for fair models is important, because the fair models are not the norm within the set. This also provides another vantage point on recent work that shows that fair models in the Rashomon set can be hard to find~\cite{ben}.
}
\item\textbf{
Even a small error tolerance leads to a lot of individual flips.}
We observe that, for uniform sampling, the overall flip probability tends to be substantially higher than the error tolerance $\epsilon$. In Figure~\ref{fig:main}(center), 
%for $\epsilon = .02$, we see that 12\%, 7\%, and 5\% of predictions are flipped on average for the German, Adult, and Health datasets respectively.  
for $\epsilon = .02$, we see that 12\%, of predictions are flipped on average for the German dataset. In Appendix~\ref{appendix:flip_probs}, Figure~\ref{fig:appendix_overall_flip_probs}, we see that this trend continues, with  Adult and Health having 7\% and 5\% respectively. 
The overall flip probability for models that optimize fairness tends to be higher than the overall flip probability for uniform sampling, for lower $\epsilon$ values where the optimization method is not able to remove all of the disparity.  Once the disparity is removed, the flip probability for optimal models levels off, while the flip probability for uniform sampling continues to increase with $\epsilon$.
    
    %supports other work showing that, 
    %\item\textbf{Increasing error tolerance $\epsilon$ not only increases the \emph{number} of flips that occur, but \emph{which individuals} are likely to get flipped.}.
    \item \textbf{Increasing error tolerance $\epsilon$ not only increases the \emph{number} of flips that occur, but \emph{which individuals} are likely to get flipped: more ``certain'' cases get flipped with higher $\epsilon$.} As $\epsilon$ increases, individuals with true probability $p_i$ further and further away from a 50/50 coin toss, i.e., more ``certain'' cases of a positive or negative outcome get flipped. For example, as we see in Figure~\ref{fig:main}(left), in the German Credit dataset at low $\epsilon$ (red line) everyone who has a predicted Bayes probability below 0.4 or above 0.6 has near-zero chance of experiencing a flip in prediction, but at higher $\epsilon$ (brown line), we see that individuals with $p_i$ between 0.1 and 0.9 have a non-negligible chance of getting flipped.  Some prior work has suggested the normative view that individuals with higher certainty in their outcome should be flipped less often~\cite{black2021leave}---to the extent that this is true in certain contexts, it may be important to balance the flip probability over a threshold of certainty with the need to reduce outcome-based unfairness. We 
    also see \textbf{large differences in flip probabilities between datasets}: for the same $\epsilon$ value, an individual with a given true probability $p_i$ is much less likely to be flipped for German Credit as compared to Adult or Health.
    %\item \textbf{The distribution of flip probabilities change substantially across different datasets and error tolerance $\epsilon$.} However, there is significant concentration of flip probabilities around Bayes probability of 0.5 (i.e. individuals where the Bayes model gives a 50\% chance of being a member of each class). Interestingly, increasing error tolerance $\epsilon$ not only increases the \emph{number} of flips that occur, but \emph{which individuals} are likely to get flipped---in the German Credit dataset, for example, at low $\epsilon$ (red line) everyone who has a predicted Bayes probability below 0.4 or above 0.6 has zero chance of experiencing a flip in prediction, but at higher $\epsilon$ (brown line), we see that individuals up until 0.1 and 0.9 have a non-zero chance of getting flipped, and 

    \item \textbf{Asymmetries in the underlying model---e.g. uneven distributions of predicted probabilities across groups---lead to disparities in flip probabilities across demographic groups.} As we can see from Figure~\ref{fig:main}(right) and Appendix~\ref{appendix:flip_probs}, Figure~\ref{fig:appendix_stratified_flip_probs}, all three datasets have disparities in their average flip probabilities, though it is not always the disadvantaged group with a higher flip probability. Since the flip probability for uniform random sampling is a function of an individual's weight within the data--i.e., their distance from the threshold probability of $0.5$--the individuals who are flipped more often 
    %in the %overall 
    %Rashomon set 
    are those for whom
    the Bayes-optimal model is less certain of its prediction. In the Adult dataset, it is %actually 
    the advantaged group that has a higher density of true probabilities $p_i$ around 0.5, meaning that they are more likely to get flipped. In the German Credit and Health datasets, the disadvantaged group has a higher density of $p_i \approx 0.5$, and thus a higher flip probability. 

    \item  \textbf{We observe across datasets that optimizing for fairness and uniform sampling lead to large differences in who is flipped}: while both optimization and uniform sampling approaches tend to flip individuals who are near the decision boundary and thus have lower weights $w_i$, the optimization approaches also tend to flip individuals who are from the group that is less represented in the dataset and thus have higher values $v_i$, because flipping one person's prediction has a larger impact on the group average for the group that is smaller in size.  In our datasets, the disadvantaged group (women for German and Adult, individuals over the age of 60 for Health) is also less represented. Thus, if the disadvantaged group is already flipped more than the advantaged group on average, because they tend to be closer to the decision boundary (as is the case for German and Health), optimizing for fairness will further exacerbate this disparity in who is receiving arbitrary and inconsistent predictions.  For example, 
for German Credit, at $\epsilon=0.004$ (the point at which PPR disparity is eliminated), the flip proportion for uniform random sampling is balanced (6.2\% for women vs. 5.1\% for men), but the model that optimizes PPR demonstrates a substantial disparity in who is flipped (11.6\% for women vs. 3.3\% for men).  On the other hand, if the disadvantaged group is flipped substantially less than the advantaged group on average, because they tend to be farther from the decision boundary (as is the case for Adult), optimizing for fairness will instead mitigate this disparity. 
 
    %\dn{HERE!!! We will have one more takeaway - comparing amount of disparity in flip probability in optimal vs. uniform, based on the discussion above.}
\end{itemize}

\section{Rashomon Set Size and Error Tolerance}\label{sec:rashomon_set_size_and_error}
\ignore{
\begin{figure}[t]
    \centering
    % First subfigure
    \begin{subfigure}[b]{0.32\textwidth}  % Adjust width to your needs
        \centering
        \includegraphics[width=\textwidth]{  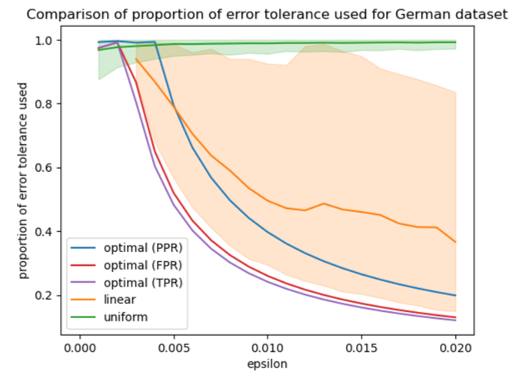}  % Replace with your image
    \end{subfigure}
    \hfill
    % Second subfigure
    \begin{subfigure}[b]{0.32\textwidth}
        \centering
        \includegraphics[width=\textwidth]{  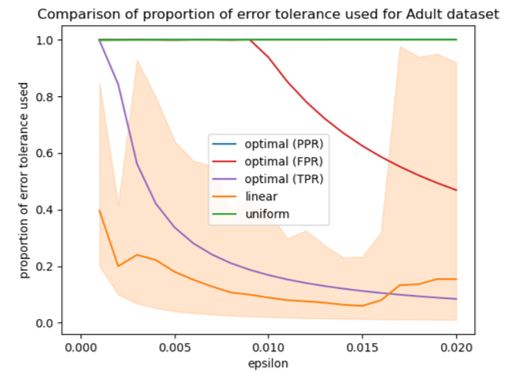}  % Replace with your image
    \end{subfigure}
     % Second subfigure
    \begin{subfigure}[b]{0.32\textwidth}
        \centering
        \includegraphics[width=\textwidth]{  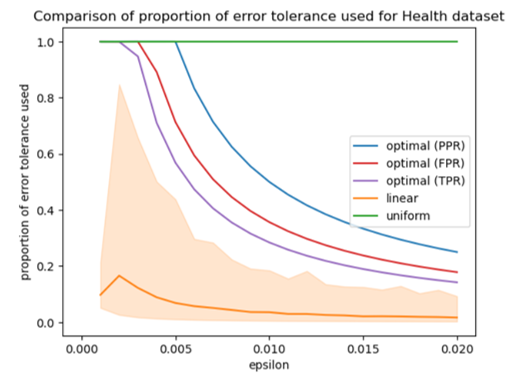}  % Replace with your image
    \end{subfigure}
    \caption{Proportion of error tolerance used, $\frac{\theta \cdot W_N}{N\epsilon}$, for the German, Adult, and Health datasets, as a function of the error tolerance $\epsilon$. Comparison of methods for optimizing PPR (Section~\ref{sec:optimizing-PPR}), optimizing FPR (Section~\ref{sec:optimizing-TPR and FPR}), optimizing TPR (Section~\ref{sec:optimizing-TPR and FPR}), uniform random sampling (Section~\ref{sec:sampling}), and sampling linear models (Section~\ref{sec:linear}) over the Rashomon set $R_N(\epsilon)$.}
    \label{fig:error_proportion}
    \vspace{-1em}
\end{figure}
}

\begin{figure}[t]
    \centering
\includegraphics[width=\textwidth]{  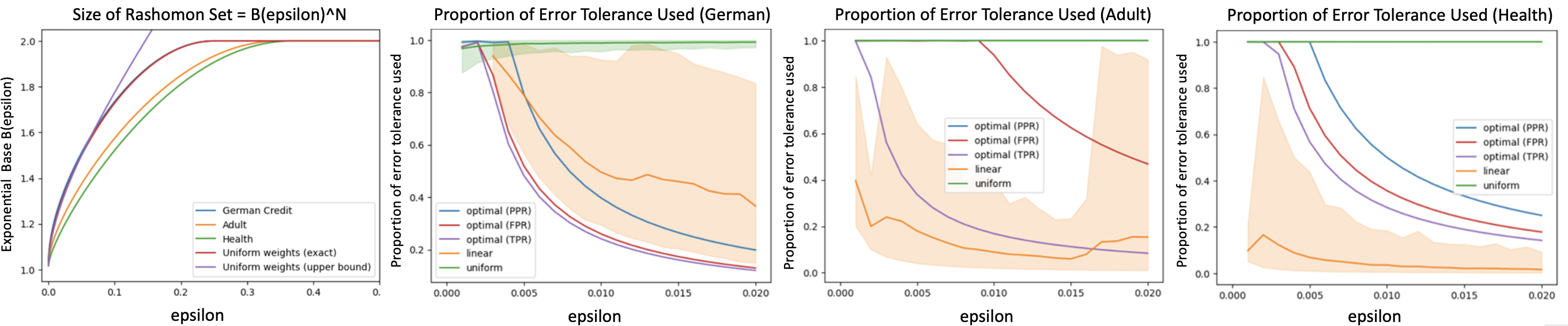}
    \caption{Left: Rashomon set size as a function of $\epsilon$  for Adult, German Credit, and Health datasets, and for uniformly distributed weights. Note that the German Credit and uniform weights curves coincide. The size of the Rashomon set is $|R_N(\epsilon)| = B(\epsilon)^N$, where the exponential base $B$ (plotted here) ranges between 1 (for $\epsilon = 0$) and 2 (for large $\epsilon$). We also separately plot $|R_N(\epsilon)|$ for each dataset in Appendix~\ref{appendix:rashomon_set_size-experiments}, Figure~\ref{fig:appendix_Rashomon_set_size}.
    Right three figures: Proportion of error tolerance used, $\frac{\theta \cdot W_N}{N\epsilon}$, for the German, Adult, and Health datasets, as a function of the error tolerance $\epsilon$. Comparison of methods for optimizing PPR (Section~\ref{sec:optimizing-PPR}), optimizing FPR (Section~\ref{sec:optimizing-TPR and FPR}), optimizing TPR (Section~\ref{sec:optimizing-TPR and FPR}), uniform random sampling (Section~\ref{sec:sampling}), and sampling linear models (Section~\ref{sec:linear}) over the Rashomon set $R_N(\epsilon)$.}
    \label{fig:error_proportion}
\end{figure}

In this section, we present results on the size of the Rashomon set and the distribution of how much of the error tolerance $\epsilon$ is used in the models of the Rashomon set. From these results, we suggest another set of takeaways--- that when a company sets out to do a search for a less discriminatory algorithm (LDA), they should choose the highest error tolerance possible. However, especially when relying on repeated random sampling as an LDA search method, they should make sure they are comfortable with having a model that uses all of the error tolerance provided.

\subsection{Rashomon set size}\label{sec:rashomon_set_size}

We derive an analytical expression for the asymptotic size of the Rashomon set, $|R_N(\epsilon)|$, as a function of the error tolerance $\epsilon$, as the number of data records $N$ that the Rashomon set is defined over goes to $\infty$.  We note that $|R_N(\epsilon)|$ also depends on the distribution of weights $f(w)$ and thus is dataset-dependent.  We provide the theorem here, with proof in Appendix~\ref{appendix:proofs}, Theorem~\ref{theorem:asymptotic_size}.

\begin{theorem}[Asymptotic size of Rashomon set]\label{theorem:asymptotic_size_mainpaper}
Given the preliminaries and assumptions in Section~\ref{sec:flip_probs-assumptions} above,
  let $R_N(\epsilon)$ denote the Rashomon set of models for error tolerance $\epsilon$ defined over data records $\langle d_1, \ldots, d_N\rangle$.  Then 
\[
\lim_{N\rightarrow\infty} \frac{\log |R_N(\epsilon)|}{N} = \log B(\epsilon),
\]
where $B(\epsilon) = \exp\left(\int_0^\epsilon C(x) dx\right)$, $C(\epsilon) = g^{-1}(\epsilon)$, and $g(C) = \int_0^1 \frac{w f(w)}{1+\exp(Cw)} \, dw$.
\end{theorem}

In other words, for large $N$, the size of the Rashomon 
set $|R_N(\epsilon)|$ converges (in the sense above) to 
$B(\epsilon)^N$.  Thus the size of the Rashomon set grows exponentially in $N$, the number of elements in the dataset, but the base of the exponential function $B$ is an increasing function of $\epsilon$.  For $\epsilon=0$ and $f(w)$ continuous, $|R_N(\epsilon)|=1$ regardless of $N$, so $B=1$.  For sufficiently large $\epsilon$, all $2^N$ flip vectors are in the Rashomon set, so $B=2$.  But the rate at which $B$ increases from 1 to 2 with $\epsilon$ will vary between datasets, depending on the distribution of weights $f(w)$, as we show in Figure~\ref{fig:error_proportion}(left). We give details on how to calculate $B(\epsilon)$, and therefore the size of the Rashomon set $B(\epsilon)^N$, in Appendix~\ref{appendix:proofs}, Theorem~\ref{theorem:asymptotic_size}. We also derive an exact value and an upper bound for $B(\epsilon)$ when the distribution of weights within the data records is uniform (Appendix~\ref{appendix:proofs}, Corollary~\ref{corr:asymptotic_size_uniform}).

As we discuss in our takeaways, although a company has no control over $N$, it does have control over $\epsilon$. As $\epsilon$ determines the base of the exponent $B$, this means that the size of the Rashomon set $|R_N(\epsilon)| = B(\epsilon)^N$ grows extremely quickly in $\epsilon$ as well.

\subsection{Usage of error within the Rashomon set}

We now show that as the number of data records $N$ over which the Rashomon set $R_N(\epsilon)$ is defined goes to infinity (i.e., as the dataset grows large), as long as the error tolerance $\epsilon$ is sufficiently small (less than half of the average weight $w_i$), the models in the Rashomon set will use almost all of the error tolerance. That is, the average accuracy of a model in the Rashomon set will converge to the accuracy of the Bayes-optimal model minus $\epsilon$. 

Let $\overline{acc}(R_N(\epsilon))$ denote the average accuracy of models in $R_N(\epsilon)$, and let $acc_N(\theta_0)$ denote the accuracy of the Bayes-optimal classifier $f(x_i) = \mathbf{1}\{p_i > 0.5\}$ for data records $\langle d_1,\ldots,d_N\rangle$. The \emph{average error tolerance used} is the difference 
$acc_N(\theta_0) - \overline{acc}(R_N(\epsilon))$, and must be less than or equal to $\epsilon$. \ignore{We show in Appendix~\ref{appendix:proofs}, Definition~\ref{def:avg_acc}, that the error tolerance used can be written in terms of the flip probabilities as $\frac{1}{N} \sum_{i=1\ldots N} w_i q_{N,i}$, for finite samples.  We can also consider the \emph{asymptotic average error tolerance used}, $\lim_{N\rightarrow\infty} (acc_N(\theta_0) - \overline{acc}(R_N(\epsilon)))$, and we show in Appendix~\ref{appendix:proofs}, Definition~\ref{def:avg_acc}, that it can be written in terms of the asymptotic flip probabilities $q(w)$ as 
$\int_0^1 w \, f(w) \, q(w) \, dw$.}  We now formally state the main result below, with proof in Appendix~\ref{appendix:proofs} (Theorem~\ref{theorem:all_error}):

\begin{theorem}[Asymptotic use of the entire error tolerance]\label{theorem:all_error_mainpaper}

Given the preliminaries and assumptions in Section~\ref{sec:flip_probs-assumptions} and the definitions above, let $R_N(\epsilon)$ denote the Rashomon set of models for error tolerance $\epsilon$ defined over data records $\langle d_1, \ldots, d_N\rangle$.  

Then as $N\rightarrow\infty$, the average error tolerance used by models in the Rashomon set converges to $\epsilon$:
\[
\lim_{N\rightarrow\infty} (acc_N(\theta_0) - \overline{acc}(R_N(\epsilon))) = \epsilon.
\]
\end{theorem}

This result implies that, for large $N$, there is a clear tradeoff between having a larger space of models to search over (since the size of the Rashomon set grows very rapidly with increasing $\epsilon$) and the accuracy of the models one might find with this search, since the vast majority of models in the Rashomon set have accuracy very close to the Bayes-optimal accuracy minus $\epsilon$.
While we are not typically interested in the Rashomon set for very large values of $\epsilon$ where the assumption that $\epsilon$ is less than half of the average weight would not hold, we note that in such cases the entire error tolerance would not be used. Instead, as $N$ becomes large, all flip probabilities would converge to 0.5, all or almost all of the $2^N$ possible flip vectors would be in the Rashomon set, and the average amount of error tolerance used would converge to half the average weight, which is less than $\epsilon$.

\subsection{Experiments on real data}

\subsubsection{Rashomon set size experiments}
\label{sec:rashomon_set_size-experiments}

Given that the size of the Rashomon set $|R_N(\epsilon)|$ can be written as $B(\epsilon)^N$, where the exponential base $B$ increases from 1 to 2 for increasing $\epsilon$, we plot the values of $B$ as a function of $\epsilon$ for the German Credit, Adult, and Health datasets in Figure~\ref{fig:error_proportion}(left).
As noted above, we also derived both the exact value and the upper bound of $B$ for uniformly distributed weights (Appendix~\ref{appendix:proofs}, Corollary~\ref{corr:asymptotic_size_uniform}), and we plot these in Figure~\ref{fig:error_proportion}(left) for comparison. For small $\epsilon$, the upper bound for uniformly distributed weights, $B(\epsilon) = \exp(\pi \sqrt{\epsilon/3})$, coincides closely with the exact values. \ignore{For the maximum $\epsilon$ value we consider, $\epsilon = 0.2$, we have $B=1.32$ for German Credit (and uniform weights), $B=1.22$ for Adult, and $B=1.17$ for Health.} We also plot the Rashomon set size $|R_N(\epsilon)|$ separately for the German Credit, Adult, and Health datasets in Appendix~\ref{appendix:rashomon_set_size-experiments}, Figure~\ref{fig:appendix_Rashomon_set_size}. While we do not yet have a way of computing the (reduced) Rashomon set size when restricting our search to the space of linear models ($L_2$-penalized logistic regression) as described in Section~\ref{sec:linear}, we can nevertheless examine what fraction of the sampled linear models are in the Rashomon set as a function of $\epsilon$.  This is shown for the German Credit, Adult, and Health Datasets in Appendix~\ref{appendix:rashomon_set_size-experiments}, Figure~\ref{fig:appendix_linear_proportion_in_Rashomon_set}.

\subsubsection{Use of error tolerance experiments}

Given that, as $N\rightarrow\infty$, we expect the entire error tolerance $\epsilon$ to be used by models in the Rashomon set, we examine whether this holds for the three real-world datasets as well.  In Figure~\ref{fig:error_proportion}(right), we plot the average proportion of the error tolerance used, $\frac{\theta \cdot W_N}{N\epsilon}$, for 950 flip vectors sampled uniformly at random from the Rashomon set $R_N(\epsilon)$, as described in Section~\ref{sec:sampling}, for each $\epsilon \in \{0.001, 0.002, \ldots, 0.02\}$.  We also plot the 95\% interquantile range for proportion of the error tolerance used, using the 2.5 and 97.5 percentiles of this distribution.  For comparison, we also plot for each dataset in Figure~\ref{fig:error_proportion}(right) the proportion of the error tolerance used when (i) optimizing PPR, TPR, FPR, and over the Rashomon set, as in Sections~\ref{sec:optimizing-PPR} and~\ref{sec:optimizing-TPR and FPR}, and (ii) searching over the set of linear models ($L_2$-penalized logistic regression) in the Rashomon set, as in Section~\ref{sec:linear}. 

\subsection{Takeaways for policy and practice}
\begin{itemize}[left=7pt]
    \item \textbf{Companies should set as large of an error tolerance as possible when defining the parameters for an LDA search, since increasing $\epsilon$ drastically increases the size of the Rashomon set, especially for smaller $\epsilon$.}
    For example, in the German Credit dataset ($N=1,000$), increasing $\epsilon$ from 0.005 to 0.02 moves the exponential base from 1.16 to 1.32 (Figure~\ref{fig:error_proportion} (left)), increasing the Rashomon set size from $5\times 10^{65}$ to $6 \times 10^{119}$ (Appendix~\ref{appendix:rashomon_set_size-experiments}, Figure~\ref{fig:appendix_Rashomon_set_size}).
    \item \textbf{Especially when using random sampling to search for fairer models, companies should expect the models they find to use up all of the error tolerance.} 
    In Figure~\ref{fig:error_proportion}(right), we see that, as expected from Theorem~\ref{theorem:all_error_mainpaper}, the average proportion of the error tolerance used by uniform random sampling over the entire Rashomon set is very close to 1 for all three datasets, and for the larger datasets (Adult and Health), even the 2.5 percentile of the distribution is virtually indistinguishable from 1.  Similarly, for the optimization approaches, all of the error tolerance is used until the entire disparity is mitigated; then the proportion of error tolerance used decreases as $\frac{1}{\epsilon}$ for larger $\epsilon$.
    Thus, while the widest \emph{possible} error tolerance should be used, the company should be ready to use a model within the outer limits of that tolerance.
    \item \textbf{Caveats: Searching within particular model classes may \emph{not} use up all of the error tolerance.} As we see from Figure~\ref{fig:error_proportion}(right), when restricting the search to linear models, the (non-exhaustive) set of linear models we found in the Rashomon set did not use up all of the error tolerance.  In fact, the average proportion of the error tolerance used by randomly sampled linear models is much smaller than 1: about 40\%, 15\%, and 2\% for German Credit, Adult, and Health datasets respectively.  
\end{itemize}

\section{Conclusion}\label{sec:conclusion}
We introduce key results that help us to understand the largest possible Rashomon set, from the average fairness of models within the Rashomon set, to the probability of individuals having their prediction changed across all models in the set, and the size of the Rashomon set. These results lead us to several takeaways: (1) it is critical to search for fair models within the Rashomon set (to be \emph{intentional} about fairness); (2) the arbitrariness of prediction within the Rashomon set changes drastically depending on the dataset and the error tolerance $\epsilon$; and (3) companies should think carefully about setting $\epsilon$ when searching for fairer models within the Rashomon set, balancing flexibility of the search with accuracy of the resulting models. We hope this work shows the importance of searching for fair models within the Rashomon set, and sheds light on how to balance fairness gains with risks of arbitrariness.
%Future work will further investigate the potential benefits for fairness, and drawbacks for arbitrariness, arising from the existence of the Rashomon set.

\begin{acks}
We gratefully acknowledge funding support from the National Science Foundation Program on Fairness in Artificial Intelligence in Collaboration with Amazon, grant IIS-2040898. Any opinions, findings, and conclusions or recommendations expressed in this material are those of the authors and do not necessarily reflect the views of the National Science Foundation or Amazon.
\end{acks}

\bibliographystyle{acm}
\bibliography{main}
\clearpage
\appendix
\section{Optimizing Fairness over the Rashomon Set}\label{appendix:optifairness}
In this section, we propose efficient algorithms to find (i) the exact fairest model in the Rashomon set as measured by positive prediction rate (PPR) disparity, and (ii) a model that is guaranteed to have false positive rate (FPR) or true positive rate (TPR) disparity within $\mathcal{O}(\frac{1}{N})$ of the fairest model in the Rashomon set.

\subsection{Optimizing for statistical parity}
\label{appendix:statdisp} 
In this sub-section, we propose an efficient knapsack solution (Algorithm~\ref{algorithm:knapsack}), 
to find the exact fairest model that minimizes disparities in positive prediction rate (PPR) over the Rashomon set $R_N(\epsilon)$, in $\mathcal{O}(N\log N)$ time.

Our first step is to derive expressions for the FPR and TPR disparities corresponding to a given flip vector $\theta$.  To do so, let $P_A$ and $P_B$ be the vectors of Bayes-optimal probabilities $p_i$ for subgroups $A$ (the \emph{protected class}, data records $d_i$ with sensitive attribute value $A_i=a$) and $B$ (the \emph{non-protected class}, data records $d_i$ with sensitive attribute value $A_i\ne a$) respectively.  Let $F^\text{opt} = \langle F_A^\text{opt}, F_B^\text{opt} \rangle$ denote the vector of Bayes-optimal binary predictions $f_\text{opt}(x_i)$, and let $F = \langle F_A,F_B\rangle$ denote the vector of binary predictions $f(x_i)$ corresponding to flip vector $\theta = \langle \theta_A, \theta_B\rangle$. 
We note that $F_A = F^\text{opt}_A \odot (1-\theta_A) + (1-F^\text{opt}_A) \odot \theta_A$, and $F_B = F^\text{opt}_B \odot (1-\theta_B) + (1-F^\text{opt}_B) \odot \theta_B$, where $\odot$ denotes element-wise product.

We can then define the positive prediction rate disparity as:
\begin{align*}
\text{disparity}_{PPR} &= \left|\mathbf{E}[f(x_i) \:|\: d_i\in A] - \mathbf{E}[f(x_i) \:|\: d_i \in B] \right| \\ &= \left|\mbox{Pr}(f(x_i)=1 \:|\: d_i\in A) - \mbox{Pr}(f(x_i)=1 \:|\: d_i\in B)\right| \\ &= \left|\frac{||F_A||_1}{|A|} - \frac{||F_B||_1}{|B|}\right| \\ &= \left|\frac{F_A \cdot \mathbf{1}}{|A|} - \frac{F_B \cdot \mathbf{1}}{|B|}\right| \\ &=\left|\frac{(F^\text{opt}_A \odot (1-\theta) + (1-F^\text{opt}_A) \odot \theta) \cdot \mathbf{1}}{|A|} - \frac{(F^\text{opt}_B \odot (1-\theta) + (1-F^\text{opt}_B) \odot \theta) \cdot \mathbf{1}}{|B|}\right| \\ &=\left|\frac{F^\text{opt}_A \cdot (1-\theta) + (1-F^\text{opt}_A) \cdot \theta}{|A|} - \frac{F^\text{opt}_B \cdot (1-\theta) + (1-F^\text{opt}_B) \cdot \theta}{|B|}\right|.
\end{align*}

As noted in Section~\ref{sec:optimizing-fairness}, we can express the minimization of PPR disparity over flip vectors $\theta$, subject to the constraint that  $\theta$ is in the Rashomon set $R_N(\epsilon)$, as a \emph{knapsack problem}, where each data record $d_i$ has a weight $w_i = |2p_i-1|$ and a value $v_i$, and $\theta_i=1$ corresponds to the inclusion of element $i$ in the knapsack, adding $w_i$ to the total weight and $v_i$ to the total value.  The 0-1 knapsack problem is then the constrained optimization with capacity $N\epsilon$: $\max \sum_i \theta_i v_i$ subject to $\theta_i \in \{0,1\}$ and $\sum_i \theta_i w_i \le N\epsilon$.

We now consider the expression for $v_i$, the change in PPR disparity when the prediction $f(x_i)$ is flipped (i.e., when $\theta_i$ is changed from 0 to 1).  Assume without loss of generality that subgroup $A$ has higher PPR, $\frac{||F_A||_1}{|A|} > \frac{||F_B||_1}{|B|}$. Then we see from the expression for PPR disparity above that flipping a prediction in group $A$ from 1 to 0, or flipping a prediction in group $B$ from 0 to 1, reduces disparity by $\frac{1}{|A|}$ or $\frac{1}{|B|}$ respectively, while other flips increase disparity.  To see this, for a data record $d_i \in A$ with $F^\text{opt}_i=1$,
\begin{align}
&v_i = \frac{F^\text{opt}_A \cdot (1-\theta) + (1-F^\text{opt}_A) \cdot \theta}{|A|}\Bigg|_{\theta=0, F^\text{opt}_A=1} - \frac{F^\text{opt}_A \cdot (1-\theta) + (1-F^\text{opt}_A) \cdot \theta}{|A|}\Bigg|_{\theta=1,F^\text{opt}_A=1} = \frac{1}{|A|}.
\end{align}
We note that $v_i$ for other cases can be calculated similarly. Thus we can write the value of element $i$ for the knapsack problem as $v_i = \frac{1}{|A|}$ if $d_i \in A$ and $F^\text{opt}_i = 1$, $v_i = \frac{1}{|B|}$ if $d_i \in B$ and $F^\text{opt}_i = 0$, and $v_i = 0$ otherwise.

We now make the key observation that enables our efficient knapsack algorithm: there are only two distinct values $v_i > 0$, $\frac{1}{|A|}$ and $\frac{1}{|B|}$ for group $A$ and $B$ respectively. Thus, the optimal knapsack solution will consist of the $k_A$ lowest-weight items from group $A$ and the $k_B$ lowest-weight items from group $B$, for some $k_A$ and $k_B$.

Optimal values of $k_A$ and $k_B$ (i.e., those values that most reduce the disparity) could be calculated by a $\mathcal{O}(N^2)$ brute-force search across all combinations of $k_A$ and $k_B$ that fit the capacity. However, through incremental search, that is, by keeping track of the optimal $k_B$ for a given $k_A$, one can incrementally update $k_B$ for $k_A-1$ by adding the remaining lowest weight $B$ items until the capacity is full, resulting in an incremental linear $\mathcal{O}(N)$ search. Thus the run time is dominated by the $\mathcal{O}(N\log N)$ sorting of items by weight. We present the algorithm below. 

%https://tex.stackexchange.com/questions/29816/algorithm-over-2-pages
\begin{algorithm} [H]                 
\caption{0-1 Knapsack Algorithm for minimizing $PPR$ disparity}\label{algorithm:knapsack}
\begin{algorithmic}[1]
\STATE Given: data records $D_N=(x_i,y_i)|_{i=1}^{N}$, Bayes-optimal (true) probabilities $P_N=p_i|_{i=1}^{N}$, capacity $N\epsilon$
\STATE Output: final disparity $fd$, flip vector $\theta$ \newline

\STATE Calculate the value $v_i$ of each record (either $\frac{1}{|A|}$, $\frac{1}{|B|}$, or 0) as described in the text above\;
\STATE Calculate the weight of each record $i$, $w_i=|2p_i-1|$\;
\STATE Calculate the initial disparity in the data, $id$\; \newline

\IF{$id = 0$ or $\epsilon = 0$}
\RETURN $\theta=0$, $fd=id$
\ENDIF \newline

\STATE Calculate weights $W_A$, sorted in ascending order, along with their indices $I_A$ for records in $A$ with $v_i > 0$\;
\STATE Calculate weights $W_B$, sorted in ascending order, along with their indices $I_B$ for records in $B$ with $v_i > 0$\; \newline

\STATE Calculate the maximum number of items in $A$ with $v_i>0$, $maxA$, that fit the capacity, adding items in ascending order of weight
\STATE Calculate the maximum number of items in $B$ with $v_i>0$, $maxB$, that fit the capacity along with the $maxA$ items, adding items in ascending order of weight \newline

\STATE Initialize the best value ($bestval$) to $\frac{maxA}{|A|} + \frac{maxB}{|B|}$
\STATE Initialize the best number of items in $A$ ($k_A$) to $maxA$ and the best number of items in $B$ ($k_B$) to $maxB$ \newline

\FOR{$a$ from $maxA-1$ to 0}
\STATE Remove the highest-weight item in $A$ with $v_i>0$, and add items in $B$ with $v_i>0$ until the capacity is filled, adding items in ascending order of weight. Let $b$ be the total number of items in $B$ with $v_i>0$ that have been added.
\IF{$\frac{a}{|A|} + \frac{b}{|B|}>bestval$}
\STATE Set $k_A=a$, $k_B=b$, $bestval=\frac{a}{|A|} + \frac{b}{|B|}$
\ENDIF
\ENDFOR \newline

\STATE Calculate optimal flip vector $\theta$ and final disparity $fd$, setting $\theta_i=1$ for the $k_A$ lowest-weight items in $A$ with $v_i>0$ 
 and the $k_B$ lowest-weight items in $B$ with $v_i>0$.
\RETURN $\theta$, $fd$
\end{algorithmic}
\end{algorithm}

\subsection{Optimizing for error rate balance}
\label{appendix:errbal}
In this sub-section, we propose an efficient, $\mathcal{O}(N\log N)$ fractional knapsack solution (Algorithm~\ref{algorithm:fractional_knapsack}), 
to find the model that minimizes disparities in false positive rate (FPR) or true positive rate (TPR) over the Rashomon set $R_N(\epsilon)$, to within $\mathcal{O}(\frac{1}{N})$ of the optimal disparity, in $\mathcal{O}(N\log N)$ time.

Our first step is to derive expressions for the FPR and TPR disparities corresponding to a given flip vector $\theta$.  To do so, as in Appendix~\ref{appendix:statdisp}, let $P_A$ and $P_B$ be the vectors of Bayes-optimal probabilities $p_i$ for subgroups $A$ (the \emph{protected class}, data records $d_i$ with sensitive attribute value $A_i=a$) and $B$ (the \emph{non-protected class}, data records $d_i$ with sensitive attribute value $A_i\ne a$) respectively.  Let $F^\text{opt} = \langle F_A^\text{opt}, F_B^\text{opt} \rangle$ denote the vector of Bayes-optimal binary predictions $f_\text{opt}(x_i)$, and let $F = \langle F_A,F_B\rangle$ denote the vector of binary predictions $f(x_i)$ corresponding to flip vector $\theta = \langle \theta_A, \theta_B\rangle$. 
We note that $F_A = F^\text{opt}_A \odot (1-\theta_A) + (1-F^\text{opt}_A) \odot \theta_A$, and $F_B = F^\text{opt}_B \odot (1-\theta_B) + (1-F^\text{opt}_B) \odot \theta_B$, where $\odot$ denotes element-wise product.

We can then define the false positive rate disparity and true positive rate disparity as:
\begin{align*}
\text{disparity}_{FPR} &= \left|\mathbf{E}[f(x_i) \:|\: d_i\in A, y_i = 0] - \mathbf{E}[f(x_i) \:|\: d_i \in B, y_i = 0] \right| \\
 &= \left|\mbox{Pr}(f(x_i)=1 \:|\: d_i\in A, y_i=0) - \mbox{Pr}(f(x_i)=1 \:|\: d_i\in B, y_i=0)\right| \\
&= \left| \frac{\mbox{Pr}(f(x_i)=1, y_i=0 \:|\: d_i\in A)}{\mbox{Pr}(y_i=0 \:|\: d_i\in A)} - \frac{\mbox{Pr}(f(x_i)=1, y_i=0 \:|\: d_i\in B)}{\mbox{Pr}(y_i=0 \:|\: d_i\in B)} \right| \\
&= \left|\frac{(1-P_A) \cdot F_A}{||1-P_A||_1} - \frac{(1-P_B) \cdot F_B}{||1-P_B||_1}\right|.
\end{align*}
\begin{align*}
\text{disparity}_{TPR} &= \left|\mathbf{E}[f(x_i) \:|\: d_i\in A, y_i = 1] - \mathbf{E}[f(x_i) \:|\: d_i \in B, y_i = 1] \right| \\
 &= \left|\mbox{Pr}(f(x_i)=1 \:|\: d_i\in A, y_i=1) - \mbox{Pr}(f(x_i)=1 \:|\: d_i\in B, y_i=1)\right| \\
&= \left| \frac{\mbox{Pr}(f(x_i)=1, y_i=1 \:|\: d_i\in A)}{\mbox{Pr}(y_i=1 \:|\: d_i\in A)} - \frac{\mbox{Pr}(f(x_i)=1, y_i=1 \:|\: d_i\in B)}{\mbox{Pr}(y_i=1 \:|\: d_i\in B)} \right| \\
&= \left|\frac{P_A \cdot F_A}{||P_A||_1} - \frac{P_B \cdot F_B}{||P_B||_1}\right|.
\end{align*}

We now compute the values $v_i$ (the change in disparity when the prediction $f(x_i)$ is flipped, i.e., when $\theta_i$ is changed from 0 to 1) for FPR and TPR respectively.  

For FPR, assume without loss of generality that subgroup $A$ has higher FPR, $\frac{(1-P_A) \cdot F_A}{||1-P_A||_1} > \frac{(1-P_B) \cdot F_B}{||1-P_B||_1}$.
Then flipping a prediction in group $A$ from 1 to 0, or flipping a prediction in group $B$ from 0 to 1, reduces the disparity by $\frac{1-p_i}{||1-P_A||_1}$ or $\frac{1-p_i}{||1-P_B||_1}$ respectively, while other flips increase disparity.  Thus we can write the value of element $i$ for the knapsack problem as $v_i = \frac{1-p_i}{||1-P_A||_1}$ if $d_i \in A$ and $F^\text{opt}_i = 1$, $v_i = \frac{1-p_i}{||1-P_B||_1}$ if $d_i \in B$ and $F^\text{opt}_i = 0$, and $v_i = 0$ otherwise.

For TPR, assume without loss of generality that subgroup $A$ has higher TPR, $\frac{P_A \cdot F_A}{||P_A||_1} > \frac{P_B \cdot F_B}{||P_B||_1}$. Then flipping a prediction in group $A$ from 1 to 0, or flipping a prediction in group $B$ from 0 to 1, reduces the disparity by $\frac{p_i}{||P_A||_1}$ or $\frac{p_i}{||P_B||_1}$ respectively, while other flips increase disparity.  Thus we can write the value of element $i$ for the knapsack problem as $v_i = \frac{p_i}{||P_A||_1}$ if $d_i \in A$ and $F^\text{opt}_i = 1$, $v_i = \frac{p_i}{||P_B||_1}$ if $d_i \in B$ and $F^\text{opt}_i = 0$, and $v_i = 0$ otherwise.

To minimize FPR or TPR disparity over the Rashomon set $R_N(\epsilon)$, we note that elements have more than two distinct values, so we cannot apply the solution for PPR above. Instead, we approximate the 0-1 knapsack problem with the fractional knapsack problem: $\max \sum_i \theta_i v_i$ subject to $\theta_i \in [0,1]$ and $\sum_i \theta_i w_i \le N\epsilon$.
The standard solution to the fractional knapsack, which requires $\mathcal{O}(N\log N)$ time, adds elements to the knapsack (setting $\theta_i = 1$) in descending order of their ratio $\frac{v_i}{w_i}$ until no further elements can be (fully) added, then adds a fraction of the next element ($0 < \theta_i < 1$) to fill the remaining capacity. Rather than adding the fractional element, we show that it would reduce disparity by an amount $\theta_i v_i$ that is $\mathcal{O}(\frac{1}{N})$.

To see this, we note for FPR disparity, for an individual in group $A$, that $\theta_i < 1$, and $v_i = \frac{1-p_i}{||1-P_A||_1} < \frac{1}{||1-P_A||_1} = \frac{1}{N \Pr(y_i = 0,\, d_i \in A)} = \mathcal{O}(\frac{1}{N})$. Therefore $\theta_i v_i = \mathcal{O}(\frac{1}{N})$.  The other cases, for TPR disparity and for group B, proceed similarly.  

Finally, since the fractional knapsack solution $\sum \theta_i v_i$ is an upper bound on the 0-1 knapsack solution, we know that our solution (excluding the fractional element) reduces disparity to within $\mathcal{O}(\frac{1}{N})$ of the optimal disparity.

Thus we propose an $\mathcal{O}(N \log N)$ fractional knapsack algorithm to find the final disparity and flip vector. The algorithm is a linear scan of values and weights, sorted by the ratio of their value to their weight. Thus the run time is dominated by the $\mathcal{O}(N\log N)$ sorting of items by their ratio of value to weight. We present the algorithm below. 
\clearpage
\begin{algorithm} [H]                 
\caption{Fractional Knapsack Algorithm for minimizing $FPR$ or $TPR$ disparity}\label{algorithm:fractional_knapsack}
\begin{algorithmic}[1]
\STATE Given: data records $D_N=(x_i,y_i)|_{i=1}^{N}$, Bayes-optimal (true) probabilities $P_N=p_i|_{i=1}^{N}$, capacity $N\epsilon$
\STATE Output: final disparity $fd$, flip vector $\theta$ \newline
\STATE Calculate the value $v_i$ of each record as described in the text above\;
\STATE Calculate the weight of each record $i$, $w_i=|2p_i-1|$\;
\STATE Calculate the initial disparity in the data, $id$\; \newline
\IF{$id = 0$ or $\epsilon = 0$}
\RETURN $\theta=0$, $fd=id$, $fracVal=0$
\ENDIF \newline
 \STATE Calculate weights $W$ and values $V$, along with their indices $I$, of records with $v_i > 0$, sorted by $\frac{v_i}{w_i}$ in descending order
\STATE Initialize record index variable $i$, total weight $totWei$, and total value $totVal$ to 0 \newline

\WHILE{$totVal < id$ \textbf{and} $totWei < Capacity$ \textbf{and} $i < len(I)$}
\STATE \# Attempt to add next element $i$ with $v_i>0$; note that elements are added in descending order of $\frac{v_i}{w_i}$.
    \IF{$totWei + w_i \le Capacity$}
        \STATE Add weight of element $i$ to total weight $totWei$
        \STATE Add value of element $i$ to total value $totVal$
        \STATE Increment $i$ by 1
    \ELSE
        \STATE Calculate the fractional value of element $i$ that would fill the entire capacity, $fracVal = \left(\frac{Capacity - totWei}{w_i}\right) v_i$
        \STATE \textbf{break}
    \ENDIF
\ENDWHILE \newline

\STATE Calculate flip vector $\theta$, setting $\theta_i=1$ for all elements $i$ added to the knapsack, excluding the fractional element.
\STATE Calculate final disparity $fd = id - totVal$
\STATE \# Note that $fracVal$ is an upper bound on the difference between $fd$ and the true optimal disparity.
\RETURN $fd$, $\theta$, $fracVal$
\end{algorithmic}
\end{algorithm}
\clearpage

\section{Gibbs sampling algorithm for uniform sampling from the Rashomon set}
\label{appendix:gibbs}

\begin{algorithm}[H]
\caption{Sampling Flip Vectors Uniformly at Random from the Rashomon Set}\label{algorithm:gibbs}
\begin{algorithmic}[1]
    \STATE Given: error tolerance $\epsilon$; number of data records $N$; weight vector $W_N = \langle w_1,\ldots,w_N\rangle$. Note that $w_i = |2p_i -1|$, where $p_i$ is the Bayes-optimal probability $\mbox{Pr}(y=1 \:|\: x=x_i)$ for data record $d_i$, $i\in \{1,\ldots,N\}$.
    \STATE Initialize $\Theta$ as an empty list
    \STATE Initialize $\theta = \theta_0$, where $\theta_0$ is the length-$N$ binary flip vector consisting of all zeros.
    \FOR{$t = 1$ to $T$}
        \FOR{$i$ in random permutation of $\{1, \dots, N\}$}
         \STATE Calculate current amount of error tolerance used, $E_{\text{current}} = \frac{\theta\cdot W_N}{N}$
            \STATE Calculate $\Pr(\theta_i = 1 \mid \theta_{-i})$:
            \IF{$E_{\text{current}} +(1-\theta_i)\frac{w_i}{N} \le \epsilon$}
                \STATE $\Pr(\theta_i = 1 \mid \theta_{-i}) = \frac{1}{2}$
            \ELSE
                \STATE $\Pr(\theta_i = 1 \mid \theta_{-i}) = 0$
            \ENDIF
            \STATE Sample $\theta_i$ from $\text{Bernoulli}(\Pr(\theta_i = 1 \mid \theta_{-i}))$
        \ENDFOR
        \IF{$t > B$ and $(t - B) \bmod K == 0$}
            \STATE Append $\theta$ to $\Theta$
        \ENDIF
    \ENDFOR
    \RETURN $\Theta$
\end{algorithmic}
\end{algorithm}

The definition of the Rashomon set $R_N(\epsilon)$ in Section~\ref{sec:prelim} suggests that a simple \emph{rejection sampling} approach could be used to draw models uniformly at random from the Rashomon set. That is, one could draw a binary flip vector $\theta \in \{0,1\}^N$ uniformly at random from the set of all $2^N$ possible flip vectors by drawing $\theta_i \sim \text{Bernoulli}(0.5)$ for all $i \in \{1 \ldots N\}$, and then keep only those vectors $\theta$ that are in the Rashomon set, i.e.,  with $\frac{\theta \cdot W_N}{N} \le \epsilon$.  The problem with this simple approach is that, as $N$ increases, the probability that $\theta$ is in the Rashomon set goes to 0. This can be seen from Appendix~\ref{appendix:proofs}, Theorem~\ref{theorem:asymptotic_size}: for $\epsilon$ less than half the average weight, $B(\epsilon) < 2$, and $\lim_{N \rightarrow \infty} \frac{|R_N(\epsilon)|}{2^N}$ = $\lim_{N \rightarrow \infty} \frac{B(\epsilon)^N}{2^N}$ = 0.

Thus we propose an alternative approach based on \emph{Gibbs sampling}~\cite{geman1984}. The key idea is to sequentially sample one element $\theta_i$ of the flip vector at a time, conditional on all the other elements $\theta_{-i}$.  While we do not have a closed form for the joint distribution of $\langle\theta_1,\ldots,\theta_N\rangle$ for $\theta \in R_N(\epsilon)$, computing the \emph{conditional} distribution of $\theta_i$ given $\theta_{-i}$ is straightforward. Let $\theta_{i=0} = \langle\theta_1,\ldots,\theta_{i-1},0,\theta_{i+1},\ldots,\theta_N\rangle$ and 
$\theta_{i=1} = \langle\theta_1,\ldots,\theta_{i-1},1,\theta_{i+1},\ldots,\theta_N\rangle$. Then we know that $\frac{\theta_{i=0}\cdot W_N}{N} \le 
\frac{\theta\cdot W_N}{N} \le
\frac{\theta_{i=1}\cdot W_N}{N}$. This implies that, if $\theta \in R_N(\epsilon)$ and $\theta_i = 1$, then $\theta_{i=0} $ and $\theta_{i=1}$ are both in the Rashomon set, so $\mbox{Pr}(\theta_i = 1 \:|\: \theta_{-i}) = \frac{1}{2}$. If $\theta \in R_N(\epsilon)$ and $\theta_i = 0$, then $\theta_{i=0} \in R_N(\epsilon)$, but we must check whether $\theta_{i=1} \in R_N(\epsilon)$, i.e., whether $\frac{\theta \cdot W_N + w_i}{N} \le \epsilon$. If so, then $\mbox{Pr}(\theta_i = 1 \:|\: \theta_{-i}) = \frac{1}{2}$, and if not, then 
$\mbox{Pr}(\theta_i = 1 \:|\: \theta_{-i}) = 0$.  

Given this simple and computationally efficient conditional sampling step, our Gibbs sampling approach starts with the zero vector $\theta_0$, which is guaranteed to be in the Rashomon set, and iteratively samples $\theta_i \sim \text{Bernoulli}(p)$, where $p = \mbox{Pr}(\theta_i = 1 \:|\: \theta_{-i})$ as described above, for each $i \in \{1,\ldots,N\}$.  To ensure uncorrelated samples from the joint distribution, we take one sample every 10 iterations (where one iteration includes resampling all $N$ elements of $\theta$ in randomly permuted order), after an initial burn-in period of 500 iterations.  For each dataset and each value of $\epsilon$ considered, we run 10,000 iterations of Gibbs sampling, resulting in 950 samples. 

Algorithm~\ref{algorithm:gibbs} presents the pseudocode for our Gibbs sampling approach, enabling us to sample length-$N$ binary flip vectors uniformly at random from the Rashomon set $R_N(\epsilon)$.
The sampling algorithm follows the idea~\cite{geman1984}, in which the Markov chain is the sequence of flip vectors $\theta^{(0)}, \theta^{(1)}, \dots, \theta^{(T)}$ generated as the algorithm progresses. Each $\theta^{(t)}$ is a point $\theta \in \{0,1\}^N$ that does not violate the Rashomon set constraint $\frac{\theta \cdot W_N}{N} \le \epsilon$, and thus $\theta^{(t)} \in R_N(\epsilon)$.

Specifically, we start by initializing the flip vectors $\Theta$ as an empty list. We then initialize flip vector $\theta = \theta_0$, the length-$N$ binary vector of zeros, which is guaranteed to be in the Rashomon set since $\frac{\theta_0 \cdot W_N}{N} = 0$.

Throughout $T=10,000$ iterations, where each iteration involves resampling each of the $N$ elements of $\theta$ in randomly permuted order, we keep track of the current amount of error tolerance used, $E_{\text{current}} = \frac{\theta \cdot W_N}{N}$, which can be done through incremental updates of $E_\text{current}$ whenever an element $\theta_i$ is modified.  To resample element $\theta_i$, we first compute $\Pr(\theta_i = 1 \mid \theta_{-i})$, the probability of $\theta_i = 1$ conditional on the current values of the other elements of $\theta$, and then draw $\theta_i \sim \text{Bernoulli}(\Pr(\theta_i = 1 \mid \theta_{-i}))$.  To compute $\Pr(\theta_i = 1 \mid \theta_{-i})$, we note that all flip vectors in the Rashomon set must be equally likely to be drawn.  Thus, if $\theta_{i=0}$ and $\theta_{i=1}$ are both in the Rashomon set, we know $\Pr(\theta_i = 1 \mid \theta_{-i}) = \frac{1}{2}$, while if $\theta_{i=0}$ is in the Rashomon set and $\theta_{i=1}$ is not, then 
$\Pr(\theta_i = 1 \mid \theta_{-i}) = 0$.  We note that $\theta_{i=0}$ will always be in the Rashomon set, as described in the main text.  To check if $\theta_{i=1}$ is in the Rashomon set, we must evaluate whether $\frac{\theta_{i=1} \cdot W_N}{N} = E_\text{current} + (1-\theta_i)\frac{w_i}{N}\le \epsilon$.

To ensure that each sampled flip vector $\theta$ is drawn independently from the joint distribution of $\langle\theta_1,\ldots,\theta_N\rangle$, we begin recording $\theta$ only after the number of iterations exceeds the burn-in period $B = 500$, and thereafter sample one value of $\theta$ every $K = 10$ iterations (i.e., at iterations 510, 520, \ldots). When the algorithm terminates, all recorded flip vectors are appended to $\Theta$, resulting in the final list of sampled vectors.

We see that each iteration requires stepping through the $\mathcal{O}(N)$ data records.  For each data record, we perform an $\mathcal{O}(1)$ check (whether or not the flip vector with $\theta_i=1$ is in the Rashomon set; note that we keep track of the current value of $\frac{\theta\cdot W_N}{N}$ throughout for computational efficiency) and an $\mathcal{O}(1)$ resampling of $\theta_i$ from either $\text{Bernoulli}(0.5)$ or $\text{Bernoulli}(0)$. Since the number of iterations is a fixed constant, this means that the overall runtime of the algorithm is $\mathcal{O}(N)$.

\section{Proofs of Theorems}\label{appendix:proofs}
In this section, we formally derive the theoretical results in the main paper, Sections~\ref{sec:flip_probs} and~\ref{sec:rashomon_set_size_and_error}.  Note that the order in which we derive these results is different than the order in which they are presented in the main paper, as many of the results build on each other.

Let $\langle d_1, d_2, \ldots \rangle$ denote an infinite sequence of data records drawn i.i.d. from distribution $D$, and let $D_N$ denote the subsequence $\langle d_1, d_2, \ldots, d_N \rangle$.  Assume that each $d_i = (x_i, y_i)$ where $x_i$ represents a set of predictor variables and $y_i$ is a binary outcome variable.  Let $p_i$ denote the Bayes-optimal probability, $p_i = \mbox{Pr}(y = 1 \:|\: x = x_i)$, and let $w_i$ denote the corresponding weight, $w_i = |2p_i - 1|$. We define the vectors $P_N=\langle p_1, p_2, \ldots, p_N \rangle$, and $W_N=\langle w_1, w_2, \ldots, w_N \rangle$. Moreover, let $P$ and $W$ be the distributions of Bayes-optimal probabilities and weights respectively, for data records drawn i.i.d. from $D$.  

Let $R_N(D_N,\epsilon)$ denote the largest possible Rashomon set of models for data records $\langle d_1, \ldots, d_N\rangle$.  Since $R_N$ can be computed using only the weights $W_N$, we can also write $R_N(P_N,\epsilon)$, $R_N(W_N, \epsilon)$, or simply $R_N(\epsilon)$ when the context (specifically, the weight vector $W_N$) is clear. Each distinct model in $R_N(\epsilon)$ represents a different binary classification of the data records $\langle d_1, \ldots, d_N\rangle$, $\langle f(x_1), \ldots, f(x_N) \rangle$, where $f(x_i) \in \{0,1\}$ for all $i \in \{1, \ldots, N\}$, and thus there are at most $2^N$ models in $R_N(\epsilon)$. Note that the classifier can be probabilistic, i.e., two data records with identical $x_i$ could have different $f(x_i)$ values. The \emph{Bayes-optimal classifier} (the classifier with the lowest expected 0/1 loss, or equivalently, the highest expected accuracy) is a deterministic function of the Bayes-optimal probabilities $p_i$: $f_\text{opt}(x_i) = 1$ if $p_i > 0.5$, and $f_\text{opt}(x_i) = 0$ otherwise. We represent each model in $R_N(\epsilon)$ by a binary \emph{flip vector} $\theta \in \{0,1\}^N$, where $\theta_i = 1$ if $f(x_i) \ne f_\text{opt}(x_i)$, and $\theta_i = 0$ if $f(x_i) = f_\text{opt}(x_i)$.

\begin{definition}[Accuracy of a model defined by a flip vector $\theta$]
\label{def:acc}
Let $D_N = \langle d_1,d_2,\ldots,d_N \rangle$ be data records drawn i.i.d. from distribution $D$ with corresponding Bayes-optimal probabilities $P_N = \langle p_1,p_2,\ldots,p_N \rangle$ and corresponding weights $W_N = \langle w_1,w_2,\ldots,w_N \rangle$, where $w_i = |2p_i-1|$.  The accuracy of a model with flip vector $\theta$ is
\begin{align*}
acc(\theta) &= \frac{1}{N} \sum_{i=1\ldots N} \left(p_i f(x_i) + (1-p_i) (1-f(x_i)) \right)\\
&= acc(\theta_0) + \frac{1}{N} \sum_{i=1\ldots N} \theta_i \left(
((1-p_i) - p_i) \mathbf{1}\{p_i > 0.5\})
+(p_i - (1-p_i)) \mathbf{1}\{p_i \le 0.5\})
\right) \\
&= acc(\theta_0) - \frac{1}{N} \sum_{i=1\ldots N} \theta_i |2p_i-1| \\
&= acc(\theta_0) - \frac{\theta \cdot W_N}{N},
\end{align*}
where $acc(\theta_0)$ is the accuracy of the Bayes-optimal classifier (and corresponding flip vector $\theta_0$ consisting of all zeros).
\end{definition}

\begin{definition}[Rashomon set]
The Rashomon set of models $R_N(\epsilon)$ for error tolerance $\epsilon$ defined over data records $\langle d_1,\ldots,d_N\rangle$ is defined as the set of all models with corresponding flip vectors $\theta \in \{0,1\}^N$ such that $acc(\theta) \ge acc(\theta_0) - \epsilon$.  Therefore, from Definition~\ref{def:acc} above, $R_N(\epsilon) = \{ \theta \in \{0,1\}^N : \frac{\theta \cdot W_N}{N} \le \epsilon \}$.
\end{definition}

\begin{definition}[Flip probability]
Let $R_N(\epsilon)$ denote the Rashomon set of models for error tolerance $\epsilon$ defined over data records $\langle d_1, \ldots, d_N \rangle$.
For a given data record $d_i$, $i \in \{1, \ldots, N\}$, the \emph{flip probability} $q_{N,i}$ is defined as the proportion of models in the Rashomon set for which
the model prediction $f(x_i)$ differs from the Bayes-optimal prediction $f_\text{opt}(x_i) = \mathbf{1}\{p_i > 0.5\}$, or equivalently, the proportion of flip vectors for which $\theta_i = 1$:
\[q_{N,i} = \frac{|\theta \in R_N(\epsilon): \theta_i = 1|}{|R_N(\epsilon)|}.\]
\end{definition}

\begin{lemma}[Relationship between flip probability, weight, and Rashomon set size]\label{lemma:size}
\[
q_{N,i} = \frac{\left|R_{N,-i}\left(\frac{N\epsilon-w_i}{N-1}\right)\right|}
{\left|R_{N,-i}\left(\frac{N\epsilon-w_i}{N-1}\right)\right| 
+ \left|R_{N,-i}\left(\frac{N\epsilon}{N-1}\right)\right|},
\]
where $R_{N,-i}(\epsilon)$ is the Rashomon set of models for error tolerance $\epsilon$ defined over the $N-1$ data records $\langle d_1, \ldots, d_{i-1}, d_{i+1}, \ldots, d_N\rangle$. 
\end{lemma}
\begin{proof}
We can rewrite the criterion for membership in the Rashomon set, $\frac{\theta \cdot W_N}{N} \le \epsilon$, as 
$\theta_{-i} \cdot W_{N,-i} + \theta_i w_i \le N\epsilon$, where $\theta_{-i}$ and $W_{N,-i}$ are the flip vector omitting element $i$ and the weight vector omitting element $i$ respectively. The numerator of the above expression, and the first term of the denominator, represent the flip vectors $\theta$ for which $\theta_i = 1$. To satisfy $\frac{\theta \cdot W_N}{N} \le \epsilon$ for these flip vectors, we must have $\theta_{-i} \cdot W_{N,-i} \le N\epsilon - w_i$, or 
$\frac{\theta_{-i} \cdot W_{N,-i}}{N-1} \le \frac{N\epsilon - w_i}{N-1}$.
The second term of the denominator represents the flip vectors $\theta$ for which $\theta_i = 0$. To satisfy $\frac{\theta \cdot W_N}{N} \le \epsilon$ for these flip vectors, we must have $\theta_{-i} \cdot W_{N,-i} \le N\epsilon$, or $\frac{\theta_{-i} \cdot W_{N,-i}}{N-1} \le \frac{N\epsilon}{N-1}$.
\end{proof}

\begin{lemma}[Asymptotic Pairwise Independence of Flip Probabilities]\label{lemma:independence}
Let $R_N(\epsilon)$ denote the Rashomon set of models for error tolerance $\epsilon$ defined over data records $\langle d_1, \ldots, d_N\rangle$.  For any $i,j \in \{1,\ldots,N\}$, $i \ne j$, as $N \rightarrow \infty$, the flip probability $q_j$ becomes independent of $\theta_i$. Specifically:
\[
\lim_{N\rightarrow \infty} (q_{N,j} \:|\: \theta_i = 1) = \lim_{N \rightarrow \infty} (q_{N,j} \:|\: \theta_i = 0).
\]

\end{lemma}
\begin{proof}

We consider two cases:

\textbf{Case 1: $w_j = 0$.}

When $w_j = 0$, flipping element $j$ has no impact on total error. Therefore $q_{N,j} = \frac{1}{2}$ regardless of $\theta_i$, so $\lim_{N\rightarrow \infty} (q_{N,j} \:|\: \theta_i = 1) = \lim_{N \rightarrow \infty} (q_{N,j} \:|\: \theta_i = 0) = \frac{1}{2}.$

\textbf{Case 2: $w_j \ne 0$.}

We compute the asymptotic odds ratio $\lim_{N\rightarrow\infty} \frac{q_{N,j}}{1-q_{N,j}}$ for both $\theta_i = 0$ and $\theta_i = 1$, and show that these two quantities are equal.

(1) For $\theta_i = 0$, using identical logic to Lemma~\ref{lemma:size} above, the odds ratio is:
\[
\frac{q_{N.j}}{1 - q_{N,j}} = 
\frac{\left|R_{N,-i,-j}\left(\frac{N\epsilon-w_j}{N-2}\right)\right|}
{\left|R_{N,-i,-j}\left(\frac{N\epsilon}{N-2}\right)\right|},
\]
where $\left|R_{N,-i,-j}(\epsilon)\right|$ is the size of the Rashomon set with error tolerance $\epsilon$ over the $N-2$ elements of $\langle d_1,\ldots,d_N \rangle$ excluding $d_i$ and $d_j$.

Next, we define $\log B(\epsilon) = \lim_{N \rightarrow \infty} \frac{\log |R_N(\epsilon)|}{N}$, and note that, since $R_N(\epsilon)$ has minimum size 1 (for $\epsilon=0$) and maximum size $2^N$ (for large $\epsilon$), $B(\epsilon) \in [1,2]$ for all $0 \le \epsilon \le 1$. We can also write $\log B(\epsilon) = \lim_{N \rightarrow\infty} \frac{\log|R_{N,-i,-j}(\epsilon)|}{N-2}$.  Taking the logarithm of the expression above and letting $N\rightarrow\infty$, we obtain:
\begin{align*}
\lim_{N \rightarrow \infty} \log \left(\frac{q_{N.j}}{1 - q_{N,j}}\right) &= 
\lim_{N \rightarrow \infty} \left(  \log \left| R_{N,-i,-j}\left(
\frac{N\epsilon-w_j}{N-2}
\right) \right| - \log \left| R_{N,-i,-j}\left(
\frac{N\epsilon}{N-2}
\right) \right| \right) \\
&= \lim_{N \rightarrow \infty} (N-2)\left(\log B\left( \frac{N\epsilon - w_j}{N-2} \right) - \log B\left( \frac{N\epsilon}{N-2} \right)\right).
\end{align*}

By the definition of the derivative of $\log B$, and noting that $-\frac{w_j}{N-2} \rightarrow 0$ as $N\rightarrow\infty$, we can write:
\[ \lim_{N\rightarrow\infty} \left(\log B\left( \frac{N\epsilon - w_j}{N-2} \right) - \log B\left( \frac{N\epsilon}{N-2} \right)\right)
= \lim_{N\rightarrow\infty} -\frac{w_j}{N-2}\left( \frac{d \log B}{d\epsilon} \right)\bigg|_{\frac{N\epsilon}{N-2}}, \]
and thus,
\begin{align*}
\lim_{N \rightarrow \infty} \log \left(\frac{q_{N.j}}{1 - q_{N,j}}\right) &= \lim_{N\rightarrow\infty} -w_j\left( \frac{d \log B}{d\epsilon} \right)\bigg|_{\frac{N\epsilon}{N-2}} \\
&= -w_j\left( \frac{d \log B}{d\epsilon} \right)\bigg|_\epsilon.
\end{align*}

(2) For $\theta_i = 1$, using identical logic to Lemma~\ref{lemma:size} above, the odds ratio is:
\[
\frac{q_{N.j}}{1 - q_{N,j}} = 
\frac{\left|R_{N,-i,-j}\left(\frac{N\epsilon-w_i-w_j}{N-2}\right)\right|}
{\left|R_{N,-i,-j}\left(\frac{N\epsilon-w_i}{N-2}\right)\right|}.
\]

Taking the logarithm and letting $N\rightarrow\infty$, we obtain:
\begin{align*}
\lim_{N \rightarrow \infty} \log \left(\frac{q_{N.j}}{1 - q_{N,j}}\right) &= 
\lim_{N \rightarrow \infty} \left(  \log \left| R_{N,-i,-j}\left(
\frac{N\epsilon-w_i-w_j}{N-2}
\right) \right| - \log \left| R_{N,-i,-j}\left(
\frac{N\epsilon-w_i}{N-2}
\right) \right| \right) \\
&= \lim_{N \rightarrow \infty} (N-2)\left(\log B\left( \frac{N\epsilon - w_i-w_j}{N-2} \right) - \log B\left( \frac{N\epsilon-w_i}{N-2} \right)\right).
\end{align*}

By the definition of the derivative of $\log B$, and noting that $-\frac{w_j}{N-2} \rightarrow 0$ as $N\rightarrow\infty$, we can write:
\[ \lim_{N\rightarrow\infty} \left(\log B\left( \frac{N\epsilon - w_i - w_j}{N-2} \right) - \log B\left( \frac{N\epsilon - w_i}{N-2} \right)\right)
= \lim_{N\rightarrow\infty} -\frac{w_j}{N-2}\left( \frac{d \log B}{d\epsilon} \right)\bigg|_{\frac{N\epsilon-w_i}{N-2}}, \]
and thus,
\begin{align*}
\lim_{N \rightarrow \infty} \log \left(\frac{q_{N.j}}{1 - q_{N,j}}\right) &= \lim_{N\rightarrow\infty} -w_j\left( \frac{d \log B}{d\epsilon} \right)\bigg|_{\frac{N\epsilon-w_i}{N-2}} \\
&= -w_j\left( \frac{d \log B}{d\epsilon} \right)\bigg|_\epsilon.
\end{align*}

Thus for both $\theta_i = 0$ and $\theta_i = 1$, $\lim_{N \rightarrow \infty} \log \left(\frac{q_{N.j}}{1 - q_{N,j}}\right) = -w_j\left( \frac{d \log B}{d\epsilon} \right)\bigg|_\epsilon$, and the proof is completed.
\end{proof}

\begin{lemma}[Functional form of flip probabilities]\label{lemma:funcform}
Let $D_N = \langle d_1,d_2,\ldots,d_N \rangle$ be data records drawn i.i.d. from distribution $D$ with corresponding Bayes-optimal probabilities $P_N = \langle p_1,p_2,\ldots,p_N \rangle$ and corresponding weights $W_N = \langle w_1,w_2,\ldots,w_N \rangle$, where $w_i = |2p_i-1|$.  Assume $w_i \sim W$ where distribution $W$ has pdf $f(w) > 0$ for $w \in [0,1]$.  Let $R_N(\epsilon)$ denote the Rashomon set of models for error tolerance $\epsilon$ defined over data records $\langle d_1, \ldots, d_N\rangle$.  Consider the flip probabilities $q_{N,i}$ corresponding to Rashomon set $R_N(\epsilon)$, and define $q_i = \lim_{N \rightarrow \infty} q_{N,i}$.  Then we can write:
\[
q_i = \frac{1}{1 + \exp(C w_i)},
\]
where $C$ is constant for a given $\epsilon$ and a given weight distribution $W$.
\end{lemma}
\begin{proof}
Consider any three data records $d_i$, $d_j$, and $d_k$ such that $w_i + w_j = w_k$.  We know that such triples exist as $N\rightarrow\infty$ because of the continuity and positivity of the weight distribution. 
Next we consider any pair of flip vectors $\theta, \theta' \in \{0,1\}^N$ such that $(\theta_i, \theta_j, \theta_k) = (1,1,0)$,
$(\theta'_i, \theta'_j, \theta'_k) = (0,0,1)$, and $\theta_l = \theta'_l$ for all $l \in \{1,2,\ldots,N\}$, $l \not\in \{i,j,k\}$.
The total weight is the same for both vectors: $\theta \cdot W_N = w_i + w_j + W_{rest} = w_k + W_{rest} = \theta' \cdot W_N$, where $W_{rest} = \sum_{l \in \{1,2,\ldots,N\}, l \not\in \{i,j,k\}} \theta_l w_l$, and thus either both flip vectors $\theta, \theta' \in R_N(\epsilon)$, or both flip vectors $\theta, \theta' \not\in R_N(\epsilon)$.  This means that, for flip vectors $\theta \in R_N(\epsilon)$, the probability that $(\theta_i, \theta_j, \theta_k) = (1,1,0)$ and the probability that $(\theta_i,\theta_j,\theta_k) = (0,0,1)$ are equal.  For $N\rightarrow\infty$, pairwise independence (Lemma~\ref{lemma:independence}) allows us to write these probabilities as $q_i(q_j)(1-q_k)$ and $(1-q_i)(1-q_j)q_k$ respectively, and thus $q_i(q_j)(1-q_k) = (1-q_i)(1-q_j)q_k$.  We can then rearrange terms and take the logarithm: 
\[ \log\left(\frac{q_i}{1-q_i}\right) + \log\left(\frac{q_j}{1-q_j}\right) = \log\left(\frac{q_k}{1-q_k}\right). \]

This establishes that for any data elements $d_i$, $d_j$, and $d_k$ with $w_i + w_j = w_k$, $h(w_i) + h(w_j) = h(w_k)$, where the function $h(w) = \log\left(\frac{q(w)}{1-q(w)}\right)$  The equation \( h(w_i) + h(w_j) = h(w_i + w_j) \) is Cauchy's functional equation for additive functions.

Moreover, the function $h(w)$ is monotone for $w\in (0,1)$.  To see this, we note that flip probability $q(w)$ is monotonically decreasing with $w$, since for every possible configuration $\theta_{-i}$, higher weight $w_i$ monotonically decreases (i.e., does not increase) the probability that $\theta_i = 1$ is in the Rashomon set, and does not change the probability that $\theta_i = 0$ is in the Rashomon set. Moreover, $\log\left(\frac{q}{1-q}\right)$ is increasing with $q$, so 
$h(w) = \log\left(\frac{q(w)}{1-q(w)}\right)$ is monotone.

Monotonicity of $h(w)$ is a sufficient condition for ensuring that the Cauchy functional equation does not have pathological (non-linear) solutions, and thus the only solutions are linear functions $h(w) = -C w$, where $C$ is a constant (for a given $W$ and $\epsilon$). 
Therefore, the log-odds of the flip probability is proportional to the weight:
\[
\log\left( \frac{q_i}{1 - q_i} \right) = -C w_i.
\]

Finally, the flip probability \( q_i \) can be expressed as:
  \[
  q_i = \frac{1}{1 + \exp(C w_i)}.
  \]
which completes the proof. (Note that the value of $C$, as a function of $\epsilon$, will be obtained in Corollary~\ref{corr:C} below.)
\end{proof}

\begin{lemma}[Asymptotic size of Rashomon set as a function of $C$]\label{lemma:asymptotic_size}
Let $D_N = \langle d_1,d_2,\ldots,d_N \rangle$ be data records drawn i.i.d. from distribution $D$ with corresponding Bayes-optimal probabilities $P_N = \langle p_1,p_2,\ldots,p_N \rangle$ and corresponding weights $W_N = \langle w_1,w_2,\ldots,w_N \rangle$, where $w_i = |2p_i-1|$.  Assume $w_i \sim W$ where distribution $W$ has pdf $f(w) > 0$ for $w \in [0,1]$.  Let $R_N(\epsilon)$ denote the Rashomon set of models for error tolerance $\epsilon$ defined over data records $\langle d_1, \ldots, d_N\rangle$.  Let $C(\epsilon)$ denote the constant in the asymptotic flip probability, $q_i = \lim_{N\rightarrow\infty} q_{N,i} = \frac{1}{1+\exp(C(\epsilon)w_i)}$, for error tolerance $\epsilon$. Then: 
\[
\lim_{N\rightarrow\infty} \frac{\log |R_N(\epsilon)|}{N} = \log B(\epsilon),
\]
where
\[
B(\epsilon) = \exp\left(\int_0^\epsilon C(x) dx\right).
\]
\end{lemma}
\begin{proof}
From Lemma~\ref{lemma:size}, we know
\[
q_i = \lim_{N\rightarrow\infty} \frac{\left|R_{N,-i}\left(\frac{N\epsilon-w_i}{N-1}\right)\right|}
{\left|R_{N,-i}\left(\frac{N\epsilon-w_i}{N-1}\right)\right| 
+ \left|R_{N,-i}\left(\frac{N\epsilon}{N-1}\right)\right|},
\]
where $R_{N,-i}(\epsilon)$ is the Rashomon set of models for error tolerance $\epsilon$ defined over the $N-1$ data records $\langle d_1, \ldots, d_{i-1}, d_{i+1}, \ldots, d_N\rangle$. And from Lemma~\ref{lemma:funcform}, we know that 
$q_i = \frac{1}{1 + \exp(C(\epsilon) w_i)}$.  Setting these quantities equal to each other, we have:

\[
\lim_{N\rightarrow\infty} \frac{\left|R_{N,-i}\left(\frac{N\epsilon-w_i}{N-1}\right)\right|}
{\left|R_{N,-i}\left(\frac{N\epsilon-w_i}{N-1}\right)\right| 
+ \left|R_{N,-i}\left(\frac{N\epsilon}{N-1}\right)\right|} = \frac{1}{1 + \exp(C(\epsilon) w_i)}.
\]

Inverting both sides:

\[
\lim_{N\rightarrow\infty} \left( 1 + \frac{\left|R_{N,-i}\left(\frac{N\epsilon}{N-1}\right)\right|}
{\left|R_{N,-i}\left(\frac{N\epsilon-w_i}{N-1}\right)\right|} \right) = 1 + \exp(C(\epsilon) w_i).
\]

Subtracting 1 and taking the logarithm of both sides:

\[
\lim_{N\rightarrow\infty}
\left( \log\left|R_{N,-i}\left(\frac{N\epsilon}{N-1}\right)\right| -
\log\left|R_{N,-i}\left(\frac{N\epsilon-w_i}{N-1}\right)\right|\right) = C(\epsilon) w_i.
\]

Dividing both sides by $w_i$:

\[
\lim_{N\rightarrow\infty}
\frac{ \log\left|R_{N,-i}\left(\frac{N\epsilon}{N-1}\right)\right| -
\log\left|R_{N,-i}\left(\frac{N\epsilon-w_i}{N-1}\right)\right|
}{(N-1)\frac{w_i}{N-1}} = C(\epsilon).
\]

By the definition of derivative, noting that $\frac{w_i}{N-1} \rightarrow 0$ as $N\rightarrow\infty$:

\[
\lim_{N\rightarrow\infty}
\left(\frac{1}{N-1}\right) \frac{d \log\left|R_{N,-i}(\epsilon)\right|
}{d\epsilon}\bigg|_\frac{N\epsilon}{N-1} = C(\epsilon).
\]

Equivalently, we can write:

\[
\lim_{N\rightarrow\infty}
\left(\frac{1}{N}\right) \frac{d \log\left|R_N(\epsilon)\right|
}{d\epsilon}\bigg|_\epsilon = C(\epsilon).
\]

Integrating both sides with respect to $\epsilon$:

\[
\lim_{N\rightarrow\infty}
\left(\frac{1}{N}\right) \log\left|R_N(\epsilon)\right| = \int_0^\epsilon C(x) \, dx + \text{constant}
\]

We know that the constant is 0 since, for $\epsilon=0$, we have:
\[
\lim_{N\rightarrow\infty} \left|R_N(\epsilon)\right| = 1.
\]

Thus we have:

\[
\lim_{N\rightarrow\infty}
\frac{\log\left|R_N(\epsilon)\right|}{N}  = \int_0^\epsilon C(x) \, dx.
\]

Finally, defining $B(\epsilon) =\exp\left(\int_0^\epsilon C(x) dx\right)$, we can write:

\[
\lim_{N\rightarrow\infty}
\frac{\log\left|R_N(\epsilon)\right|}{N}  = \log B(\epsilon).
\]
\end{proof}

\begin{definition}[Average accuracy of models in the Rashomon set]\label{def:avg_acc}
Let $D_N = \langle d_1,d_2,\ldots,d_N \rangle$ be data records drawn i.i.d. from distribution $D$ with corresponding Bayes-optimal probabilities $P_N = \langle p_1,p_2,\ldots,p_N \rangle$ and corresponding weights $W_N = \langle w_1,w_2,\ldots,w_N \rangle$, where $w_i = |2p_i-1|$.  Let $R_N(\epsilon)$ denote the Rashomon set of models for error tolerance $\epsilon$ defined over data records $\langle d_1, \ldots, d_N\rangle$, and consider the corresponding flip probabilities $q_{N,i}$.  Then the average accuracy of models $\theta \in R_N(\epsilon)$ can be written as:
\begin{align*}
\overline{acc}(R_N(\epsilon)) &= \frac{1}{|R_N(\epsilon)|} \sum_{\theta\in R_N(\epsilon)} acc(\theta) \\
&= acc_N(\theta_0) - \frac{1}{|R_N(\epsilon)|} \sum_{\theta\in R_N(\epsilon)} \frac{\theta \cdot W_N}{N} \\
&= acc_N(\theta_0) - \frac{1}{|R_N(\epsilon)|} \sum_{\theta\in R_N(\epsilon)} \sum_{i=1\ldots N}\frac{\theta_i w_i}{N} \\
&= acc_N(\theta_0) - \frac{1}{N} \sum_{i=1\ldots N} w_i \frac{1}{|R_N(\epsilon)|} \sum_{\theta\in R_N(\epsilon)} \theta_i \\
&= acc_N(\theta_0) - \frac{1}{N} \sum_{i=1\ldots N} w_i q_{N,i}, \\
\end{align*}
where $acc_N(\theta_0)$ is the Bayes-optimal accuracy for data elements $\langle d_1,\ldots, d_N\rangle$.

Further, assume $w_i \sim W$ where distribution $W$ has pdf $f(w) > 0$ for $w \in [0,1]$.  Then we can write the asymptotic average accuracy as:
\[
\lim_{N\rightarrow\infty} \overline{acc}(R_N(\epsilon)) =  acc(\theta_0) - \int_0^1 w \, q(w) \, f(w) \, dw,
\]
where $acc(\theta_0)$ is the asymptotic Bayes-optimal accuracy, $acc(\theta_0)=\lim_{N\rightarrow\infty} acc_N(\theta_0)$, and $q(w)$ denotes the asymptotic flip probability $q_i = \lim_{N\rightarrow\infty} q_{N,i}$ corresponding to an element $i$ with weight $w$.

\end{definition}

\begin{theorem}[Asymptotic use of the entire error tolerance]\label{theorem:all_error}

Let $D_N = \langle d_1,d_2,\ldots,d_N \rangle$ be data records drawn i.i.d. from distribution $D$ with corresponding Bayes-optimal probabilities $P_N = \langle p_1,p_2,\ldots,p_N \rangle$ and corresponding weights $W_N = \langle w_1,w_2,\ldots,w_N \rangle$, where $w_i = |2p_i-1|$.  Assume $w_i \sim W$ where distribution $W$ has pdf $f(w) > 0$ for $w \in [0,1]$. Let $R_N(\epsilon)$ denote the Rashomon set of models for error tolerance $\epsilon$ defined over data records $\langle d_1, \ldots, d_N\rangle$, and assume that $\epsilon$ is less than half of the average weight, i.e., $\int_{0}^{1} w \, f(w) \, dw > 2\epsilon$. Let $q(w)$ denote the asymptotic flip probability $q_i = \lim_{N\rightarrow\infty} q_{N,i}$ corresponding to an element $i$ with weight $w$.

Then as $N\rightarrow\infty$, the average error tolerance used by models in the Rashomon set converges to $\epsilon$:
\[
acc(\theta_0) - \lim_{N\rightarrow\infty} \overline{acc}(R_N(\epsilon)) = \int_{0}^{1} w \, q(w) \, f(w) \, dw = \epsilon.
\]
\end{theorem}

\begin{proof}

As every flip vector $\theta$ in the Rashomon set $R_N(\epsilon)$ must satisfy $\sum_{i=1}^N \frac{\theta_i w_i}{N} \leq \epsilon$, we have:

\begin{align*}
\frac{1}{|R_N(\epsilon)|} \sum_{\theta\in R_N(\epsilon)} \sum_{i=1\ldots N}\frac{\theta_i w_i}{N} &=
\frac{1}{N} \sum_{i=1\ldots N} w_i \frac{1}{|R_N(\epsilon)|} \sum_{\theta\in R_N(\epsilon)} \theta_i \\
&= \frac{1}{N} \sum_{i=1\ldots N} w_i q_{N,i} \\ 
&\le \epsilon.
\end{align*}

As $N\rightarrow\infty$, we can replace the summation with the integral:
\[
\int_{0}^{1} w \, q(w) \, f(w) \, dw \le \epsilon.
\]

Since $q(w) = \frac{1}{1 + \exp(C(\epsilon) \, w)} \) from Lemma~\ref{lemma:funcform}, we have:

\[
\int_{0}^{1} \frac{w}{1 + \exp(C(\epsilon) \, w)} \, f(w) \, dw \leq \epsilon.
\]

Next, the assumption $\int_{0}^{1} w \, f(w) \, dw > 2\epsilon$ ensures that $C(\epsilon) > 0$.  To see this, we first note that the expected error $\int_{0}^{1} w \, q(w) \, f(w) \, dw$ is a monotonically decreasing function of $C(\epsilon)$. Then if we assume $C(\epsilon) \le 0$, the flip probability becomes $q(w) \geq \frac{1}{2}$, leading to an average error of at least $\frac{1}{2} \int_{0}^{1} w \, f(w) \, dw > \epsilon$, which contradicts the fact above that the expected error can be at most $\epsilon$.

Now consider any $\delta > 0$ with $0 < \delta < \epsilon$. 
We apply Lemma~\ref{lemma:asymptotic_size} to compute the asymptotic sizes of the Rashomon sets at $\epsilon$ and $\epsilon - \delta$ respectively:

\[
\lim_{N\rightarrow\infty} \frac{\log|R_N(\epsilon)|}{N} = \int_{0}^{\epsilon} C(x) \, dx.
\]

and

\[
\lim_{N\rightarrow\infty} \frac{\log|R_N(\epsilon-\delta)|}{N} = \int_{0}^{\epsilon-\delta} C(x) \, dx.
\]

Thus we can write

\[
\lim_{N\rightarrow\infty}\frac{1}{N} \log\frac{|R_N(\epsilon)|}{|R_N(\epsilon - \delta)|} =  \left( \int_{0}^{\epsilon} C(x) \, dx - \int_{0}^{\epsilon - \delta} C(x) \, dx \right) = \int_{\epsilon - \delta}^{\epsilon} C(x) \, dx .
\]

Since $C(x) > 0$ for $x \le \epsilon$, we know that
$ \int_{\epsilon - \delta}^{\epsilon} C(x) \, dx > 0$, and $\exp\left(\int_{\epsilon - \delta}^{\epsilon} C(x) \, dx\right) > 1$. Therefore we have:

\[
\lim_{N\rightarrow\infty} \frac{|R_N(\epsilon)|}{|R_N(\epsilon - \delta)|} = \lim_{N\rightarrow\infty} \left( \exp\left( \int_{\epsilon - \delta}^{\epsilon} C(x) \, dx \right) \right)^N \rightarrow \infty.
\]

This implies that for large $N$, the number of flip vectors (models) with total error in the interval  $[\epsilon - \delta, \epsilon]$ dominates the Rashomon set. The proportion of flip vectors with error less than $\epsilon - \delta$ becomes negligible. Since almost all flip vectors in $R_N(\epsilon)$ have errors between $\epsilon - \delta$ and $\epsilon$, the asymptotic expected error $\int_{0}^{1} w \, q(w) \, f(w) \, dw$ is greater than or equal to $\epsilon - \delta$. Because $\delta > 0$ is arbitrary and can be made as small as desired, we have:

\[
\int_{0}^{1} w \, q(w) \, f(w) \, dw \ge \epsilon - \delta \quad \text{for all} \, \delta > 0.
\]

Combining with the initial inequality, we have:

\[
\epsilon - \delta \le \int_{0}^{1} w \, q(w) \, f(w) \, dw \le \epsilon \quad \text{for all} \, \delta > 0,
\]

and thus

\[
\int_{0}^{1} w \, q(w) \, f(w) \, dw = \epsilon,
\]

which completes the proof.
\end{proof}

\begin{corollary}[Value of $C$]\label{corr:C}
Let $D_N = \langle d_1,d_2,\ldots,d_N \rangle$ be data records drawn i.i.d. from distribution $D$ with corresponding Bayes-optimal probabilities $P_N = \langle p_1,p_2,\ldots,p_N \rangle$ and corresponding weights $W_N = \langle w_1,w_2,\ldots,w_N \rangle$, where $w_i = |2p_i-1|$.  Assume $w_i \sim W$ where distribution $W$ has pdf $f(w) > 0$ for $w \in [0,1]$. Let $R_N(\epsilon)$ denote the Rashomon set of models for error tolerance $\epsilon$ defined over data records $\langle d_1, \ldots, d_N\rangle$, and assume that $\epsilon$ is less than half of the average weight, i.e., $\int_{0}^{1} w \, f(w) \, dw > 2\epsilon$.

Then the value of $C(\epsilon)$ in Lemmas~\ref{lemma:funcform} and~\ref{lemma:asymptotic_size} can be written as $C(\epsilon) = g^{-1}(\epsilon)$, where $g(C) = \int_0^1 \frac{w f(w)}{1+\exp(Cw)} \, dw$.
\end{corollary}
\begin{proof}
From Theorem~\ref{theorem:all_error}, we have
\[
\int_{0}^{1} w \, q(w) \, f(w) \, dw = \epsilon,
\]
and from Lemma~\ref{lemma:funcform}, we have
\[
q(w) = \frac{1}{1+\exp(C(\epsilon)\, w)}.
\]

Putting these together with the function $g$ defined above, we have 
\[
g(C(\epsilon)) = \int_0^1 \frac{w f(w)}{1+\exp(C(\epsilon)\, w)} \, dw = \epsilon,
\]
or equivalently, $C(\epsilon) = g^{-1}(\epsilon)$, and the proof is completed.
\end{proof}

\begin{corollary}[Value of $C$ for uniformly distributed weights]\label{corr:C_uniform}
Let $D_N = \langle d_1,d_2,\ldots,d_N \rangle$ be data records drawn i.i.d. from distribution $D$ with corresponding Bayes-optimal probabilities $P_N = \langle p_1,p_2,\ldots,p_N \rangle$ and corresponding weights $W_N = \langle w_1,w_2,\ldots,w_N \rangle$, where $w_i = |2p_i-1|$.  Assume $w_i \sim \text{Uniform}[0,1]$.   Let $R_N(\epsilon)$ denote the Rashomon set of models for error tolerance $\epsilon$ defined over data records $\langle d_1, \ldots, d_N\rangle$, and assume that $\epsilon$ is less than half of the average weight, i.e., $\int_{0}^{1} w \, f(w) \, dw > 2\epsilon$.

Then the value of $C(\epsilon)$ in Lemmas~\ref{lemma:funcform} and~\ref{lemma:asymptotic_size} can be written as $C(\epsilon) = g^{-1}(\epsilon)$, where:

\begin{align*}
g(C) &= \int_0^1 \frac{w}{1+\exp(Cw)} \, dw \\
&= \frac{12\, \text{Li}_2(-e^{-C}) - 12 C \log(e^{-C}+1) + \pi^2}{12 C^2}.
\end{align*}

Moreover, $C(\epsilon) < \frac{\pi}{\sqrt{12 \epsilon}}$,
and $C(\epsilon) \approx \frac{\pi}{\sqrt{12 \epsilon}}$ for small $\epsilon$.
\end{corollary}

\begin{proof}
Given $w_i \sim \text{Uniform}[0,1]$, we know that the pdf $f(w)=1$ for $w\in [0,1]$.  Plugging this into the result of Corollary~\ref{corr:C}, we obtain 

\[
\int_0^1 \frac{w}{1+\exp(Cw)} \, dw = \epsilon.
\]

The integral can be computed as:
\[ \int_0^1 \frac{w}{1+\exp(Cw)} \, dw 
= \frac{12\, \text{Li}_2(-e^{-C}) - 12 C \log(e^{-C}+1) + \pi^2}{12 C^2},
\]
where $\text{Li}_2$ is the dilogarithmic (Spence's) function.  Since the first two terms of the rhs are negative for all $C$, $\frac{\pi^2}{12 C^2} > \epsilon$, and thus $C < \frac{\pi}{\sqrt{12\epsilon}}$.
As $\epsilon \rightarrow 0$, $C$ becomes large, and the first two terms of the rhs go to 0 from below.  Thus we have $\frac{\pi^2}{12 C^2} \approx \epsilon$, and $C \approx \frac{\pi}{\sqrt{12\epsilon}}$.
\end{proof}

\begin{theorem}[Asymptotic flip  probabilities]\label{theorem:flip_probs}
Let $D_N = \langle d_1,d_2,\ldots,d_N \rangle$ be data records drawn i.i.d. from distribution $D$ with corresponding Bayes-optimal probabilities $P_N = \langle p_1,p_2,\ldots,p_N \rangle$ and corresponding weights $W_N = \langle w_1,w_2,\ldots,w_N \rangle$, where $w_i = |2p_i-1|$.  Assume $w_i \sim W$ where distribution $W$ has pdf $f(w) > 0$ for $w \in [0,1]$.  Let $R_N(\epsilon)$ denote the Rashomon set of models for error tolerance $\epsilon$ defined over data records $\langle d_1, \ldots, d_N\rangle$, and assume that $\epsilon$ is less than half of the average weight, i.e., $\int_{0}^{1} w \, f(w) \, dw > 2\epsilon$.  Consider the flip probabilities $q_{N,i}$ corresponding to Rashomon set $R_N(\epsilon)$, and define $q_i = \lim_{N \rightarrow \infty} q_{N,i}$.  Then
\[
q_i = \frac{1}{1 + \exp(C(\epsilon)\, w_i)},
\]
where $C(\epsilon) = g^{-1}(\epsilon)$ and $g(C) = \int_0^1 \frac{w f(w)}{1+\exp(Cw)} \, dw$.
\end{theorem}
\begin{proof}
The statement follows immediately from Lemma~\ref{lemma:funcform}, which gives the functional form of $q_i$, and Corollary~\ref{corr:C}, which gives the expression for $C$. 
\end{proof}
\begin{remark}
As a consequence of this theorem, for a Rashomon set $R_N(\epsilon)$ with $N$ large, we can obtain the flip probabilities for each individual
in two steps: (1) calculate the value of $C$ for the given weight distribution $W$ and error tolerance $\epsilon$; and (2) compute $q_i = \frac{1}{1+\exp(Cw_i)}$ for each individual $i$.  To calculate $C$ if the pdf $f(w)$ of the weight distribution $W$ is known, we solve the integral equation $g(C) = \int_0^1 \frac{w f(w)}{1+\exp(Cw)} \, dw = \epsilon$.  Alternatively, given a finite dataset of size $N$, we estimate the true weight distribution $W$ using the empirical distribution $W_N$, 
and thus solve the equation $\frac{1}{N} \sum_{i=1\ldots N} \frac{w_i}{1+\exp(Cw_i)} = \epsilon$. In either case, we note that the lhs decreases monotonically with $C$, allowing an efficient solution by binary search.
\end{remark}

\begin{theorem}[Asymptotic size of Rashomon set]\label{theorem:asymptotic_size}
Let $D_N = \langle d_1,d_2,\ldots,d_N \rangle$ be data records drawn i.i.d. from distribution $D$ with corresponding Bayes-optimal probabilities $P_N = \langle p_1,p_2,\ldots,p_N \rangle$ and corresponding weights $W_N = \langle w_1,w_2,\ldots,w_N \rangle$, where $w_i = |2p_i-1|$.  Assume $w_i \sim W$ where distribution $W$ has pdf $f(w) > 0$ for $w \in [0,1]$.  Let $R_N(\epsilon)$ denote the Rashomon set of models for error tolerance $\epsilon$ defined over data records $\langle d_1, \ldots, d_N\rangle$, and assume that $\epsilon$ is less than half of the average weight, i.e., $\int_{0}^{1} w \, f(w) \, dw > 2\epsilon$.  Then 
\[
\lim_{N\rightarrow\infty} \frac{\log |R_N(\epsilon)|}{N} = \log B(\epsilon),
\]
where $B(\epsilon) = \exp\left(\int_0^\epsilon C(x) dx\right)$, $C(\epsilon) = g^{-1}(\epsilon)$, and $g(C) = \int_0^1 \frac{w f(w)}{1+\exp(Cw)} \, dw$.
\end{theorem}
\begin{proof}
The statement follows immediately from Lemma~\ref{lemma:asymptotic_size}, which gives the size of the Rashomon set in terms of $C$, and Corollary~\ref{corr:C}, which gives the expression for $C$.
\end{proof}
\begin{remark}
To compute the exponential base $B(\epsilon)$, and therefore the Rashomon set size $|R_N(\epsilon)| = B(\epsilon)^N$, given a finite dataset of size $N$, we can calculate the value of $C(\epsilon)$ for a fine grid of $\epsilon$ values by solving the equation $\frac{1}{N} \sum_{i=1\ldots N} \frac{w_i}{1+\exp(Cw_i)} = \epsilon$. We then use numerical integration to estimate $B(\epsilon) = \exp\left(\int_0^\epsilon C(x) dx\right)$. Alternatively, for $N\rightarrow\infty$ with a known distribution of weights, $w_i \sim W$ with pdf $f(w)$, we instead solve the integral $\int_0^1 \frac{w f(w)}{1+\exp(Cw)} \, dw = \epsilon$ to obtain $C(\epsilon)$.
\end{remark}

\begin{corollary}
[Asymptotic size of Rashomon set for uniformly distributed weights]\label{corr:asymptotic_size_uniform}
Let $D_N = \langle d_1,d_2,\ldots,d_N \rangle$ be data records drawn i.i.d. from distribution $D$ with corresponding Bayes-optimal probabilities $P_N = \langle p_1,p_2,\ldots,p_N \rangle$ and corresponding weights $W_N = \langle w_1,w_2,\ldots,w_N \rangle$, where $w_i = |2p_i-1|$.  Assume $w_i \sim \text{Uniform}[0,1]$.  Let $R_N(\epsilon)$ denote the Rashomon set of models for error tolerance $\epsilon$ defined over data records $\langle d_1, \ldots, d_N\rangle$, and assume that $\epsilon$ is less than half of the average weight, i.e., $\int_{0}^{1} w \, f(w) \, dw > 2\epsilon$.  Then 
\[
\lim_{N\rightarrow\infty} \frac{\log |R_N(\epsilon)|}{N} = \log B(\epsilon),
\]
where $B(\epsilon) = \exp\left(\int_0^\epsilon C(x) dx\right)$, $C(\epsilon) = g^{-1}(\epsilon)$, and $g(C) = \int_0^1 \frac{w}{1+\exp(Cw)} \, dw = \frac{12\, \text{Li}_2(-e^{-C}) - 12 C \log(e^{-C}+1) + \pi^2}{12 C^2}$.

Moreover, $B(\epsilon) <  \exp\left(\pi\sqrt{\frac{\epsilon}{3}}\right)$,
and $B(\epsilon) \approx \exp\left(\pi\sqrt{\frac{\epsilon}{3}}\right)$ for small $\epsilon$.
\end{corollary}
\begin{proof}
The statement follows from Lemma~\ref{lemma:asymptotic_size}, which gives the size of the Rashomon set in terms of $C$, and Corollary~\ref{corr:C_uniform}, which gives the exact and approximate (upper bound) expressions for $C$ for uniform weights.  Since $C(\epsilon) < \frac{\pi}{\sqrt{12 \epsilon}}$, we know that 
$B(\epsilon) = \exp\left(\int_0^\epsilon C(x) dx\right)
< \exp\left( \int_0^\epsilon \frac{\pi}{\sqrt{12 x}} dx
\right) = \exp\left(\pi\sqrt{\frac{\epsilon}{3}}\right)$.
And since $C(\epsilon) \approx \frac{\pi}{\sqrt{12 \epsilon}}$ for $\epsilon \rightarrow 0$, we know that $B(\epsilon) \approx \exp\left(\pi\sqrt{\frac{\epsilon}{3}}\right)$ for $\epsilon\rightarrow 0$.
\end{proof}

\section{Description of benchmark datasets}
\label{appendix:datasets}

Throughout this paper, we present experimental results on three real-world datasets that are commonly used as benchmarks in the fair machine learning literature: German Credit (``German''), Adult, and Heritage Health (``Health'').

We use a preprocessed version of the German Credit data~\cite{german1994} publicly available on Kaggle~\cite{german2017}, which includes credit risk as an outcome variable. Numerical attributes in the dataset, namely $Age$, $Job$, $Credit~Amount$, and $Duration$, were discretized to a categorical attribute as follows: $Age$ was discretized based on whether $Age \geq 25$ or otherwise; $Job$ was discretized based on the number of jobs (no job, one job, or more than one job); $Credit~Amount$, whose values range from $250$ to $18,400$, was discretized into five bins; and $Duration$, whose values range from $4$ to $72$, was discretized into five bins. Categorical and binary attributes were unchanged. Finally, attributes were 1-hot encoded.

We use a publicly available version of the Adult data~\cite{adult1996}, which includes income as an outcome variable. Numerical attributes in the dataset, namely $age$, $fnlwgt$ (final weight), \emph{education-num} (education level), \emph{capital-gain}, \emph{capital-loss}, and \emph{hours-per-week}, were binarized using their median value as the threshold. Categorical and binary attributes were unchanged. Finally, attributes were 1-hot encoded.

We use a publicly available version of the Heritage Health data~\cite{health2012}. We use similar features as the winning team, Market Makers ~\cite{marketmarkers}. We generate the features using the SQL script in the Appendix of~\cite{marketmarkers}, which generates the majority of the variables in data set 1. Since the $ageMISS$ feature corresponds to whether the age value is missing or not, all rows with $ageMISS=1$ were removed, and the $ageMISS$ feature was dropped. Additionally, the sensitive feature $S$ was created with $S=0$ when the age is between 0 and 59 ($age\_{05}=1$ or $age\_{15}=1$ or $age\_{25}=1$ or $age\_{35}=1$ or $age\_{45}=1$ or $age\_{55}=1$), and $S=1$ otherwise. Numerical attributes were binarized using their median value as the threshold. Categorical and binary attributes were unchanged. Finally, attributes were 1-hot encoded.

In German Credit ($N=1,000$), there are 690 men (labeled $gender = 1$ in the dataset) and 310 women ($gender = 0$).  The outcome variable (high risk) is whether an individual is considered high-risk for a loan. Women ($gender = 0$) are the minority class (31.0\% of the dataset) and are disadvantaged (35.2\% likely to be considered high risk for a loan, vs. 27.7\% for men).

In Adult ($N=46,443$), there are 15,203 women (labeled $sex=1$ in the dataset) and 31,240 men ($sex=0$).  The outcome variable ($income$) is whether a person has income over \$50,000. Women ($sex = 1$) are the minority class (32.7\% of the dataset) and are disadvantaged (11.2\% likely to be predicted high income vs. 30.9\% for men).

In Health ($N=184,308$), there are 73,535 individuals over the age of 60 (labeled $S = 1$ in the dataset) and 110,773 other individuals ($S=0$).  The outcome variable ($DaysInHospital$) represents whether a person will spend any days in the hopsital that year. Individuals over the age of 60 ($S= 1$) are the minority class (39.9\% of the dataset) and are disadvantaged ($DaysInHospital=1$ 19.6\% of the time, vs. 10.6\% for others).

As noted in Section~\ref{sec:intentional-experiments} above, we estimate the the Bayes-optimal probabilities $p_i$ for all three datasets using 5-fold cross-validation, using two different approaches, logistic regression (main paper) and XGBoost (Appendix~\ref{appendix:robustness}). 
We report the cross-validated accuracy scores for each dataset using the approximate Bayes-optimal predictions $f_\text{opt}(x_i) = \mathbf{1}\{\hat p_i > 0.5\}$ and observed outcomes $y_i$.  For logistic regression, accuracy was 73.7\%, 84.4\%, and 86.2\% for German, Adult, and Health datasets respectively. For XGBoost, accuracy was 69.4\%, 84.3\%, and 88.8\% for German, Adult, and Health datasets respectively.

\section{Additional results for Section~\ref{sec:intentional-experiments} (Intentional Fairness)}~\label{appendix:intentional-experiments}

\begin{figure}[t]
    \centering
    % First subfigure
    \begin{subfigure}[b]{0.32\textwidth}  % Adjust width to your needs
        \centering
        \includegraphics[width=\textwidth]{   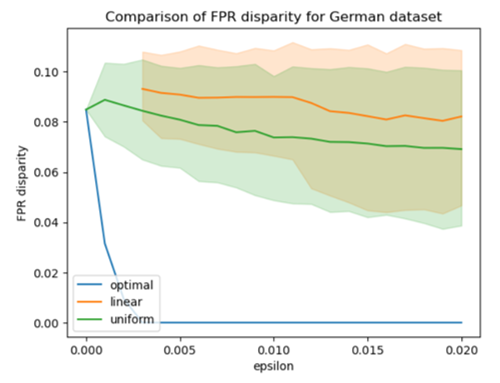}  % Replace with your image
    \end{subfigure}
    \hfill
    % Second subfigure
    \begin{subfigure}[b]{0.32\textwidth}
        \centering
        \includegraphics[width=\textwidth]{   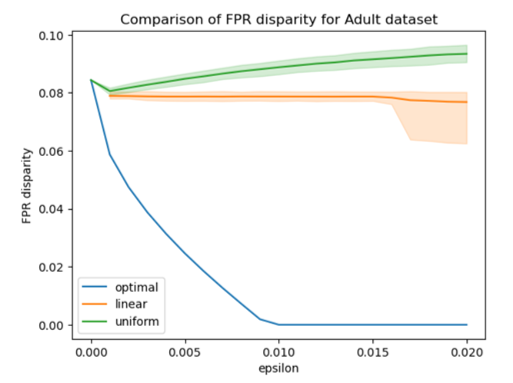}  % Replace with your image
    \end{subfigure}
     % Second subfigure
    \begin{subfigure}[b]{0.32\textwidth}
        \centering
        \includegraphics[width=\textwidth]{   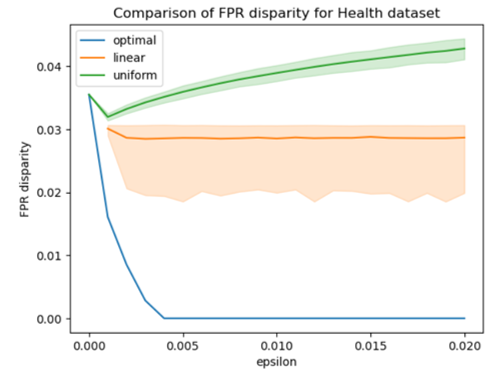}  % Replace with your image
    \end{subfigure}
    \caption{Disparity in false positive rate for the German, Adult, and Health datasets, as a function of the error tolerance $\epsilon$. Comparison of methods for optimizing FPR (Section~\ref{sec:optimizing-TPR and FPR}), uniform random sampling (Section~\ref{sec:sampling}), and sampling linear models (Section~\ref{sec:linear}) over the Rashomon set $R_N(\epsilon)$.}
    \label{fig:FPR_disparity}
\end{figure}
\begin{figure}[t]
    \centering
    % First subfigure
    \begin{subfigure}[b]{0.32\textwidth}  % Adjust width to your needs
        \centering
        \includegraphics[width=\textwidth]{   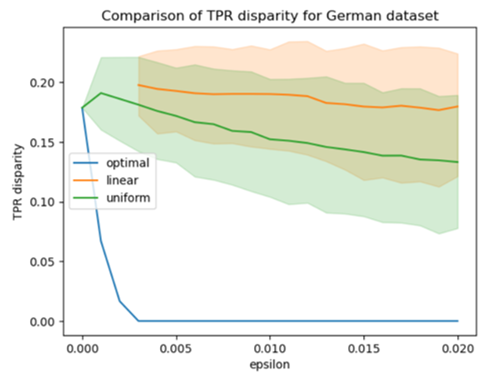}  % Replace with your image
    \end{subfigure}
    \hfill
    % Second subfigure
    \begin{subfigure}[b]{0.32\textwidth}
        \centering
        \includegraphics[width=\textwidth]{   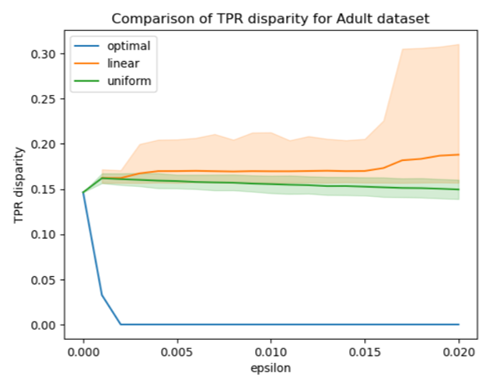}  % Replace with your image
    \end{subfigure}
     % Second subfigure
    \begin{subfigure}[b]{0.32\textwidth}
        \centering
        \includegraphics[width=\textwidth]{   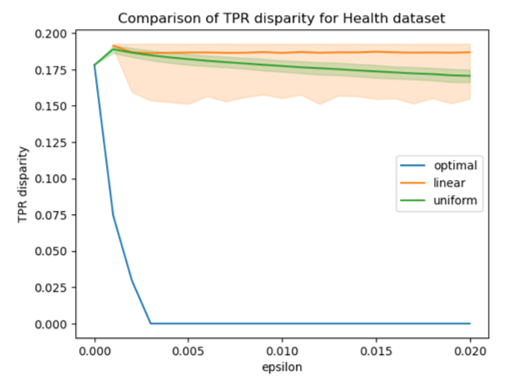}  % Replace with your image
    \end{subfigure}
    \caption{Disparity in true positive rate for the German, Adult, and Health datasets, as a function of the error tolerance $\epsilon$. Comparison of methods for optimizing TPR (Section~\ref{sec:optimizing-TPR and FPR}), uniform random sampling (Section~\ref{sec:sampling}), and sampling linear models (Section~\ref{sec:linear}) over the Rashomon set $R_N(\epsilon)$.}
    \label{fig:TPR_disparity}
\end{figure}

Figures~\ref{fig:FPR_disparity} and~\ref{fig:TPR_disparity} show the disparity in false positive rate (FPR) and true positive rate (TPR) respectively, as a function of the error tolerance $\epsilon$.  These figures compare the methods for optimizing FPR and TPR disparities over the Rashomon set $R_N(\epsilon)$ (Section~\ref{sec:optimizing-fairness}) to uniform random sampling (Section~\ref{sec:sampling}) and sampling linear models (Section~\ref{sec:linear}).   We see that both sets of results are very similar to the results presented for PPR disparity in Figure~\ref{fig:PPR_disparity}.

\section{Additional results for Section~\ref{sec:flip_probs-individual} and~\ref{sec:flip_probs-experiments} (Flip Probabilities)}
\label{appendix:flip_probs}

As noted in Section~\ref{sec:flip_probs-individual}, we can \emph{exactly} (in the large-sample limit) and \emph{efficiently} compute the \emph{average} over the entire Rashomon set $R_N(\epsilon)$ of any metric (such as accuracy, PPR disparity, FPR disparity, or TPR disparity) which can be decomposed as a linear function, $h_i^0 + h_i^1 f(x_i)$, of the individual predictions $f(x_i)$ using the flip probabilities $q_{N,i}$. To see this, we can write:
\begin{align*}
\frac{1}{|R_N(\epsilon)|}
\sum_{\theta \in R_N(\epsilon)}
\sum_{i=1\ldots N} (h_i^0 + h_i^1 f(x_i))
&= 
\frac{1}{|R_N(\epsilon)|}
\sum_{\theta \in R_N(\epsilon)}
\sum_{i=1\ldots N} (h_i^0 + h_i^1 (\theta_i \mathbf{1}\{p_i \le 0.5\} + (1-\theta_i) \mathbf{1}\{p_i > 0.5\})) \\
&=
\sum_{i=1\ldots N} h_i^0 +
\sum_{i=1\ldots N} h_i^1
\frac{1}{|R_N(\epsilon)|}
\sum_{\theta \in R_N(\epsilon)}
 (\theta_i \mathbf{1}\{p_i \le 0.5\} + (1-\theta_i) \mathbf{1}\{p_i > 0.5\}) \\
 &=
\sum_{i=1\ldots N} h_i^0 +
\sum_{i=1\ldots N} h_i^1
\left(
\mathbf{1}\{p_i > 0.5\} +
(\mathbf{1}\{p_i \le 0.5\} - \mathbf{1}\{p_i > 0.5\})
\frac{\sum_{\theta \in R_N(\epsilon)}
\theta_i}{|R_N(\epsilon)|}
 \right)
 \\
 &=
\sum_{i=1\ldots N} h_i^0 +
\sum_{i=1\ldots N} h_i^1
\left(
\mathbf{1}\{p_i > 0.5\} +
(\mathbf{1}\{p_i \le 0.5\} - \mathbf{1}\{p_i > 0.5\}) q_{N,i}
 \right).
\end{align*}

Concretely, for accuracy we have $h_i^0 = \frac{1-p_i}{N}$ and $h_i^1 = \frac{2p_i-1}{N}$.  For PPR disparity, assuming wlog that subgroup $A$ has greater PPR, we have 
$h_i^0 = 0$ and $h_i^1 = \frac{\mathbf{1}\{d_i \in A\}}{|A|} - \frac{\mathbf{1}\{d_i \in B\}}{|B|}.$  For FPR disparity, assuming wlog that subgroup $A$ has greater FPR, we have 
$h_i^0 = 0$ and $h_i^1 = \frac{(1-p_i)\mathbf{1}\{d_i \in A\}}{||1-P_A||_1} - \frac{(1-p_i)\mathbf{1}\{d_i \in B\}}{||1-P_B||_1}.$  Finally, For TPR disparity, assuming wlog that subgroup $A$ has greater TPR, we have  $h_i^0 = 0$ and $h_i^1 = \frac{p_i\mathbf{1}\{d_i \in A\}}{||P_A||_1} - \frac{p_i\mathbf{1}\{d_i \in B\}}{||P_B||_1}.$

Second, as noted in Section~\ref{sec:flip_probs-individual}, comparing the amount of arbitrariness (as defined by the average flip probability) across demographic groups provides a very different notion of group fairness compared to typical definitions including statistical parity and error rate balance.  As a simple proof-of-concept example, imagine that we have two equally-sized subgroups $A$ and $B$ with Bayes-optimal probabilities $p_i = 0.6$ for all members of group $A$, while group $B$ is evenly split between $p_i = 0.51$ and $p_i = 0.99$. The Bayes-optimal classifier would predict everyone as positive, leading to PPR = FPR = TPR = 1 for both groups and no observed disparities.  Yet the average flip probability for models in the Rashomon set $R_N(\epsilon)$ would be greater for one group than another depending on the value of $\epsilon$. For a small value of $\epsilon=.001$, group $B$ would be 32\% more likely to be flipped than group $A$, while for a larger value of $\epsilon = .02$, group $B$ would be 14\% \emph{less} likely to be flipped.

We now present three figures discussed in Section~\ref{sec:flip_probs-experiments}.
In Figure~\ref{fig:appendix_overall_flip_probs}, we graph the overall (population average) flip probability for all three datasets for models sampled uniformly at random from the Rashomon set $R_N(\epsilon)$ as a function of $\epsilon$, compared to sampling linear models from the Rashomon set (as described in Section~\ref{sec:sampling}) and the models that optimize PPR, FPR, and TPR over the Rashomon set (as described in Section~\ref{sec:optimizing-fairness}).  

\begin{figure}[t]
    \centering
    % First subfigure
    \begin{subfigure}[b]{0.32\textwidth}  % Adjust width to your needs
        \centering
        \includegraphics[width=\textwidth]{   figures/flip_prob_german_1-13-2025.png}  % Replace with your image
    \end{subfigure}
    \hfill
    % Second subfigure
    \begin{subfigure}[b]{0.32\textwidth}
        \centering
        \includegraphics[width=\textwidth]{   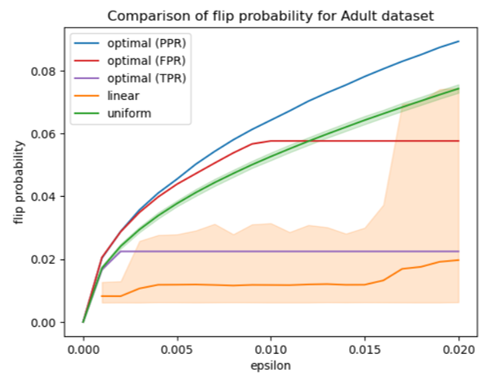}  % Replace with your image
    \end{subfigure}
     % Second subfigure
    \begin{subfigure}[b]{0.32\textwidth}
        \centering
        \includegraphics[width=\textwidth]{   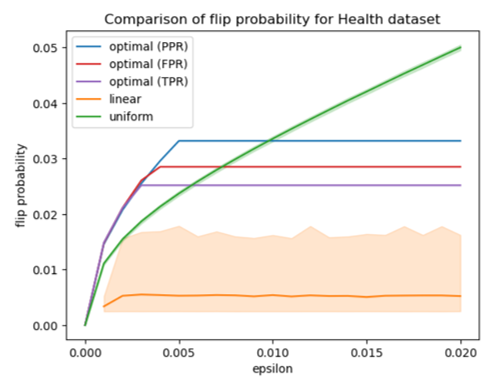}  % Replace with your image
    \end{subfigure}
    \caption{Overall (population average) flip probability for the German, Adult, and Health datasets, as a function of the error tolerance $\epsilon$. Comparison of methods for optimizing PPR (Section~\ref{sec:optimizing-PPR}), optimizing FPR (Section~\ref{sec:optimizing-TPR and FPR}), optimizing TPR (Section~\ref{sec:optimizing-TPR and FPR}),
    uniform random sampling (Section~\ref{sec:sampling}), and sampling linear models (Section~\ref{sec:linear}) over the Rashomon set $R_N(\epsilon)$.}\label{fig:appendix_overall_flip_probs}
\end{figure}

In Figure~\ref{fig:flip_vs_sampling}, we use the flip probabilities to compute the average PPR, FPR, and TPR disparities of the entire Rashomon set $R_N(\epsilon)$ as a function of the error tolerance $\epsilon$ for the German, Adult, and Health datasets.  While these quantities can also be approximated by sampling a large number of flip vectors uniformly at random from the Rashomon set and computing their sample averages, as described in Section~\ref{sec:sampling}, using the flip probabilities is both exact and much more computationally efficient.  We see that the 
sample averages (orange curves) and entire-Rashomon-set averages (blue curves) match closely in Figure~\ref{fig:flip_vs_sampling}, but the orange curves include a small amount of random noise while the blue curves are smooth. 

\begin{figure}[t]
  \centering 
  \includegraphics[width=.95\linewidth]{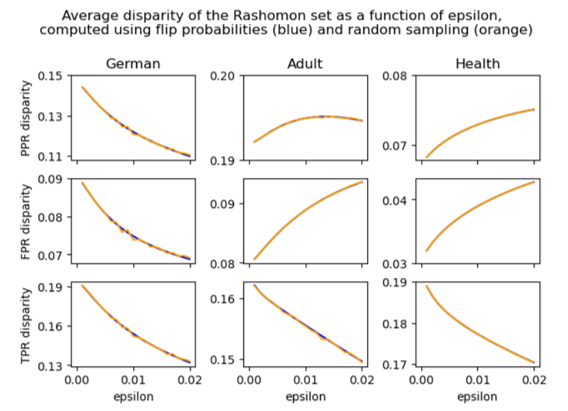}
    \caption{Comparison of calculated PPR, FPR, and TPR disparities as a function of $\epsilon$ for the German, Adult, and Health datasets. Blue curves: average disparity of the entire Rashomon set calculated using the flip probabilities, as described in Appendix~\ref{appendix:flip_probs}.  Orange curves: average disparity of 950 flip vectors sampled uniformly at random from the Rashomon set, as described in Section~\ref{sec:sampling}}.
  \label{fig:flip_vs_sampling}
\end{figure}

In Figure~\ref{fig:appendix_stratified_flip_probs}, we compute the group average flip probabilities for the protected and non-protected groups
as a function of the error tolerance $\epsilon$ for the German, Adult, and Health datasets.

\begin{figure}[htp]
    \centering
    % First subfigure
    \begin{subfigure}[b]{0.32\textwidth}  % Adjust width to your needs
        \centering
        \includegraphics[width=\textwidth]{   figures/stratified_flip_probability_german.png}  % Replace with your image
    \end{subfigure}
    \hfill
    % Second subfigure
    \begin{subfigure}[b]{0.32\textwidth}
        \centering
        \includegraphics[width=\textwidth]{   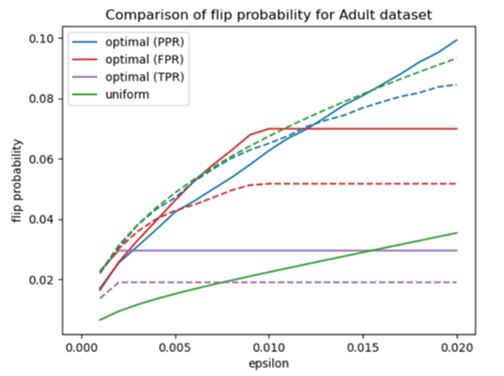}  % Replace with your image
    \end{subfigure}
     % Second subfigure
    \begin{subfigure}[b]{0.32\textwidth}
        \centering
        \includegraphics[width=\textwidth]{   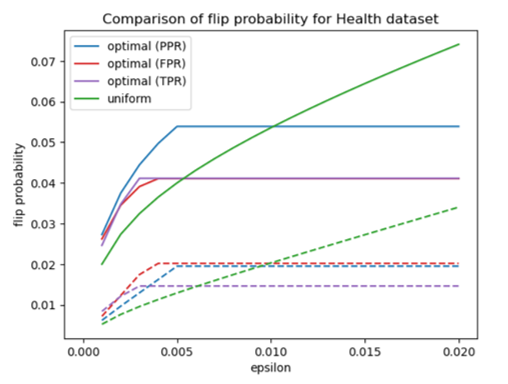}  % Replace with your image
    \end{subfigure}
    \caption{Group average flip probability, comparison between protected group (solid lines) and non-protected group (dashed lines), for the German, Adult, and Health datasets, as a function of the error tolerance $\epsilon$. Comparison of methods for optimizing PPR (Section~\ref{sec:optimizing-PPR}), optimizing FPR (Section~\ref{sec:optimizing-TPR and FPR}), optimizing TPR (Section~\ref{sec:optimizing-TPR and FPR}),
    and uniform random sampling (Section~\ref{sec:sampling}), over the Rashomon set $R_N(\epsilon)$.}\label{fig:appendix_stratified_flip_probs}
\end{figure}

\section{Additional results for Section~\ref{sec:rashomon_set_size-experiments} (Rashomon set size experiments)}
\label{appendix:rashomon_set_size-experiments}

\begin{figure}[htp]
    \centering
    % First subfigure
    \begin{subfigure}[b]{0.32\textwidth}  % Adjust width to your needs
        \centering
        \includegraphics[width=\textwidth]{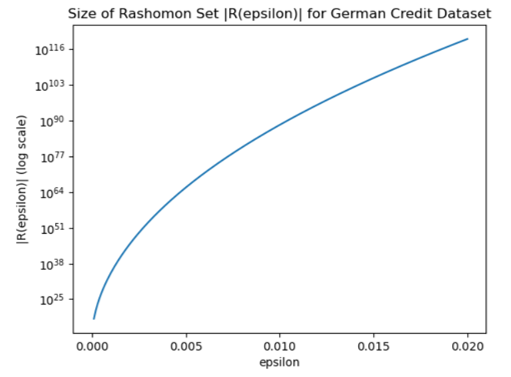}  % Replace with your image
    \end{subfigure}
    \hfill
    % Second subfigure
    \begin{subfigure}[b]{0.32\textwidth}
        \centering
        \includegraphics[width=\textwidth]{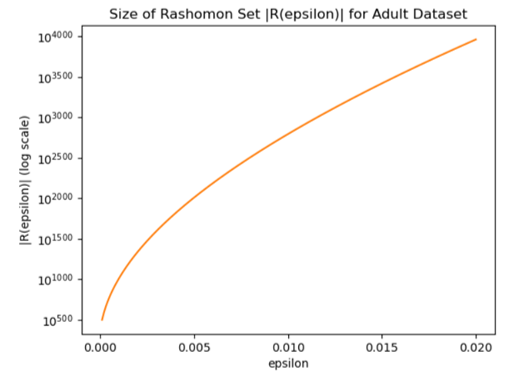}  % Replace with your image
    \end{subfigure}
     % Second subfigure
    \begin{subfigure}[b]{0.32\textwidth}
        \centering
        \includegraphics[width=\textwidth]{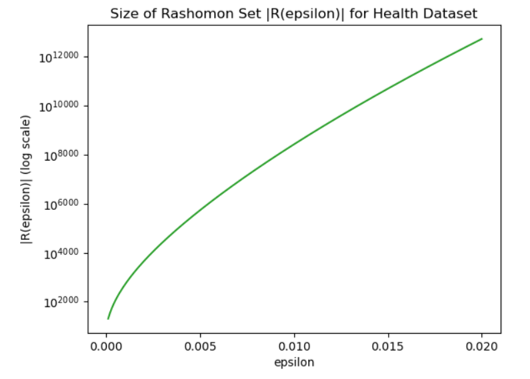}  % Replace with your image
    \end{subfigure}
    \caption{Rashomon set size $|R_N(\epsilon)|$ for the German Credit, Adult, and Health datasets. Note the logarithmic scale of the $y$-axis.}
    \label{fig:appendix_Rashomon_set_size}
\end{figure}
\begin{figure}[htp]
    \centering
    % First subfigure
    \begin{subfigure}[b]{0.32\textwidth}  % Adjust width to your needs
        \centering
        \includegraphics[width=\textwidth]{   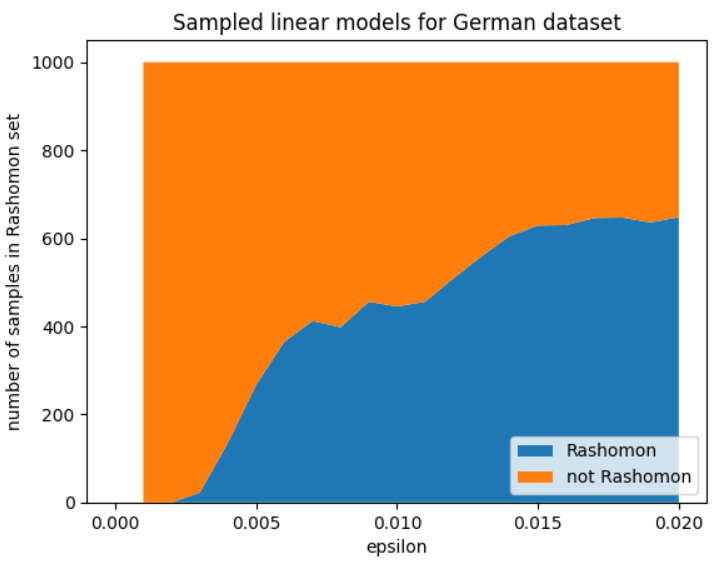}  % Replace with your image
    \end{subfigure}
    \hfill
    % Second subfigure
    \begin{subfigure}[b]{0.32\textwidth}
        \centering
        \includegraphics[width=\textwidth]{   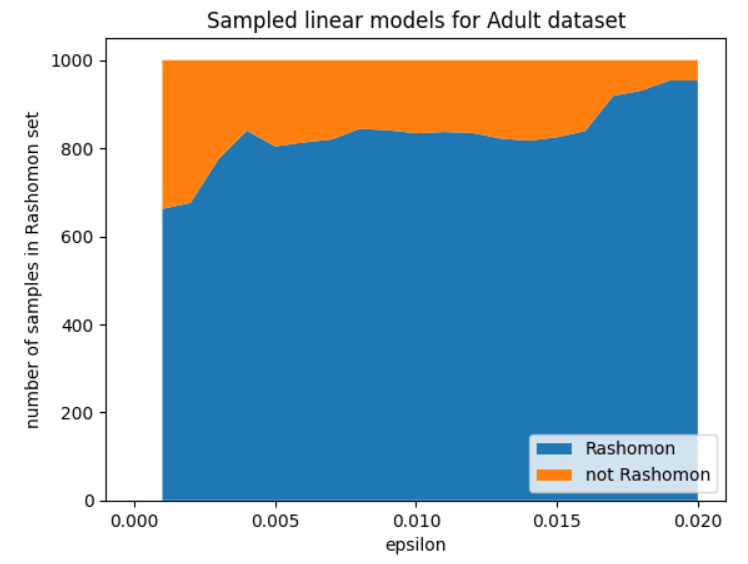}  % Replace with your image
    \end{subfigure}
     % Second subfigure
    \begin{subfigure}[b]{0.32\textwidth}
        \centering
        \includegraphics[width=\textwidth]{   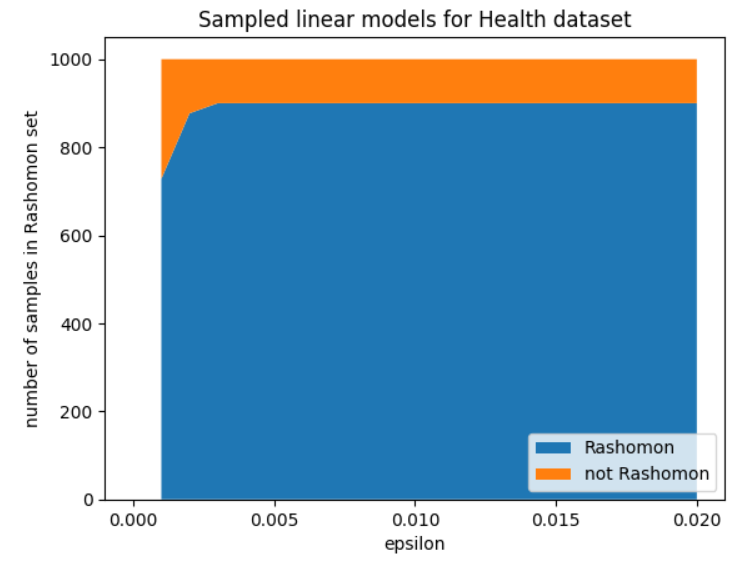}  % Replace with your image
    \end{subfigure}
    \caption{Proportion of randomly sampled linear models that are in the Rashomon set $R_N(\epsilon)$ as a function of the error tolerance $\epsilon$, for the German Credit (left), Adult (center), and Health (right) datasets. The $y$-axis represents how many of the 1000 randomly sampled linear models are (blue) and are not (orange) in the Rashomon set.}
\label{fig:appendix_linear_proportion_in_Rashomon_set}
\end{figure}

We plot the Rashomon set size $|R_N(\epsilon)|$ for the German Credit, Adult, and Health datasets in Figure~\ref{fig:appendix_Rashomon_set_size}.  We observe that the Rashomon set sizes are very large and scale rapidly with $\epsilon$, since $|R_N(\epsilon)| = B(\epsilon)^N$. For the maximum $\epsilon$ value we consider, $\epsilon = 0.02$, we have $B=1.32$ for German Credit, $B=1.22$ for Adult, and $B=1.17$ for Health.

Additionally, while we do not yet have a way of computing the (reduced) Rashomon set size when restricting our search to the space of linear ($L_2$-penalized logistic regression) models as described in Section~\ref{sec:linear} above, we can nevertheless examine what fraction of the sampled linear models are in the Rashomon set as a function of $\epsilon$.  This is shown (for $\epsilon \in \{0.001, 0.002, \ldots, 0.02\}$) for the German Credit, Adult, and Health datasets in Figure~\ref{fig:appendix_linear_proportion_in_Rashomon_set}.  We see that, for the Adult and Health datasets, most of the randomly sampled linear models are in the Rashomon set, even for low $\epsilon$-values.  For the German Credit data, a substantial fraction of linear models are not in the Rashomon set, even when $\epsilon$ is large.

\section{Robustness check: use of an alternate model to estimate Bayes-optimal probabilities}
\label{appendix:robustness}

\begin{figure}[t]
    \centering
    % First subfigure
    \begin{subfigure}[b]{0.32\textwidth}  % Adjust width to your needs
        \centering
        \includegraphics[width=\textwidth]{   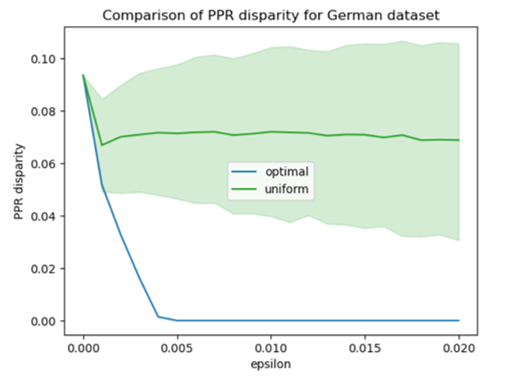}  % Replace with your image
    \end{subfigure}
    \hfill
    % Second subfigure
    \begin{subfigure}[b]{0.32\textwidth}
        \centering
        \includegraphics[width=\textwidth]{   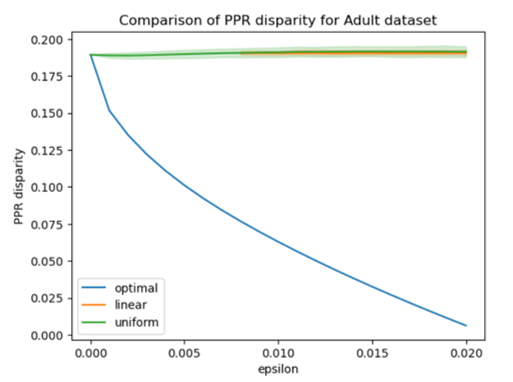}  % Replace with your image
    \end{subfigure}
     % Second subfigure
    \begin{subfigure}[b]{0.32\textwidth}
        \centering
        \includegraphics[width=\textwidth]{   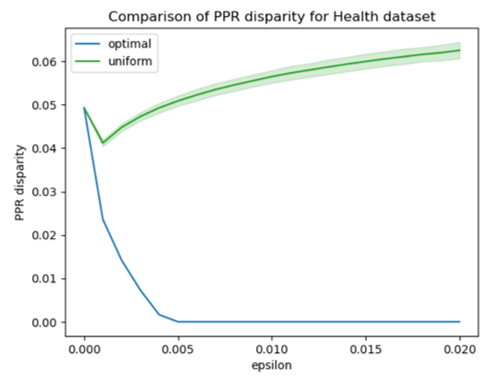}  % Replace with your image
    \end{subfigure}
    \caption{Robustness check using XGBoost instead of logistic regression to estimate Bayes-optimal probabilities. Disparity in positive prediction rate for the German, Adult, and Health datasets, as a function of the error tolerance $\epsilon$. Comparison of methods for optimizing PPR (Section~\ref{sec:optimizing-PPR}), uniform random sampling (Section~\ref{sec:sampling}), and sampling linear models (Section~\ref{sec:linear}) over the Rashomon set $R_N(\epsilon)$. Note that no linear models were in the Rashomon set for German and Health datasets.}
    \label{fig:PPR_disparity_xgb}
\end{figure}
\begin{figure}[t]
    \centering
    % First subfigure
    \begin{subfigure}[b]{0.32\textwidth}  % Adjust width to your needs
        \centering
        \includegraphics[width=\textwidth]{   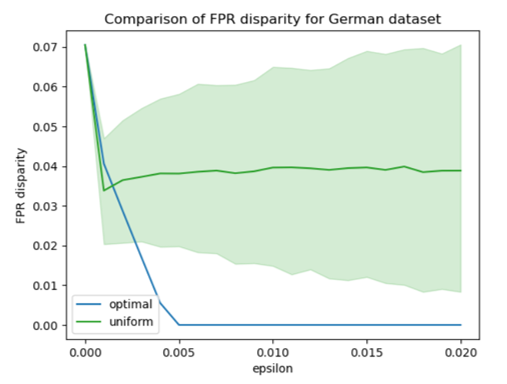}  % Replace with your image
    \end{subfigure}
    \hfill
    % Second subfigure
    \begin{subfigure}[b]{0.32\textwidth}
        \centering
        \includegraphics[width=\textwidth]{   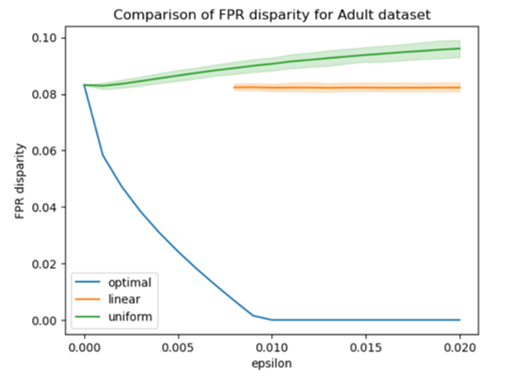}  % Replace with your image
    \end{subfigure}
     % Second subfigure
    \begin{subfigure}[b]{0.32\textwidth}
        \centering
        \includegraphics[width=\textwidth]{   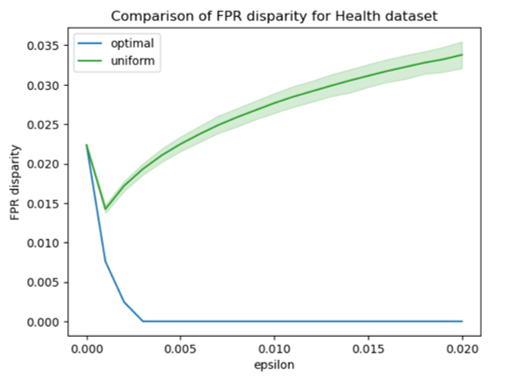}  % Replace with your image
    \end{subfigure}
    \caption{Robustness check using XGBoost instead of logistic regression to estimate Bayes-optimal probabilities. Disparity in false positive rate for the German, Adult, and Health datasets, as a function of the error tolerance $\epsilon$. Comparison of methods for optimizing FPR (Section~\ref{sec:optimizing-TPR and FPR}), uniform random sampling (Section~\ref{sec:sampling}), and sampling linear models (Section~\ref{sec:linear}) over the Rashomon set $R_N(\epsilon)$. Note that no linear models were in the Rashomon set for German and Health datasets.}
    \label{fig:FPR_disparity_xgb}
\end{figure}
\begin{figure}[t]
    \centering
    % First subfigure
    \begin{subfigure}[b]{0.32\textwidth}  % Adjust width to your needs
        \centering
        \includegraphics[width=\textwidth]{   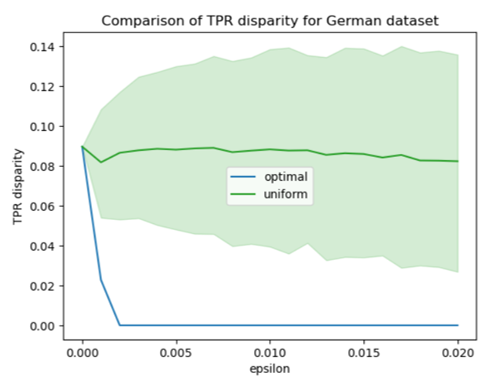}  % Replace with your image
    \end{subfigure}
    \hfill
    % Second subfigure
    \begin{subfigure}[b]{0.32\textwidth}
        \centering
        \includegraphics[width=\textwidth]{   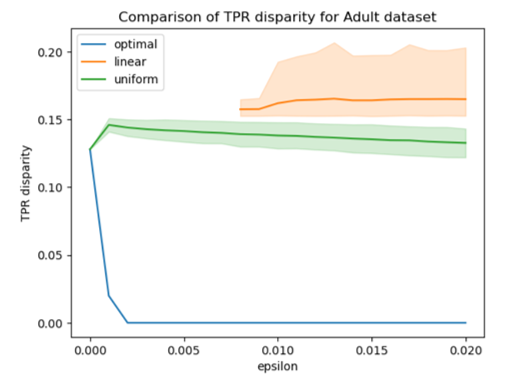}  % Replace with your image
    \end{subfigure}
     % Second subfigure
    \begin{subfigure}[b]{0.32\textwidth}
        \centering
        \includegraphics[width=\textwidth]{   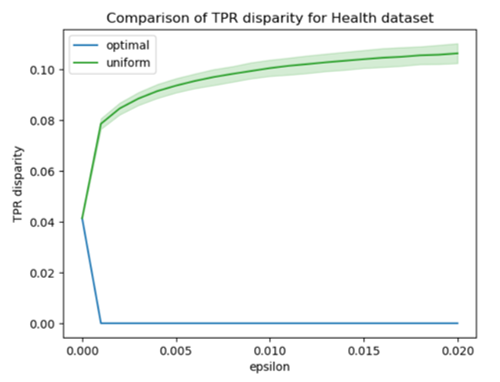}  % Replace with your image
    \end{subfigure}
    \caption{Robustness check using XGBoost instead of logistic regression to estimate Bayes-optimal probabilities. Disparity in true positive rate for the German, Adult, and Health datasets, as a function of the error tolerance $\epsilon$. Comparison of methods for optimizing TPR (Section~\ref{sec:optimizing-TPR and FPR}), uniform random sampling (Section~\ref{sec:sampling}), and sampling linear models (Section~\ref{sec:linear}) over the Rashomon set $R_N(\epsilon)$. Note that no linear models were in the Rashomon set for German and Health datasets.}
    \label{fig:TPR_disparity_xgb}
\end{figure}
\begin{figure}[t]
    \centering
    % First subfigure
    \begin{subfigure}[b]{0.32\textwidth}  % Adjust width to your needs
        \centering
        \includegraphics[width=\textwidth]{   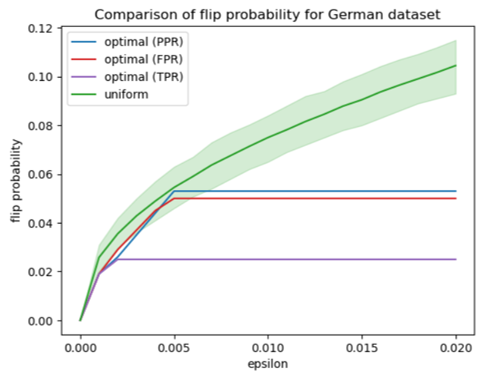}  % Replace with your image
    \end{subfigure}
    \hfill
    % Second subfigure
    \begin{subfigure}[b]{0.32\textwidth}
        \centering
        \includegraphics[width=\textwidth]{   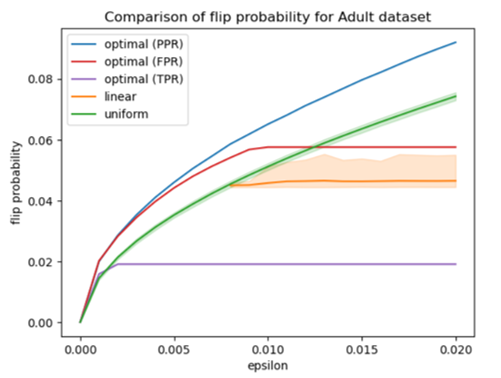}  % Replace with your image
    \end{subfigure}
     % Second subfigure
    \begin{subfigure}[b]{0.32\textwidth}
        \centering
        \includegraphics[width=\textwidth]{   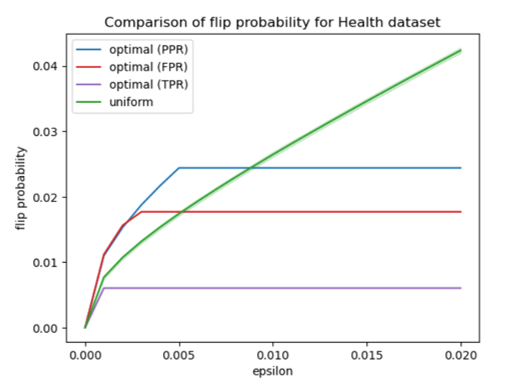}  % Replace with your image
    \end{subfigure}
    \caption{Robustness check using XGBoost instead of logistic regression to estimate Bayes-optimal probabilities. Overall (population average) flip probability for the German, Adult, and Health datasets, as a function of the error tolerance $\epsilon$. Comparison of methods for optimizing PPR (Section~\ref{sec:optimizing-PPR}), optimizing FPR (Section~\ref{sec:optimizing-TPR and FPR}), optimizing TPR (Section~\ref{sec:optimizing-TPR and FPR}),
    uniform random sampling (Section~\ref{sec:sampling}), and sampling linear models (Section~\ref{sec:linear}) over the Rashomon set $R_N(\epsilon)$. Note that no linear models were in the Rashomon set for German and Health datasets.}\label{fig:appendix_overall_flip_probs_xgb}
\end{figure}
\begin{figure}[t]
    \centering
    % First subfigure
    \begin{subfigure}[b]{0.32\textwidth}  % Adjust width to your needs
        \centering
        \includegraphics[width=\textwidth]{   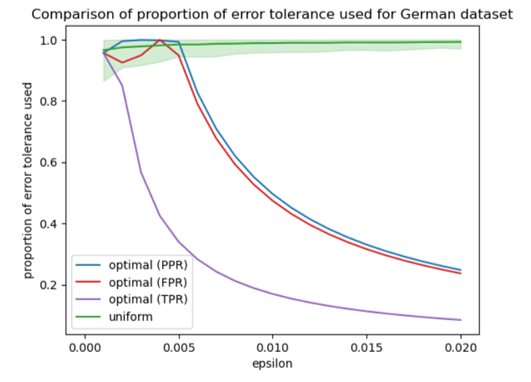}  % Replace with your image
    \end{subfigure}
    \hfill
    % Second subfigure
    \begin{subfigure}[b]{0.32\textwidth}
        \centering
        \includegraphics[width=\textwidth]{   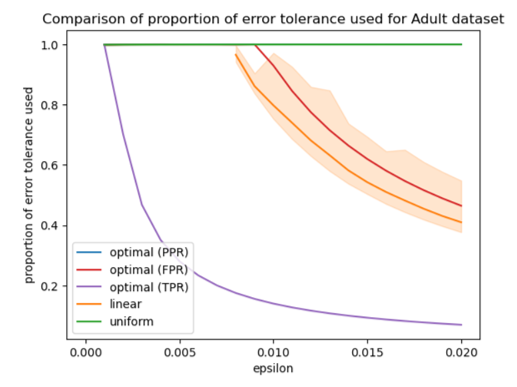}  % Replace with your image
    \end{subfigure}
     % Second subfigure
    \begin{subfigure}[b]{0.32\textwidth}
        \centering
        \includegraphics[width=\textwidth]{   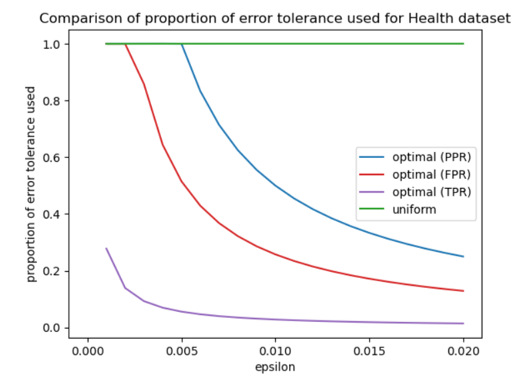}  % Replace with your image
    \end{subfigure}
    \caption{Robustness check using XGBoost instead of logistic regression to estimate Bayes-optimal probabilities. Proportion of error tolerance used, $\frac{\theta \cdot W_N}{N\epsilon}$, for the German, Adult, and Health datasets, as a function of the error tolerance $\epsilon$. Comparison of methods for optimizing PPR (Section~\ref{sec:optimizing-PPR}), optimizing FPR (Section~\ref{sec:optimizing-TPR and FPR}), optimizing TPR (Section~\ref{sec:optimizing-TPR and FPR}), uniform random sampling (Section~\ref{sec:sampling}), and sampling linear models (Section~\ref{sec:linear}) over the Rashomon set $R_N(\epsilon)$. Note that no linear models were in the Rashomon set for German and Health datasets.}
    \label{fig:appendix_error_proportion_xgb}
\end{figure}

As noted above, the Bayes-optimal probabilities $p_i$ are unknown for real-world datasets, but can be well-estimated using sufficient training data.  
In the main paper, we used logistic regression to estimate these probabilities. To check the robustness of our results to the choice of model used for estimation of $p_i$, we re-ran all experiments using the estimated probabilities $\hat p_i$ from XGBoost models learned using 5-fold cross-validation.  Here we present results comparing our methods for optimizing PPR (Section~\ref{sec:optimizing-PPR}), optimizing FPR (Section~\ref{sec:optimizing-TPR and FPR}), optimizing TPR (Section~\ref{sec:optimizing-TPR and FPR}), uniform random sampling (Section~\ref{sec:sampling}), and sampling linear models (Section~\ref{sec:linear}) over the Rashomon set $R_N(\epsilon)$. 
Results for PPR disparity, FPR disparity, TPR disparity, overall flip probability, and proportion of error tolerance used, all using the XGBoost-generated probability estimates, are shown in Figures~\ref{fig:PPR_disparity_xgb}-\ref{fig:appendix_error_proportion_xgb}.  These can be compared to the corresponding results for logistic regression-generated probability estimates for PPR disparity, FPR disparity, TPR disparity, overall flip probability, and proportion of error tolerance used in Figures~\ref{fig:PPR_disparity},~\ref{fig:FPR_disparity},~\ref{fig:TPR_disparity},~\ref{fig:appendix_overall_flip_probs}, and~\ref{fig:error_proportion}(right) respectively.  The primary difference we observe is that none of the randomly sampled linear models were in the Rashomon set for the German and Health datasets. For the Adult dataset, we observed randomly sampled linear models in the Rashomon set for  $\epsilon \ge 0.008$, as compared to $\epsilon \ge 0.001$ for the logistic regression-generated probability estimates. These differences are not surprising given that linear models might not be able to fit the more complex, non-linear relationships modeled by XGBoost. Otherwise, results are very similar to those using the logistic regression-generated probability estimates, supporting our conclusions and policy takeaways above.

\end{document}